\documentclass[%
reprint,
superscriptaddress,
frontmatterverbose, 
preprintnumbers,
longbibliography,
amsmath,amssymb,
aps,
pra,
notitlepage,
]{revtex4-1}

\usepackage[utf8]{inputenc}
\usepackage[english]{babel}
\usepackage[T1]{fontenc}
\usepackage{amsmath}
\makeatletter
\newcommand{\pushright}[1]{\ifmeasuring@#1\else\omit\hfill$\displaystyle#1$\fi\ignorespaces}
\newcommand{\pushleft}[1]{\ifmeasuring@#1\else\omit$\displaystyle#1$\hfill\fi\ignorespaces}
\makeatother

\allowdisplaybreaks

\usepackage{tikz}
\usepackage{lipsum}
\usepackage{amssymb}
\usepackage{color}
\usepackage{mathrsfs}
\usepackage{amsbsy}
\usepackage{amsthm}
\usepackage{balance}
\usepackage{thmtools,thm-restate}
\usepackage{graphicx}
\usepackage{tikz}

\usepackage{comment}

\usepackage{bbm}
\usepackage{bm}
\usepackage{epsfig}
\usepackage{xfrac}
\usepackage{xcolor}
\usepackage{enumerate}
\usepackage[shortlabels]{enumitem}
\usepackage{graphicx}
\usepackage{dcolumn}
\usepackage{bm}
\usepackage{hyperref}
\usepackage{subfigure}
\hypersetup{
colorlinks=true,
linkcolor=blue,
filecolor=blue,
citecolor=blue,  
urlcolor=blue,
}

\usepackage{hyperref}
\hypersetup{
    colorlinks=true,
    linkcolor=blue,
    filecolor=magenta,      
    urlcolor=cyan,
    pdftitle={Overleaf Example},
    pdfpagemode=FullScreen,
    }

\newtheorem{theorem}{Theorem}
\newtheorem{lemma}{Lemma}

\newcommand{\tPi}{\tilde{\Pi}}

\usepackage{braket}
\renewcommand{\v}[1]{\ensuremath{\mathbf{#1}}} 
\newcommand{\abs}[1]{\left| #1 \right|} 
\newcommand{\norm}[1]{\left\| #1 \right\|} 
\newcommand{\trace}{\mathrm{Tr}}
 
\newcommand{\sg}{\mathrm{sg}} 
\newcommand{\prob}{\mathrm{Prob}}

\newcommand{\tl}{{\textsc{l}}}
\newcommand{\teta}{{\tilde \eta}}
\newcommand{\tepsilon}{{\tilde \epsilon}}
\newcommand{\sql}{{\textsc{sql}}}

\newcommand{\mS}{{\mathcal{S}}}
\newcommand{\mC}{{\mathcal{C}}}
\newcommand{\mH}{{\mathcal{H}}}

\newcommand{\mR}{{\mathcal{R}}}
\newcommand{\mP}{{\mathcal{P}}}
\newcommand{\mE}{{\mathcal{E}}}

\newcommand{\id}{{\mathbbm{1}}}
\renewcommand{\Re}{{\mathrm{Re}}}
\renewcommand{\Im}{{\mathrm{Im}}}

\newcommand{\frakh}{{\mathfrak{h}}}

\newcommand{\fraki}{{\mathfrak{i}}}
\newcommand{\frakI}{{\mathfrak{I}}}

\newcommand{\frakL}{{\mathfrak{L}}}
\newcommand{\vh}{{\v{h}}}

\newcommand{\bR}{{\mathbb{R}}}
\newcommand{\bC}{{\mathbb{C}}}
\newcommand{\bE}{{\mathbb{E}}}
\renewcommand{\epsilon}{\varepsilon}
\newcommand{\appropto}{\mathrel{\vcenter{
  \offinterlineskip\halign{\hfil$##$\cr
    \propto\cr\noalign{\kern2pt}\sim\cr\noalign{\kern-2pt}}}}}
\definecolor{darkbyzantium}{rgb}{0.36, 0.22, 0.33}
\definecolor{cadmiumgreen}{rgb}{0.0, 0.42, 0.24}
\newcommand{\sisinew}[1]{\textcolor{black}{#1}}
\newcommand{\sisinewlong}[1]{\color{black}{#1}\color{black}}
\newcommand{\sisi}[1]{\textcolor{black}{#1}}
\newcommand{\sisilong}[1]{\color{black}{#1}\color{black}}

\let\baraccent=\= 
\renewcommand{\=}[1]{\stackrel{#1}{=}} 

\newcommand{\thmref}[1]{\hyperref[#1]{Theorem~\ref{#1}}}
\newcommand{\lemmaref}[1]{\hyperref[#1]{Lemma~\ref{#1}}}
\newcommand{\corollaryref}[1]{\hyperref[#1]{Corollary~\ref{#1}}}
\newcommand{\figref}[1]{\hyperref[#1]{Fig.~\ref{#1}}}
\newcommand{\figaref}[1]{\hyperref[#1]{Fig.~\ref{#1}(a)}}
\newcommand{\figbref}[1]{\hyperref[#1]{Fig.~\ref{#1}(b)}}
\newcommand{\figcref}[1]{\hyperref[#1]{Fig.~\ref{#1}(c)}}
\newcommand{\figdref}[1]{\hyperref[#1]{Fig.~\ref{#1}(d)}}
\newcommand{\figeref}[1]{\hyperref[#1]{Fig.~\ref{#1}(e)}}
\newcommand{\figfref}[1]{\hyperref[#1]{Fig.~\ref{#1}(f)}}
\renewcommand{\eqref}[1]{\hyperref[#1]{Eq.~(\ref{#1})}}
\newcommand{\secref}[1]{\hyperref[#1]{Sec.~\ref{#1}}}
\newcommand{\eqsref}[2]{\hyperref[#1]{Eqs.~(\ref{#1})-(\ref{#2})}}
\newcommand{\appref}[1]{\hyperref[#1]{Appx.~\ref{#1}}}

\begin{document}

\title{Achieving metrological limits using ancilla-free quantum error-correcting codes}

\author{Sisi Zhou}
\thanks{These two authors contributed equally. \\
\href{mailto:sisi.zhou26@gmail.com}{sisi.zhou26@gmail.com} (S.Z.);\\
\href{mailto:argyris.giannisismanes@yale.edu}{argyris.giannisismanes@yale.edu} (A.G.M.).}
\affiliation{Institute for Quantum Information and Matter, California Institute of Technology, Pasadena, CA 91125, USA}
\affiliation{Perimeter Institute for Theoretical Physics, Waterloo, Ontario N2L 2Y5, Canada}

\author{Argyris Giannisis Manes}
\thanks{These two authors contributed equally. \\
\href{mailto:sisi.zhou26@gmail.com}{sisi.zhou26@gmail.com} (S.Z.);\\
\href{mailto:argyris.giannisismanes@yale.edu}{argyris.giannisismanes@yale.edu} (A.G.M.).}
\affiliation{Department of Physics, Yale University, New Haven, Connecticut 06511, USA}

\author{Liang Jiang}\email{liang.jiang@uchicago.edu}
\affiliation{Pritzker School of Molecular Engineering, The University of Chicago, Chicago, IL 60637, USA}

\date{\today}

\begin{abstract}
Quantum error correction (QEC) is theoretically capable of achieving the ultimate estimation limits in noisy quantum metrology. However, existing quantum error-correcting codes designed for noisy quantum metrology generally exploit entanglement between one probe and one noiseless ancilla of the same dimension, and the requirement of noiseless ancillas is one of the major obstacles to implementing the QEC metrological protocol in practice. Here we successfully lift this requirement by explicitly constructing two types of multi-probe quantum error-correcting codes, where the first one utilizes a negligible amount of ancillas and the second one is ancilla-free. Specifically, we consider Hamiltonian estimation under Markovian noise and show that (i)~when the Heisenberg limit (HL) is achievable, our new codes can achieve the HL and its optimal asymptotic coefficient; (ii)~when only the standard quantum limit (SQL) is achievable (even with arbitrary adaptive quantum strategies), the optimal asymptotic coefficient of the SQL is also achievable by our new codes under slight modifications. 
\end{abstract}

\maketitle

\section{Introduction}

Quantum metrology, the science of estimation and measurement of quantum systems, is a central topic of modern quantum technologies~\cite{giovannetti2011advances,degen2017quantum,pezze2018quantum,pirandola2018advances}. Quantum sensors are expected to have wide applications across various fields, such as gravitational detection~\cite{caves1981quantum,yurke19862,ligo2011gravitational,ligo2013enhanced,demkowicz2013fundamental}, atomic clocks~\cite{rosenband2008frequency,appel2009mesoscopic,kaubruegger2021quantum,marciniak2022optimal}, magnetometry~\cite{wineland1992spin,leibfried2004toward,dutt2007quantum,hanson2008coherent}, super-resolution imaging~\cite{tsang2016quantum}, etc. 
With the accelerating growth of experimental methods in manipulating and probing quantum systems, environmental noise becomes a major bottleneck in advancing the sensitivity limit of quantum sensors as they scale up. Different tools, such as quantum error correction (QEC)~\cite{kessler2014quantum,dur2014improved,arrad2014increasing,ozeri2013heisenberg,unden2016quantum,zhou2018achieving,wang2022quantum,ouyang2022quantum}, dynamical decoupling~\cite{viola1999dynamical}, optimization of quantum controls~\cite{liu2017quantum}, error mitigation~\cite{yamamoto2021error}, environment monitoring~\cite{gefen2016parameter,albarelli2019restoring}, etc., have been explored to battle the effect of noise in quantum metrology. 

QEC, a standard tool to reduce noise in quantum computing~\cite{shor1995scheme,gottesman2010introduction,lidar2013quantum}, is also capable of achieving the ultimate estimation limits in noisy quantum metrology~\cite{zhou2021asymptotic,kurdzialek2022using}. In quantum metrology, the Heisenberg limit (HL) states that the optimal scaling of the estimation error is $1/N$ where $N$ is the number of probes (or $1/t$, where $t$ is the probing time)~\cite{giovannetti2006quantum}. However, in presence of noise, the estimation error will usually inevitably drop to the standard quantum limit (SQL) ($\propto 1/\sqrt{N}$ or $1/\sqrt{t}$), as indicated by the classical central limit theorem, in which case quantum enhancement is at most constant-factor~\cite{escher2011general,demkowicz2012elusive}. In the setting of Hamiltonian estimation in open quantum systems described by Lindblad master equations, the HL for estimating the Hamiltonian parameter is achievable if and only if the Hamiltonian is not in the noise subspace (called the ``Hamiltonian-not-in-Lindblad-span'' (HNLS) condition)~\cite{demkowicz2017adaptive,zhou2018achieving}, and only the SQL is achievable when the HNLS condition is violated. In both cases, whether the HL is achievable or not, there exist corresponding QEC strategies that achieve the optimal metrological limits and asymptotic coefficients, which no other adaptive quantum strategies can surpass~\cite{zhou2018achieving,zhou2019optimal,demkowicz2014using,wan2022bounds}. 

However, there are several barriers towards practical implementation of the QEC metrological strategies. 
One unique demand of QEC in quantum metrology, that has no equivalent in quantum communication or computing, is to balance the trade-off between reducing the noise and enhancing the signal, posing fundamental difficulties in designing QEC codes for noisy quantum metrology. 
So far, QEC has only been shown to achieve the metrological limits under the assumption that QEC operations can be applied fast and accurately~\cite{sekatski2017quantum,shettell2021practical,rojkov2022bias}. Furthermore, the corresponding QEC codes are generally hybrid, consisting of a probing system where the signal and noise accumulates over time, and a noiseless ancillary system of the same dimension~\cite{zhou2018achieving}. In particular, the ancillary system can be a costly resource in experiments, composed of either exceedingly stable qubits~\cite{kessler2014quantum,unden2016quantum} or noisy qubits under \sisi{code concatenation}~\cite{zhou2018achieving}.

Growing efforts have been devoted to removing the requirement of noiseless ancillas~\cite{layden2019ancilla,peng2020achieving}. 
Some results were obtained in several practically relevant sensing scenarios.
When the HNLS condition is satisfied and the HL is achievable, for qubit probes where the signal and noise acts individually on each probe, an optimal ancilla-free repetition code can be constructed~\cite{dur2014improved,arrad2014increasing,peng2020achieving}. Additionally, when the signal and the noise commute, e.g., a Pauli-Z signal acting on multiple probes under correlated dephasing noise, optimal ancilla-free codes were also proven to exist~\cite{layden2019ancilla}. In the case where the HNLS condition is violated, weakly spin-squeezed states were known to achieve the optimal SQL for a Pauli-Z signal under local dephasing noise~\cite{ulam2001spin,escher2011general}; while in general ancilla-assisted QEC strategies are necessary. 
\sisi{In this paper, we will address the open question whether there exist ancilla-free QEC codes for Hamiltonian estimation under Markovian noise that achieve the optimal HL (or SQL). We will exploit quantum entanglement among multiple probes to design two classes of optimal ancilla-free QEC codes.} 

\section{Preliminaries: QEC sensing}

Given a quantum state $\rho_\omega$ that contains an unknown parameter $\omega$, the estimation error of $\omega$ (i.e., the standard deviation of an unbiased $\omega$-estimator) is bounded by $\delta\omega \geq 1/\sqrt{N_{\rm expr} F(\rho_\omega)}$ from the quantum Cram\'{e}r--Rao bound~\cite{holevo2011probabilistic,helstrom1976quantum,braunstein1994statistical}. Here $N_{\rm expr}$ is the number of experiments and $F(\rho_\omega)$ is the quantum Fisher information (QFI) of $\rho_\omega$ defined by $F(\rho_\omega) := \trace(\rho_\omega L_{\rm sym}^2)$ where $L_{\rm sym}$ is a Hermitian operator satisfying $\partial_\omega\rho_\omega = \frac{1}{2}(L_{\rm sym} \rho_\omega + \rho_\omega L_{\rm sym})$. The bound is asymptotically attainable as $N_{\rm expr}$ goes to infinity~\cite{lehmann2006theory,kobayashi2011probability}, making the QFI a canonical measure of the power of quantum sensors. 

Consider a quantum state $\rho_\omega$ of $N$ identical probes, where the HL describes the case where $F(\rho_\omega) = \Theta(N^2)$~\cite{1footnote-heisenberg}, which is the optimal scaling allowed by quantum mechanics~\cite{giovannetti2006quantum}. For example, consider $N$ qubit probes where the density operator $\rho$ of each probe evolves according to
\begin{equation}
\label{eq:dephasing}
\frac{d\rho}{dt} = -i\left[\frac{\omega Z}{2},\rho\right] + \frac{\gamma}{2}(Z\rho Z -\rho), 
\end{equation}
where $Z = \ket{0}\bra{0} - \ket{1}\bra{1}$ is the Pauli-Z operator, $\omega Z/2$ is the Hamiltonian, and $\gamma$ is the dephasing noise rate. When the initial state is the Greenberger--Horne--Zeilinger (GHZ) state~\cite{giovannetti2006quantum,2footnote-GHZ}, i.e., $\ket{\psi_{\rm in}} = \frac{1}{\sqrt{2}} (\ket{0}^{\otimes N} + \ket{1}^{\otimes N})$, the QFI of the final state $\rho_\omega(t)$ at the probing time $t$ will be 
\begin{equation}
\label{eq:dephasing-QFI}
F(\rho_\omega(t)) = N^2 t^2 e^{-2 N \gamma t}. 
\end{equation} 
The HL is achieved when $\gamma = 0$. 
However, fixing $\gamma t> 0$, we have $F(\rho_\omega(t))\rightarrow 0$ as $N\rightarrow \infty$, indicating the HL is no longer achievable under noise. In fact, only the SQL $F(\rho_\omega) = \Theta(N)$ is achievable in presence of noise, e.g., by using product states.

\begin{figure}[tb]
\centering
\includegraphics[width=0.4\textwidth]{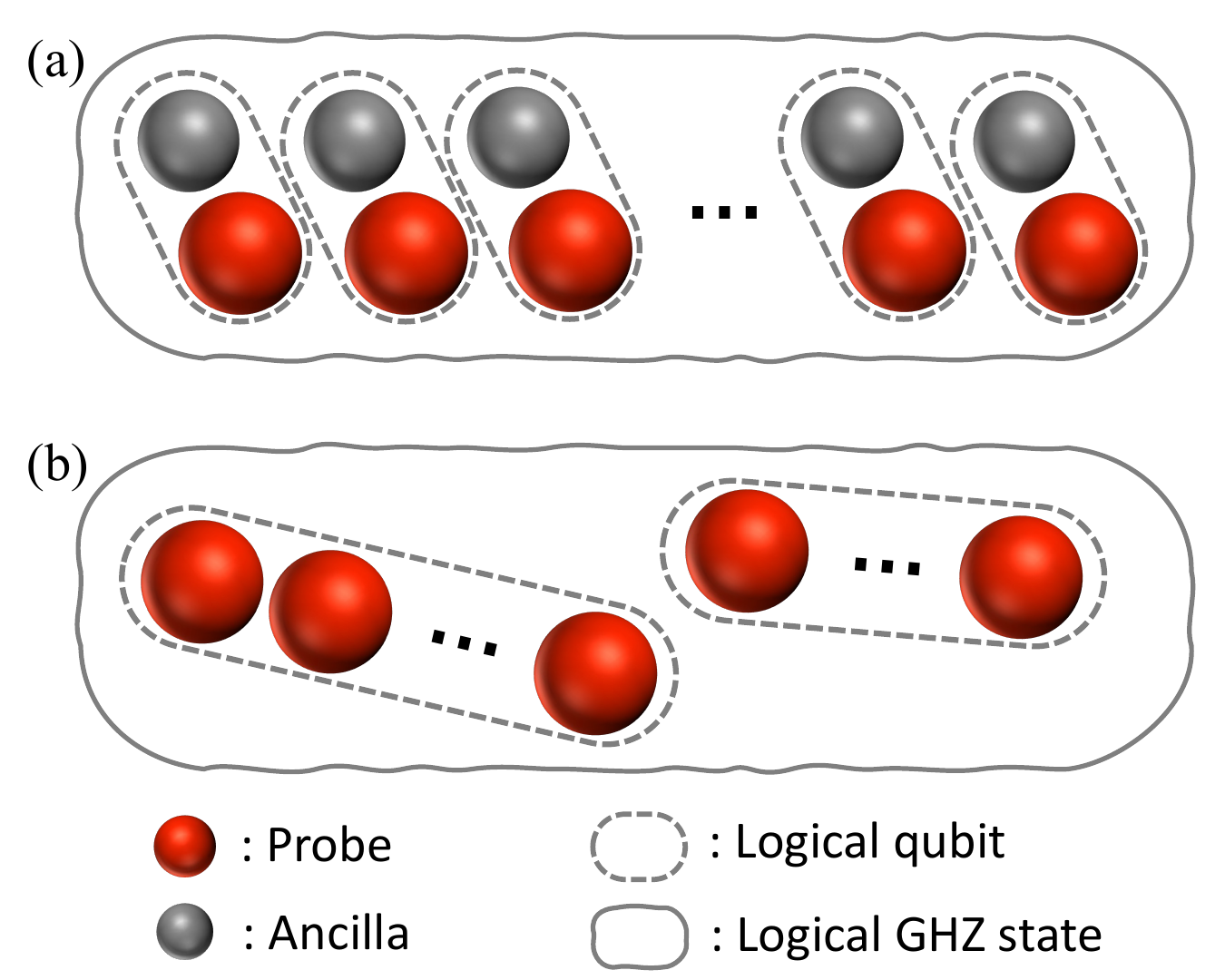}
\caption{(a)~Ancilla-assisted QEC strategy. For example, in \eqref{eq:code-ancilla}, one probe and one noiseless ancilla encode one logical qubit. The HL ($\delta \omega = O(1/N)$) is achieved by preparing a logical GHZ state using $N$ logical qubits. (b)~Ancilla-free multi-probe QEC strategy. For example, in \eqref{eq:code-ancilla-free}, $m=\Theta(N)$ probes encode one logical qubit. The HL is achieved by preparing a logical GHZ state containing $O(1)$ number of logical qubits. }
\label{fig:scheme}
\end{figure}

In general, we will assume
the density operator $\rho$ of each $d$-dimensional probe ($d \geq 2$) evolves under the Lindblad master equation (or Markovian noise), \sisi{which is the most general form for the generator of a quantum dynamical semigroup~\cite{breuer2002theory}.} We have   
\begin{equation}
\label{eq:master}
\frac{d\rho}{dt} = -i[\omega H,\rho] + \sum_{i=1}^r L_i \rho L_i^\dagger - \frac{1}{2}\{L_i^\dagger L_i,\rho\}, 
\end{equation}
where $\omega H$ is the Hamiltonian of the probe 
and $L_i$ are the Lindblad operators that cause dissipation in the system and are independent of $\omega$. \sisi{Note that some of our results below can be generalized to situations where the system evolution is an arbitrary function of $\omega$ (see details in \appref{app:beyond}), which we will not elaborate on in the main text.} 

Define $\mS$ to be the linear subspace of Hermitian operators spanned by $\id,L_i,L_i^\dagger$ and $L_i^\dagger L_j$ for all $i,j$. The HNLS condition states that the HL is achievable using adaptive quantum strategies~\footnote{By adaptive quantum strategies, we mean arbitrary quantum operations and measurements acting on the system sequentially with arbitrarily high frequencies and adaptive feedbacks (see details in \appref{app:heisenberg}).} if and only if HNLS holds, i.e., $H \notin \mS$~\cite{demkowicz2017adaptive,zhou2018achieving}. The QEC sensing strategy is a special type of adaptive quantum strategies where the QEC operations act sufficiently fast on the probe and ancilla. 
When HNLS holds, the HL is achievable using QEC~\cite{zhou2018achieving} and, up to an arbitrarily small error, 
\begin{equation}
\label{eq:QFI}
F(\rho_\omega(t)) = 4 N^2 t^2 \norm{H - \mS}^2 = 4 N^2 t^2 \min_{S \in \mS}\norm{H - S}^2, 
\end{equation}
where $\norm{\cdot}$ is the operator norm. Moreover, the HL coefficient, defined by,  
\begin{equation}
\label{eq:QFI-coeff-HL}
\sup_{t > 0} \lim_{N\rightarrow \infty}  \frac{F(\rho_\omega(t))}{N^2t^2} = 4 \norm{H - \mS}^2, 
\end{equation}
cannot be surpassed via arbitrary adaptive quantum strategies~\cite{wan2022bounds}, indicating the optimality of QEC. When HNLS fails, we must have $F(\rho_\omega(t)) = O(N)$ and the optimal SQL coefficient $\sup_{t > 0} \lim_{N\rightarrow \infty} {F(\rho_\omega(t))}/{(Nt)}$ is also achievable using QEC~\cite{zhou2019optimal}. We will say a QEC code \emph{achieves the optimal HL (or SQL)} if and only if there is a corresponding QEC strategy that achieves the HL (or SQL) and its optimal coefficient. 

Assuming HNLS holds, a QEC code assisted by noiseless ancillas can attain the optimal HL. 
As proven in~\cite{zhou2018achieving}, there exist two density matrices $\rho_0$ and $\rho_1$ in the probing system such that $\trace(\rho_0\rho_1) = 0$, and  \begin{gather}
\label{eq:rho-condition-1}
    \trace((\rho_0-\rho_1)H) = 2\norm{H - \mS}, \\
\label{eq:rho-condition-2}
    \trace((\rho_0-\rho_1)S) = 0,\;\forall S \in \mS.  
\end{gather}
Let $\rho_0 = \sum_{i=0}^{d_0-1} \lambda_i \ket{i}_P\bra{i}_P$ and  $\rho_1 = \sum_{i=d_0}^{d-1} \lambda_i \ket{i}_P\bra{i}_P$, where $\{\ket{i}_P\}_{i=0}^{d-1}$ is an orthonormal basis of the probe space $\mH_P$ and $\sum_{i=0}^{d_0-1}\lambda_i=\sum_{i=d_0}^{d-1}\lambda_i = 1$. The \emph{ancilla-assisted code}~\cite{zhou2018achieving} is defined in $\mH_P \otimes \mH_A$ by  
\begin{equation}
\label{eq:code-ancilla}
\ket{0_\tl} =\! \sum_{i=0}^{d_0-1}\! \sqrt{\lambda_i}\ket{i}_P\!\ket{i}_A,\;\;
\ket{1_\tl} =\! \sum_{i=d_0}^{d-1}\! \sqrt{\lambda_i}\ket{i}_P\!\ket{i}_A,
\end{equation}
where $\{\ket{i}_A\}_{i=0}^{d-1}$ is an orthonormal basis of the ancillary system $\mH_A$, assumed to be unaffected by both the signal and the noise. (We consider only two-dimensional codes in this paper because it is sufficient for single-parameter estimation.)
Using the ancilla-assisted code, there exists a QEC strategy such that the logical qubit evolves noiselessly as a logical Pauli-Z rotation (up to an arbitrarily small error): 
\begin{equation}
\label{eq:unitary}
\frac{d\rho_\tl}{dt} = -i\left[\omega \norm{H - \mS} Z_\tl, \rho_\tl\right], 
\end{equation}
where $Z_\tl$ is the logical Pauli-Z operator. The logical GHZ state then achieves the optimal HL (see \figaref{fig:scheme}).

\section{Qutrit Example}

Generically, every probe needs to be accompanied by a noiseless ancilla of the same dimension in order to achieve the HL. It is a stringent requirement in practical experiments, but can be relaxed in special cases. For example, consider a qubit evolving under a Pauli-Z Hamiltonian and a bit-flip (Pauli-X) noise~\cite{kessler2014quantum,dur2014improved,arrad2014increasing,ozeri2013heisenberg,peng2020achieving}, the ancilla-free repetition code that entangles three qubits $\ket{k_\tl} = \ket{k}^{\otimes 3}_P$ (for $k = 0,1$) 
performs equally well as the ancilla-assisted code $\ket{k_\tl} = \ket{k}_P\ket{k}_A$ (for $k = 0,1$) in terms of achieving the optimal HL.  

In this paper, we will see how the number of noiseless ancillas can be dramatically reduced using multi-probe encodings (see \figbref{fig:scheme}). We first illustrate this idea in a qutrit example (see details in \appref{app:qutrit}). Consider the qutrit evolution $\frac{d\rho}{dt} = -i[\omega H,\rho] + L \rho L^\dagger - \frac{1}{2}\{L^\dagger L,\rho\}$, 
where $H = \ket{0}\bra{0}-(\ket{1}\bra{1}+\ket{2}\bra{2}+\ket{1}\bra{2}+\ket{2}\bra{1})$ and $L = \ket{0}\bra{1} + \ket{0}\bra{2} + \ket{1}\bra{2}$. 
HNLS holds and we have $\rho_0 = \frac{1}{2}(\ket{0}\bra{0}+\ket{2}\bra{2})$, $\rho_1 = \ket{1}\bra{1}$.
The ancilla-assisted code is then 
\begin{equation}
\ket{0_\tl} = \frac{1}{\sqrt{2}}(\ket{0}_P\ket{0}_A+\ket{2}_P\ket{2}_A),\;\;\ket{1_\tl} = \ket{1}_P\ket{1}_A. 
\end{equation}
We call it a $(1,1)$-qutrit code, for the encoding contains $1$ qutrit probe and $1$ qutrit ancilla. 

As proven in \appref{app:qutrit}, any $(m,1)$ code in $\mH_P^{\otimes m}\otimes \mH_A$ achieves the optimal HL if it satisfies, for any $k,k' = 0,1$, $\ell \neq \ell'$ and $\ell,\ell' = 1,\ldots,m$, 
\begin{equation}
\label{eq:sensing-condition-2}
    \trace_{\backslash\{\ell,\ell'\}}(\ket{k_\tl}\bra{k'_\tl}) = (\rho_k^{(\ell)} \otimes \rho_k^{(\ell')} - Q_k^{(\ell)} \otimes Q_k^{(\ell')})\delta_{kk'},
\end{equation}
where $\trace_{\backslash\{\ell,\ell'\}}$ means tracing out the entire system except for the $\ell$-th and $\ell'$-th probes, we use $^{(\ell)}$ to denote an operator acting on the $\ell$-th probe and $Q_k$ can be any traceless Hermitian operator satisfying $\trace(Q_k L_i)  = 0, \forall i$. We call \eqref{eq:sensing-condition-2} a multi-probe QEC condition in this paper. For example, the $(3,0)$ repetition code for Pauli-Z signal and Pauli-X noise satisfies it, i.e., $\trace_{\backslash\{\ell,\ell'\}}(\ket{k}\bra{k'}^{\otimes 3}) = \rho_k^{\otimes 2}\delta_{kk'}$ where $\rho_k = \ket{k}\bra{k}$ and $Q_k = 0$. In the qutrit case above, a $(4,1)$-qutrit code satisfying this condition for $m=4$ is 
\begin{align}
\ket{0_\tl} = \frac{1}{\sqrt{6}} \big( & ( \ket{0022}_{P^{\otimes 4}}+\ket{2200}_{P^{\otimes 4}})\ket{0}_A \nonumber\\ +&(\ket{0202}_{P^{\otimes 4}}+\ket{2020}_{P^{\otimes 4}})\ket{1}_A \nonumber\\ +& (\ket{0220}_{P^{\otimes 4}}+\ket{2002}_{P^{\otimes 4}})\ket{2}_A\big),
\end{align}
and $\ket{1_\tl} = \ket{1111}_{P^{\otimes 4}}\ket{0}_{A}$. 

The advantage of a smaller number of ancillas arises when the noiseless ancilla is a limited resource. Imagine a simplified scenario where a qutrit probe and a qutrit noiseless ancilla are two equivalent units of resource (in reality the latter can be more difficult to prepare, e.g., a noiseless ancilla can be replaced by a constant number of noisy probes under \sisi{code concatenation} (see \appref{app:heisenberg})). Then the achievable QFI using $N_{\textsc{u}}$ units of resource in a probing time $t$ is $\frac{1}{4}N_{\textsc{u}}^2t^2$ via the $(1,1)$-qutrit encoding and a $\frac{N_{\textsc{u}}}{2}$-qubit logical GHZ state, and $\frac{16}{25}N_{\textsc{u}}^2t^2$ via the $(4,1)$-qutrit encoding and a $\frac{4N_{\textsc{u}}}{5}$-qubit logical GHZ state. Clearly, a higher sensitivity is achieved when more resources can be used as probes instead of ancillas.

\section{Ancilla-free QEC codes}

Here we present two types of QEC codes that achieve the optimal HL when HNLS holds, where the fraction of noiseless ancillas is negligible asymptotically.
The QEC codes that achieve the optimal SQL when HNLS fails can be defined analogously, with details provided in \appref{app:SQL}. \sisi{These codes are multi-probe QEC codes whose encoding and recovery channels both require global operations across multiple probes, in contrast to the previous ancilla-assisted code whose encoding and recovery only apply to single pairs of probe and ancilla. }

The first code, which we will refer to as \emph{small-ancilla code}, entangles $m \geq 3$ probes and one ancilla (i.e., $\mH_P^{\otimes m}\otimes \mH_A$). We will prove later that it achieves the optimal HL when $m=\Theta(N)$, although in special circumstances (e.g., the $(4,1)$-qutrit code), the optimal HL is achievable with a constant $m$. The codewords are 
\begin{equation}
\label{eq:code-few-ancilla}
\ket{k_\tl} = \frac{1}{\sqrt{\abs{W_k}}} \sum_{w \in W_k} \ket{w}_{P^{\otimes m}} \ket{\fraki_{k}(w)}_A,
\end{equation}
for $k = 0,1$. Here we use $\ket{w}$ to denote $\ket{w_1}\otimes\cdots\otimes \ket{w_m}$ for any string $w = w_1 w_2\cdots w_m \in \{0,\ldots,d-1\}^{m}$. $W_0$ and $W_1$ are two non-overlapping sets of strings of length $m$ and $\abs{W_k}$ denotes the order of $W_k$. To define $W_{0,1}$, note that for all $m$, it is possible to find a set of integers $\{m_i\}_{i=0}^{d-1}$ such that $\sum_{i=0}^{d_0-1} m_i = \sum_{i=d_0}^{d-1} m_i = m$ and 
\begin{equation}
\label{eq:rounding}
    \abs{\frac{m_i}{m} - \lambda_i} \leq \frac{1}{m},\;\forall i=0,\ldots,d-1. 
\end{equation} 
Note that it satisfies $\lim_{m\rightarrow \infty} m_i/m = \lambda_i$. 
We then define $W_0$ (or $W_1$) to be the set of strings that contains $m_i$ $i$'s for $0 \leq i \leq d_0-1$ (or $d_0 \leq i \leq d-1$). 
\sisi{$\fraki_k(w)$ are functions that satisfy $\fraki_k(w)\neq \fraki_k(w')$ when $w$ and $w'$ are different on exactly two sites (i.e., $w'$ can be obtained from $w$ (and vice versa) by swapping two sites with different numbers). We show in \appref{app:graph}
 that the highest $\dim(\mH_A)$ needed is no larger than $d^2m^2$, making the size of the ancilla exponentially smaller than that of probes. The small-ancilla code is then effectively ancilla-free (because the ancilla can always be replaced by a negligible amount of noisy probes under code concatenation, see \appref{app:heisenberg}).}

Note that the above construction guarantees that for $m \gg 1$, $k,k' = 0,1$ and $\ell \neq \ell'$, the multi-probe QEC condition (when choosing $Q_k = 0$) is satisfied approximately, which is a key ingredient in proving the optimality of the small-ancilla code. 

The second code exploits random phases such that the multi-probe QEC condition holds when $m \gg 1$. We call it the \emph{ancilla-free random code}, which entangles $m \geq 3$ probes without ancillas (i.e., $\mH_P^{\otimes m}$). It also achieves the optimal HL when $m=\Theta(N)$, and is defined by
\begin{equation}
\label{eq:code-ancilla-free}
\ket{k_\tl} = \frac{1}{\sqrt{\abs{W_k}}} \sum_{w \in W_k} e^{i\theta_w} \ket{w}_{P^{\otimes m}}, 
\end{equation}
for $k=0,1$, where $\{\theta_w\}_{w\in W_0\cup W_1}$ are sampled from a set of independent and identically distributed random variables following the uniform distribution in $[0,2\pi)$. 
The ancilla-free random code achieves the optimal HL when $m=\Theta(N)$ with probability $1 - e^{-\Omega(m)}$.

\section{QEC strategy}

\sisilong{Now we delve into the details of the QEC strategy, including how to choose the recovery channel and the input state, that allow us to achieve the optimal HL with the (effectively) ancilla-free codes we defined above. (The SQL case is discussed in \appref{app:SQL}
.) In particular, we will consider recovery channels that map the system dynamics to a Pauli-Z signal under dephasing noise in the logical space and show that the logical signal strength increases linearly with $N$, while the dephasing noise rate at most a constant, leading to the HL. 


We first examine the state evolution in the logical system under a general recovery channel $\mR(\cdot)$. 
Let $P = \ket{0_\tl}\bra{0_\tl} + \ket{1_\tl}\bra{1_\tl}$ be the projection operator onto the code subspace. We apply the QEC operation $\mP + \mR\circ\mP_\perp$ constantly on the quantum state~\cite{zhou2018achieving,zhou2019optimal}, where $\mP(\cdot) = P(\cdot)P$, $\mP_\perp(\cdot) = P_\perp (\cdot) P_\perp$, $P_\perp = \id - P$ is the projection onto the orthogonal subspace to the code and $\mR$ is a recovery operation satisfying $\mR = \mP \circ \mR$. As shown in \appref{app:dephasing}, the logical state $\rho_\tl \in \mH_P^{\otimes m} \otimes \mH_A$ (satisfying $\rho_\tl = P \rho_\tl P$) evolves according to the following logical master equation 
{\small \begin{multline}
\label{eq:QEC-operation}
    \frac{d\rho_\tl}{dt} = -i\left[\omega \mP\left(\sum_{\ell=1}^m H^{(\ell)}\right),\rho_\tl\right] + \sum_{i=1}^r \sum_{\ell=1}^m \mP\left(L_i^{(\ell)}\rho_\tl L_i^{(\ell) \dagger }\right)  \\
+ \mR\left(\mP_\perp\left(L_i^{(\ell)}\rho_\tl L_i^{(\ell)\dagger}\right)\right)
- \frac{1}{2}\left\{\mP\left(L_i^{(\ell)\dagger}L_i^{(\ell)}\right),\rho_\tl\right\}.
\end{multline}
}
Furthermore, when $\bra{0_\tl} L_i^{(\ell)\dagger} L_j^{(\ell')} \ket{1_\tl} = \bra{0_\tl} L_j^{(\ell)} \ket{1_\tl} = 0$ for all $i,j,\ell,\ell'$ (which holds true for all codes discussed here including the ancilla-assisted code (taking $m=1$), the small-ancilla code and the ancilla-free random code (taking $\dim(\mH_A) = 1$)), there exists a recovery channel $\mR(\cdot)$ such that the logical evolution is Pauli-Z signal under dephasing noise: 
\begin{multline}
\label{eq:error-corrected-master}
    \frac{d\rho_\tl}{dt} = -i\left[\frac{\omega \trace(\sum_\ell H^{(\ell)} Z_\tl)}{2}Z_\tl + \frac{\beta_\tl}{2} Z_\tl,\rho_\tl\right] \\ + \frac{\gamma_\tl}{2}(Z_\tl \rho_\tl Z_\tl - \rho_\tl),
\end{multline}
where the logical noise rate is 
{\small 
\begin{align}
    &\gamma_\tl = \sum_{i=1}^r \sum_{\ell=1}^m - \Re\big[\bra{0_\tl} L_i^{(\ell)} \ket{0_\tl} \nonumber\bra{1_\tl} L_i^{(\ell)\dagger} \ket{1_\tl}\big] \\& +  \frac{1}{2}(\bra{0_\tl} L_i^{(\ell)\dagger} L_i^{(\ell)} \ket{0_\tl} + \bra{1_\tl} L_i^{(\ell)\dagger} L_i^{(\ell)} \ket{1_\tl}) \label{eq:gamma-L}\\& - \norm{\sum_{i=1}^r \sum_{\ell=1}^m(\id - \ket{0_\tl}\bra{0_\tl})(L_i^{(\ell)} \ket{0_\tl}\bra{1_\tl} L_i^{(\ell)\dagger} )(\id - \ket{1_\tl}\bra{1_\tl})}_1, \nonumber
\end{align}
}$\norm{\cdot}_1$ is the trace norm, and $\beta_\tl$ is independent of $\omega$. 

Consider $n$ logical qubits (the total number of probes is $N = mn$), evolving under \eqref{eq:error-corrected-master} with an initial state $\ket{\psi_{\rm in}} = \frac{1}{\sqrt{2}} (\ket{0_\tl}^{\otimes n}+\ket{1_\tl}^{\otimes n})$. We have (from \eqref{eq:dephasing-QFI}) that 
\begin{equation}
\label{eq:error-corrected-QFI}
    F(\rho_\omega(t)) = n^2 \bigg(\sum_{\ell=1}^m \trace(H^{(\ell)}Z_\tl )\bigg)^2t^2 e^{- 2 n \gamma_\tl t}. 
\end{equation}
As a sanity check and a preparation for later proofs, we first show the optimal HL is achieved using the ancilla-assisted code, which was a previous result in~\cite{zhou2018achieving}. 
First note that taking $m=1$, $n = N$, the HL $F(\rho_\omega(t)) = \Theta(N^2)$ is achieved when the noise rate $\gamma_\tl = 0$. 
To see this, we first notice that the gauge transformation (i)~$L_i \leftarrow L_i - x_i \id$ for any $x_i \in \bC$ and (ii)~$L_i \leftarrow \sum_{j} u_{ij} L_j$ for any isometry $u$ satisfying $\sum_i u_{ik}^* u_{ij} = \delta_{kj}$ changes only the value of $\beta_\tl$, which is equivalent to applying a constant shift on $\omega$ that can be offset easily given $\beta_\tl$ is a known value.  
Therefore, we can assume without loss of generality (using \eqref{eq:rho-condition-2}) that 
\begin{gather}
\label{eq:gauge-1}
    \trace(\rho_0 L_i) = \trace(\rho_1 L_i) = 0,\;\forall i,\\
\label{eq:gauge-2}
    \trace(\rho_0 L_i^\dagger L_j) = \trace(\rho_1 L_i^\dagger L_j) = \mu_{i}\delta_{ij}, \;\forall i,j, 
\end{gather}
where $\mu_i \geq 0$.
\eqref{eq:code-ancilla} implies $\braket{0_\tl|L_i|0_\tl} = \braket{1_\tl|L_i|1_\tl} = 0$ and $\braket{0_\tl|L_i^\dagger L_j|0_\tl} = \braket{1_\tl|L_i^\dagger L_j|1_\tl} = \mu_{i}\delta_{ij}$, leading to 
$\gamma_\tl = \sum_{i=1}^r \mu_i - \big\|\sum_{i=1}^r L_i\ket{0_\tl}\bra{1_\tl}L_i^\dagger\big\|_1 = 0$, as expected. Moreover, $\sum_\ell \trace(H^{(\ell)} Z_\tl) = m\trace(H (\rho_0 - \rho_1)) = 2m\norm{H-\mS}$. Then the optimal HL (\eqref{eq:QFI-coeff-HL}) is achieved using the ancilla-assisted code.

For the small-ancilla code and the ancilla-free random code, $\gamma_\tl$ can be positive due to the fact that \eqref{eq:sensing-condition-2} is only approximately true. Nonetheless, the HL is still achievable, when $\gamma_\tl$ is at most a constant and $\trace(\sum_\ell H^{(\ell)} Z_\tl)$ increases linearly with $N$. (Note that without QEC, the noise rate will increase linearly with $N$.) Specifically, we have 
\begin{equation}
    F(\rho_\omega(t)) = n^2 m^2 \trace((\tilde{\rho}_0-\tilde{\rho}_1)H)^2 t^2 e^{- 2 n \gamma_\tl t},
\end{equation}
where $\tilde\rho_k = \trace_{\backslash\{\ell\}}(\ket{k_\tl}\bra{k_\tl})$ for $k=0,1$, where $\trace_{\backslash\{\ell\}}$ means tracing out the entire system except for the $\ell$-th probe. (Note that $\tilde\rho_k = \rho_k + O(1/m)$ is independent of $\ell$.) The optimal HL (\eqref{eq:QFI-coeff-HL}) is then achievable, taking $n=\Theta(1)$ and $m = \Theta(N)$, if  
\begin{gather}
\label{eq:limit-1}
    \lim_{m\rightarrow \infty} \trace((\tilde{\rho}_0-\tilde{\rho}_1)H) = 2\norm{H-\mS},\\
\label{eq:limit-2}
    \limsup_{m\rightarrow \infty} \gamma_\tl = \overline{\gamma}_\tl < + \infty. 
\end{gather} 
In particular, when $\overline{\gamma}_\tl > 0$, the optimal HL is achieved in the regime $t \ll 1/\overline{\gamma}_\tl$. The proofs of \eqref{eq:limit-1} and \eqref{eq:limit-2} are left in \appref{app:limit}
, which also includes a numerical simulation of the qutrit code up to $m = 7$. 
}

\section{Discussion}

This paper answers the long-standing open question of the requirement of noiseless ancillas in noisy quantum metrology. To achieve this, we constructed two types of multi-probe QEC codes that are (effectively) ancilla-free and achieve the optimal HL when HNLS holds (or the optimal SQL when HNLS fails). We started from the previously known optimal ancilla-assisted QEC code that entangles one probe and one ancilla and designed multi-probe codes such that their 2-local reduced density operators resemble those of ancilla-assisted codes. We then proved their optimality by choosing suitable recovery channels. Note that our construction does not rule out the possibility that the metrological limits can be achieved using other types of QEC encoding and recovery, which have not yet been explored. 

There are several open problems remaining to be addressed. 
First, our codes achieve the optimal HL with respect to the number of probes, but sometimes fails when the probing time is too long and it remains open whether there are ancilla-free codes that achieve the optimal HL with respect to the probing time. Second, efficient encoding and decoding algorithms are yet to be found for our multi-probe codes. \sisinew{Finally, all QEC operations are assumed to be perfect here and the efficiency of our protocols under faulty QEC operations is left for future study. }

\begin{acknowledgments}

We thank Senrui Chen for helpful discussion. S.Z. acknowledges funding provided by the Institute for Quantum Information and Matter, an NSF Physics Frontiers Center (NSF Grant PHY-1733907) and Perimeter Institute for Theoretical Physics, a research institute supported in part by the Government of Canada through the Department of Innovation, Science and Economic Development Canada and by the Province of Ontario through the Ministry of Colleges and Universities.
A.G.M. acknowledges support from the Caltech Summer Undergraduate Research Fellowships (SURF) Program and the John Preskill Group. L.J. acknowledges support from the ARO (W911NF-23-1-0077), ARO MURI (W911NF-21-1-0325), AFOSR MURI (FA9550-19-1-0399, FA9550-21-1-0209), AFRL (FA8649-21-P-0781), DoE Q-NEXT, NSF (OMA-1936118, ERC-1941583, OMA-2137642), NTT Research, and the Packard Foundation (2020-71479).

\end{acknowledgments}

\onecolumngrid

\appendix

\setcounter{theorem}{0}
\setcounter{proposition}{0}
\setcounter{lemma}{0}
\setcounter{figure}{0}
\renewcommand{\thefigure}{S\arabic{figure}}
\renewcommand{\thelemma}{S\arabic{lemma}}
\renewcommand{\thetheorem}{S\arabic{theorem}}
\renewcommand{\thecorollary}{S\arabic{corollary}}
\renewcommand{\theproposition}{S\arabic{proposition}}
\renewcommand{\theHfigure}{Supplement.\arabic{figure}}
\renewcommand{\theHlemma}{Supplement.\arabic{lemma}}
\renewcommand{\theHtheorem}{Supplement.\arabic{theorem}}
\renewcommand{\theHcorollary}{Supplement.\arabic{corollary}}

\setcounter{page}{1}





\section{Adaptive quantum strategies and QEC strategies for noisy quantum metrology}
\label{app:heisenberg}

In this appendix, we explain in detail the common metrological strategies used in noisy quantum metrology, and clarify our definitions of the HL and the SQL in the main text. We consider Hamiltonian estimation under Markovian noise where the probe evolution is described by~\eqref{eq:master}. Let the quantum channel 
\begin{equation}
    \mE_\omega(\rho) = \rho + \left( -i[\omega H,\rho] + \sum_{i=1}^r L_i \rho L_i^\dagger - \frac{1}{2}\{L_i^\dagger L_i,\rho\}\right) dt + O(dt^2)
\end{equation}
describe the probe evolution in a small time $dt$. 

\emph{Adaptive quantum strategies} (or sequential strategies, shown in \figaref{fig:strategy}) are usually considered as the most general quantum strategies~\cite{demkowicz2014using,sekatski2017quantum,demkowicz2017adaptive}, where the system consists of a probe and an ancillary system of an arbitrarily large dimension. Arbitrary quantum controls (i.e., quantum channels) $\mC_1,\mC_2,\ldots$ are applied instantaneously every time interval $dt$. The total number of controls is $T/dt$, where $T$ is the probing time. The goal is to maximize the QFI of final states $\rho_\omega$ over all quantum controls and initial states $\ket{\psi_{\rm in}}$, taking the limit $dt \rightarrow 0$. It was shown that~\cite{demkowicz2017adaptive,zhou2018achieving,wan2022bounds}, for any adaptive quantum strategy, we always have 
\begin{equation}
    F(\rho_\omega(T)) \leq 4 \norm{H - \mS} T^2 + O(T^{3/2}), 
\end{equation}
when HNLS holds (i.e., $H \notin \mS$). In this case, there exists an ancilla-assisted QEC strategy achieving $F(\rho_\omega) = \Theta(T^2)$~\cite{zhou2018achieving}. On the other hand, when HNLS fails (i.e., $H \in \mS$), we always have
\begin{equation}
\label{eq:alpha}
    F(\rho_\omega(T)) \leq 4 {o_1} T, \quad {o_1}: = 
\min_{\substack{\vh \in \bC^r, \text{~Hermitian~}\frakh \in \bC^{r \times r}, \text{~s.t.~} \\ H - \sum_{j}(\vh_j^* L_j + \vh_j L_j^\dagger) - \sum_{ij} \frakh_{ij}L_i^\dagger L_j \propto \id}} \norm{O_1},
\end{equation}
where 
\begin{equation}
    O_1 = \sum_i \bigg(\vh_i \id +  \sum_j \frakh_{ij}L_j\bigg)^\dagger \bigg(\vh_i \id +  \sum_j  \frakh_{ij}L_j\bigg). 
\end{equation} 
It means the QFI cannot go beyond the SQL.

\begin{figure}[t]
\centering
\includegraphics[width=0.9\textwidth]{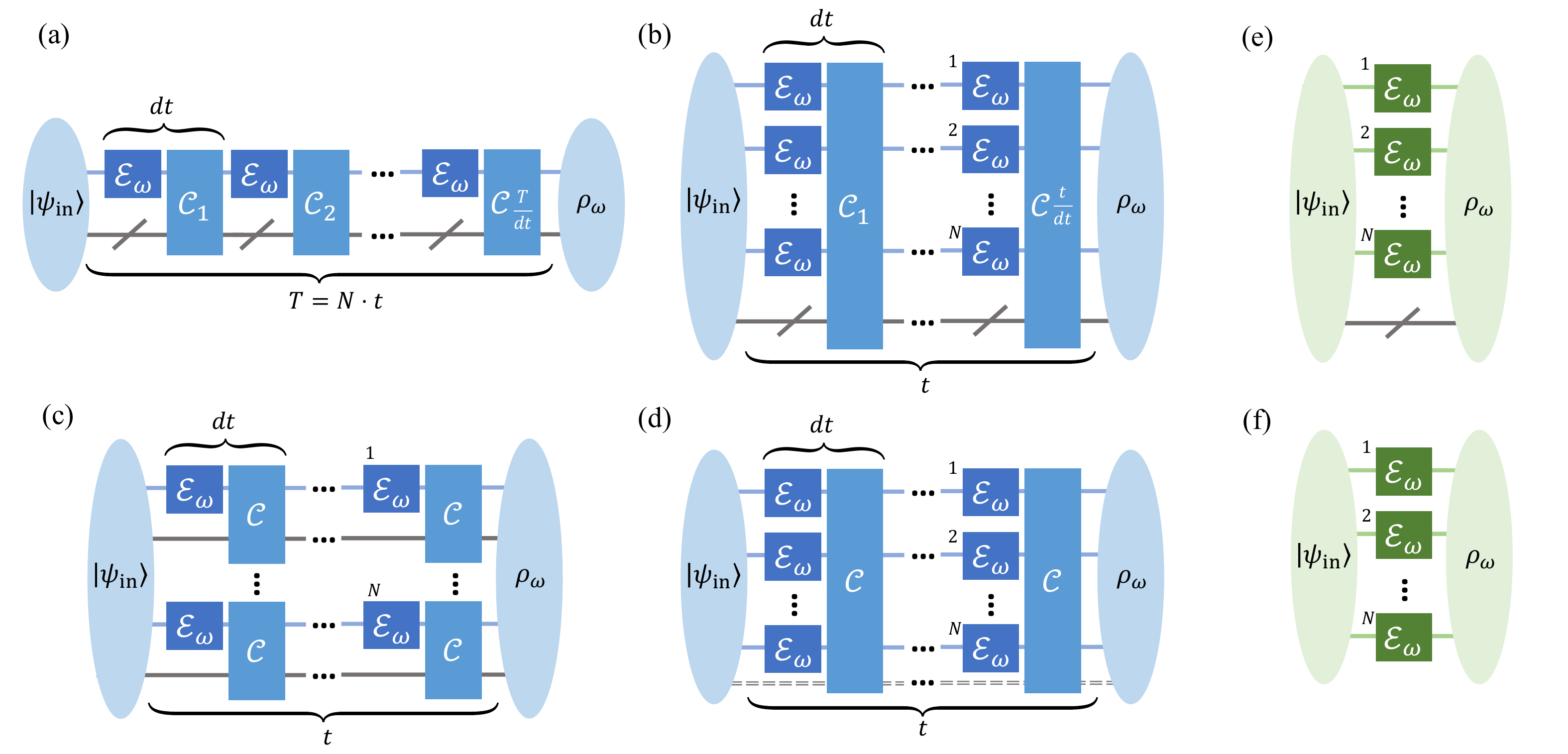}
\caption{(a)-(d) Hamiltonian estimation under Markovian noise. (a)~Adaptive quantum strategies (or sequential strategies). (b)~Parallel strategies. (c)~Ancilla-assisted QEC strategies. (d)~Multi-probe QEC strategies that are effectively ancilla-free. Here we set the number of probes in a logical qubit $m$ equal to the total number of probes $N$; while in general $N$ can be a multiple of $m$ and the QEC operation will act repeatedly in parallel on each multi-probe logical qubit like in (c). Note that we use slashes over lines to represent ancillary systems whose dimension can be arbitrarily large; and dashed lines to denote small ancillary systems whose dimensions are exponentially small compared to the dimension of the probes. (e)-(f) Quantum channel estimation. (e) Parallel strategies. (f) Ancilla-free parallel strategies. }
\label{fig:strategy}
\end{figure}

\emph{Parallel strategies} form a subset of adaptive quantum strategies, as shown in \figbref{fig:strategy}, where the system consists of $N$ probes and an ancillary system of an arbitrarily large dimension and the probing time is taken as $t = T/N$ (in comparison to the probing time $T$ in adaptive quantum strategies). \figbref{fig:strategy} can be reduced to \figaref{fig:strategy} if we apply in each step in \figaref{fig:strategy} a swap gate between the probe and one ancilla, making parallel strategies a subset of adaptive strategies. Therefore, we have, for parallel strategies,  
\begin{equation}
    F(\rho_\omega(t)) \leq
    \begin{cases}
    4 \norm{H - \mS} N^2 t^2 + O((Nt)^{3/2}), & \text{~when~} H \notin \mS, \\ 
    4 {o_1} N t, & \text{~when~} H \in \mS. 
    \end{cases}
\end{equation}
When we fix $t$ and let $N \rightarrow \infty$, the scaling $F(\rho_\omega) = \Theta(N^2)$ is \emph{the HL with respect to the number of probes $N$}, which is the definition of HL we focus on in this paper. In particular, we define the asymptotic HL (or SQL) coefficient to be $\sup_{t > 0} \lim_{N\rightarrow \infty} {F(\rho_\omega(t))}/{(Nt)^2}$ (or $\sup_{t > 0} \lim_{N\rightarrow \infty} {F(\rho_\omega(t))}/{(Nt)}$) and we say a metrological protocol achieves the optimal asymptotic HL (or SQL) coefficient if the coefficient reaches its upper bound $4 \norm{H - \mS}^2$ (or $4{o_1}$). On the other hand, when we fix $N$ and let $t \rightarrow \infty$, the scaling $F(\rho_\omega) = \Theta(t^2)$ is \emph{the HL with respect to the probing time $t$}. Our ancilla-free QEC strategies do not necessarily reach the HL with respect to $t$. 

We discuss in this paper both the case where the HL (with respect to $N$) is achievable and the case where only the SQL (with respect to $N$) is achievable. In particular, we design (effectively) ancilla-free QEC strategies that achieve the optimal HL (or SQL) with respect to $N$, where only the ancilla-assisted QEC protocols were known previously~\cite{zhou2018achieving,zhou2019optimal}. 

The QEC metrological strategies are special types of parallel strategies. \figcref{fig:strategy} illustrates the ancilla-assisted QEC strategy, e.g., using the ancilla-assisted code (\eqref{eq:code-ancilla}) when HNLS holds, or the ancilla-assisted code$^\sql$ (\eqref{eq:code-ancilla-SQL}) when HNLS fails, where the QEC operation $\mC$ applies repeatedly in parallel on each pair of a probe and an ancilla of the same dimension (or twice the probe's dimension in the SQL case). 

\figdref{fig:strategy} illustrates our newly introduced multi-probe QEC strategies that are effectively ancilla-free. For simplicity, \figdref{fig:strategy} specifically shows the case where $m$, the number of probes in a logical qubit, is equal to $N$ (i.e., we encode the logical qubit in the entire system); while in general $N$ can be a multiple of $m$ and the QEC operation will act repeatedly in parallel on each multi-probe logical qubit like in \figcref{fig:strategy}. Our multi-probe QEC strategies include QEC using the small-ancilla code (\eqref{eq:code-few-ancilla}) and the ancilla-free random code (\eqref{eq:code-ancilla-free}) when HNLS holds, and the the small-ancilla code$^\sql$ (\eqref{eq:code-few-ancilla-SQL}) and the qubit-ancilla random code$^\sql$ (\eqref{eq:code-qubit-ancilla-SQL}) when HNLS fails. We prove that when $m = \Theta(N)$, the optimal HL is achievable using our multi-probe QEC strategies. 
Moreover, we show that in order to achieve the optimal HL, we only need an ancilla whose dimension is exponentially smaller than that of the probes. It means that these codes are effectively ancilla-free because the ancilla can always be replaced by a negligible amount of probes under \sisi{code concatenation~\cite{zhou2018achieving}}.

\sisilong{

\subsection*{Small-ancilla codes are effectively ancilla-free using code concatenation}

\begin{figure}[tb]
\centering
\includegraphics[width=0.5\textwidth]{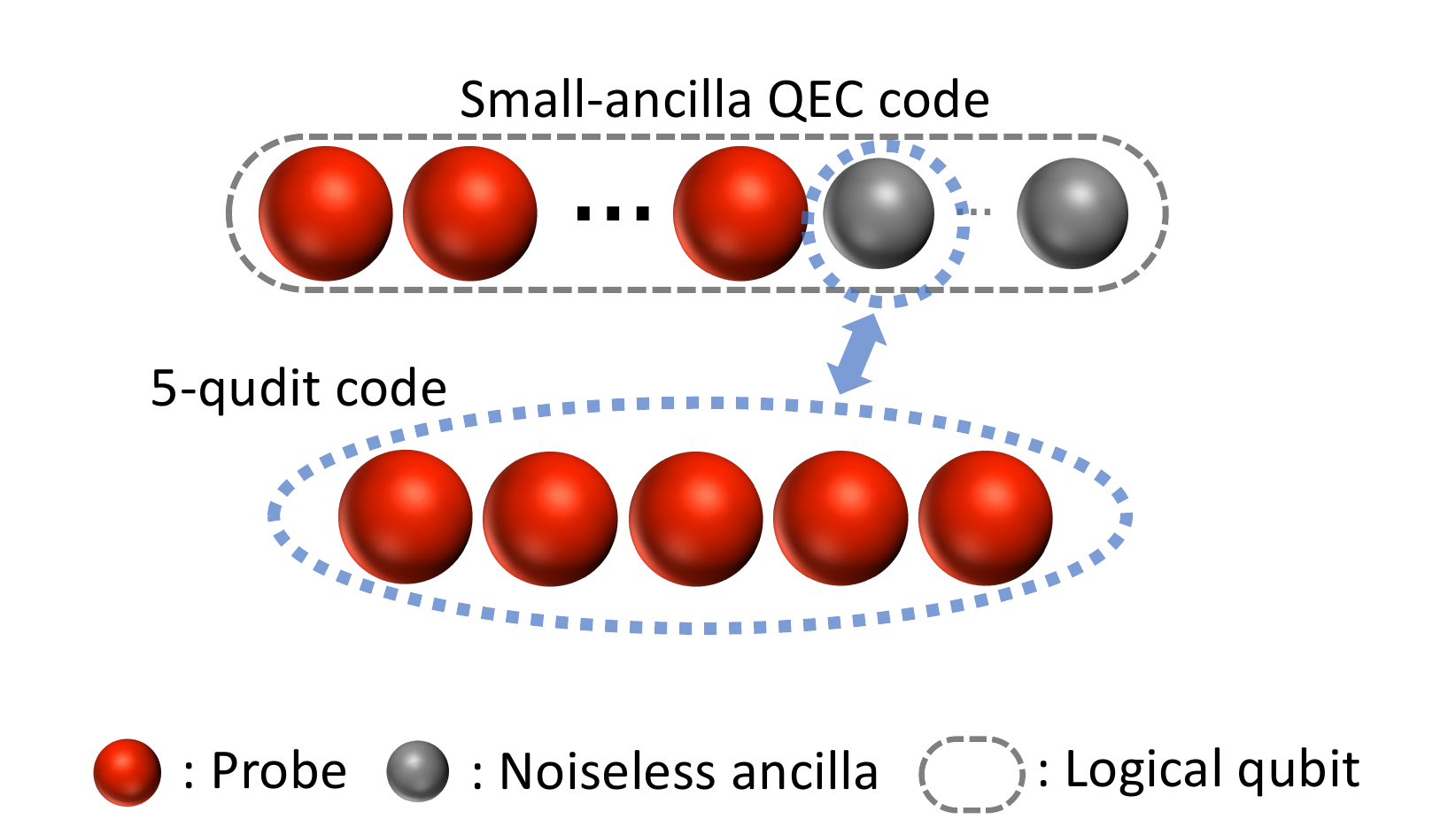}
\caption{\sisi{Small-ancilla code is effectively ancilla-free because one can set aside $O(\log(N))$ probes among $N$ total probes to make them function as noiseless ancillas in our construction through an inner layer of QEC, e.g., through the five-qudit code~\cite{laflamme1996perfect,chau1997five}. Note that we say a code achieves the optimal HL if its leading order term $O(N^2)$ is optimal, therefore $O(\log(N))$ probes are negligible in terms of achieving the optimal HL. }}
\label{fig:effective}
\end{figure}

Here we explain in detail why the small-ancilla code can be considered as an effectively ancilla-free code, even though the noiseless ancilla is still used in the construction. The idea is that given a large number of probes, say $N$ probes, we can always set aside of a negligible amount of probes of size $O(\log(N))$ and let them function as the noiseless ancilla part of the small-ancilla code by performing an inner layer of QEC on them (see \figref{fig:effective}). As we only care about the leading order $O(N^2)$ of the QFI, probes of size $O(\log(N))$ are negligible in terms of achieving the optimal HL, making the small-ancilla code \emph{effectively} ancilla-free.  This procedure is called code concatenation, that has also been used before in QEC-assisted metrology~\cite{zhou2018achieving}. 

From the main text, the small-ancilla code is a QEC code that entangles $m$ $d$-dimensional probes and one ancillary system of dimension at most $m^2d^2$. Given a total number of $N$ probes, we consider encoding the entire system using the small-ancilla code. For example, we could take
\begin{equation}
N = m + m_A,
\end{equation}
where $m$ is the number of probes used in the encoding and $m_A$ is the number of probes to be used as the ancillary system. Using the five-qudit code~\cite{laflamme1996perfect,chau1997five} that perfectly corrects arbitrary single-qudit errors, $m_A$ noisy probes can be converted into $m_A/5$ $d$-dimensional noiseless ancillas (assuming $m_A$ is a multiple of $5$). Then the small-ancilla code exists, as long as 
\begin{equation}
\dim(\mH_A) = d^{\frac{m_A}{5}} \geq m^2d^2,
\end{equation}
since the highest $\dim(\mH_A)$ needed is no larger than $m^2d^2$. For example, we can take 
\begin{equation}
m_A = \left\lceil 5 \frac{\log(m^2d^2)}{\log d} \right\rceil. 
\end{equation}
Then since $\lim_{N\rightarrow} m / N = 1$, the optimal HL coefficient 
\begin{equation}
    \sup_{t > 0} \lim_{N\rightarrow \infty}  \frac{F(\rho_\omega(t))}{N^2t^2} = \sup_{t > 0} \lim_{m\rightarrow \infty}  \frac{F(\rho_\omega(t))}{m^2t^2}= 4 \norm{H - \mS}^2
\end{equation}
is still attainable using the small-ancilla code when HNLS is satisfied, even without the help of an actual physical noiseless ancillas. Similar arguments can show that both the small-ancilla code$^\sql$ and the qubit-ancilla code$^\sql$ are effectively ancilla-free in terms of achieving the optimal SQL, when HNLS is violated. 

Note that here we assume all QEC operations are perfect so that five noisy qudit probes can perfectly simulate one noiseless ancilla. However, such a replacement may not work when QEC operations are faulty, and the discussion of fault-tolerant QEC is beyond the scope of this paper. 
}

\subsection*{Previous results on the role of the ancilla}

Finally, we remark on previous results on the role of the ancilla in quantum metrology and related fields. 

In quantum channel estimation where $\mE_\omega$ is an arbitrary quantum channel as a function of $\omega$, it was previously known that the ancilla-assisted parallel quantum strategy (\figeref{fig:strategy}) performs better than the ancilla-free parallel quantum strategy (\figfref{fig:strategy})~\cite{demkowicz2014using}. For phase estimation under amplitude damping noise, 
\begin{equation}
    \mE_\omega(\cdot) = (\ket{0}\bra{0}+\sqrt{\eta}\ket{1}\bra{1})e^{i\frac{\omega}{2} Z}(\cdot)e^{-i\frac{\omega}{2} Z}(\ket{0}\bra{0}+\sqrt{\eta}\ket{1}\bra{1}) + (1-\eta)\ket{0}\bra{1}e^{i\frac{\omega}{2} Z}(\cdot)e^{-i\frac{\omega}{2} Z}\ket{1}\bra{0}, \;0 < \eta < 1, 
\end{equation}
the optimal QFI of the ancilla-assisted parallel strategy is $\frac{4N\eta}{1-\eta}$, larger than the optimal QFI of the ancilla-free parallel strategy $\frac{N\eta}{1-\eta}$~\cite{demkowicz2014using}. The result seemingly contradicts with our conclusion that ancilla-free strategies performs equally well as ancilla-assisted strategies in the case of Hamiltonian estimation under Markovian noise. The reason could lie in that we allow frequent QEC acting repeatedly in the system evolution (e.g., in \figdref{fig:strategy}), while in quantum channel estimation no controls are  allowed during the channel application, but only before and after it (\figfref{fig:strategy}). 

Besides single-parameter quantum metrology, ancilla also plays an important role in multi-parameter quantum metrology. In multi-parameter estimation, the incompatibility of optimal measurements for different parameters is a crucial issue and ancilla helps alleviate such incompatibility. For example, in quantum sensing of magnetic field strength and frequency, a maximally entangled state between the probe and ancillary system reaches the optimal sensitivity~\cite{yuan2016sequential,hou2021super}. In the related field of multi-parameter quantum learning, there can be exponential separations between the sample complexity in the sequential ancilla-assisted case and that in the sequential ancilla-free case~\cite{chen2022quantum}. It would be intriguing, in future works, to compare the performance of ancilla-free QEC strategies with the ancilla-assisted QEC strategies~\cite{gorecki2020optimal} in multi-parameter noisy quantum metrology.

\section{\texorpdfstring{$(4, 1)$}{(4,1)}-qutrit code example}
\label{app:qutrit}

In this appendix, we consider a Markovian noise model on a qutrit. We first derive the ancilla-assisted code which entangles one probe qutrit with one ancillary qutrit (\eqref{eq:code-ancilla}) and show it achieves the optimal HL. Then we show that the size of ancilla can be reduced by presenting a code that entangles four probes with one ancillary qutrits and we call it a $(4,1)$-qutrit code, which is a special case of the general small-ancilla code (\eqref{eq:code-few-ancilla}) in the main text. We show the $(4,1)$-qutrit code also achieves the optimal HL using a multi-probe QEC sensing condition (\lemmaref{lemma:multi-probe}). 

\subsection{System evolution and ancilla-assisted code}

Here we consider a qutrit state $\rho$ evolving according to a master equation 
\begin{equation}
\frac{d\rho}{dt} = -i[\omega H,\rho] + L \rho L^\dagger - \frac{1}{2}\{L^\dagger L,\rho\},
\end{equation}
where (we use the $(i+1)$-th row/column to represent a basis state $\ket{i}$ ($i=0,1,2$)) 
\begin{equation}
H = \begin{pmatrix}
1 & 0 & 0 \\
0 & -1 & -1 \\
0 & -1 & -1 \\
\end{pmatrix},\quad 
L = 
\begin{pmatrix}
0 & 1 & 1 \\
0 & 0 & 1 \\
0 & 0 & 0 \\
\end{pmatrix}. 
\end{equation}
The Lindblad span is $\mS = {\rm span}\left\{\id,L+L^\dagger,i(L-L^\dagger),L^\dagger L\right\}$, where 
\begin{equation}
L + L^\dagger = 
\begin{pmatrix}
0 & 1 & 1 \\
1 & 0 & 1 \\
1 & 1 & 0 \\
\end{pmatrix},
\quad 
i(L - L^\dagger) = 
\begin{pmatrix}
0 & i & i \\
-i & 0 & i \\
-i & -i & 0 \\
\end{pmatrix},
\quad    
L^\dagger L = 
\begin{pmatrix}
0 & 0 & 0 \\
0 & 1 & 1 \\
0 & 1 & 2 \\
\end{pmatrix}. 
\end{equation}
We first observe that $H \notin \mS$. Moreover, 
\begin{equation}
\norm{H - \mS} = \norm{H - S} = \frac{1}{2}, \text{~~~where~} 
S = \frac{1}{2}\id - L^\dagger L = 
\text{~and~} H-S = 
\begin{pmatrix}
\frac{1}{2} & 0 & 0 \\
0 & -\frac{1}{2} & 0 \\
0 & 0 & \frac{1}{2} \\
\end{pmatrix}. 
\end{equation}
To find $\rho_{0,1}$ satisfying \eqref{eq:rho-condition-1} and \eqref{eq:rho-condition-2}, we use the fact that $\rho_0$ is supported on the eigenspace corresponding to the maximum eigenvalue of $H-S$ and $\rho_1$ is supported on the eigenspace corresponding to the minimum eigenvalue of $H-S$. Then the only solutions to $\rho_0$ and $\rho_1$ satisfying \eqref{eq:rho-condition-2} are 
\begin{equation}
    \rho_{0} = \frac{1}{2}\left(\ket{0}\bra{0} + \ket{2}\bra{2}\right),\quad 
    \rho_{1} = \ket{1}\bra{1}. 
\end{equation}
Then the ancilla-assisted code is given by 
\begin{equation}
\ket{0_\tl} = \frac{1}{\sqrt{2}}(\ket{0}_P\ket{0}_A+\ket{2}_P\ket{2}_A),\;\;\ket{1_\tl} = \ket{1}_P\ket{1}_A, 
\end{equation}
whose reduced density operators in $\mH_P$ are $\rho_{0,1}$. 
We call it a $(1,1)$-qutrit code, for the encoding contains $1$ qutrit probe and $1$ qutrit ancilla. There exists a recovery channel such that the logical state evolves according to~\cite{zhou2018achieving}
\begin{equation}
    \frac{d\rho_\tl}{dt} = - i \left[\frac{\omega}{2}Z_\tl,\rho_\tl\right]. 
\end{equation}

\subsection{Multi-probe QEC sensing condition}

Before we present the $(4,1)$-qutrit code, we first prove a useful lemma for general QEC sensing that provides a \emph{sufficient} condition for a code to achieve the optimal HL. Note that it is not a necessary condition for multi-probe QEC sensing to achieve the HL (for example, we can replace $Q_k$ with $i Q_k$ below, or introduce more than one $Q_k$ and sum up all $Q_k\otimes Q_k$ terms, and the lemma still holds), but in this paper all our constructed codes follow this sufficient condition (at least approximately). Thus, we will focus on it here for simplicity. 
\begin{lemma}[Multi-probe QEC sensing] 
\label{lemma:multi-probe}
If a two-dimensional QEC code in $\mH_P^{\otimes m} \otimes \mH_A$ ($m\geq 2$) satisfies for $k,k' = 0,1$, $\ell \neq \ell'$ and $\ell,\ell' = 1,\ldots,m$, 
\begin{equation}
\label{eq:2-local}
    \trace_{\backslash\{\ell,\ell'\}}(\ket{k_\tl}\bra{k'_\tl}) = \left(\rho_k^{(\ell)} \otimes \rho_k^{(\ell')} - Q_k^{(\ell)} \otimes Q_k^{(\ell')}\right)\delta_{kk'},
\end{equation}
where $\trace_{\backslash\{\ell,\ell'\}}$ means tracing out the entire system except for the $\ell$-th and $\ell'$-th probes and $Q_k$ can be any traceless Hermitian operator satisfying $\trace(Q_k L_i) = 0, \forall i$, then there exists a QEC strategy such that a logical state $\rho_\tl$ evolves according to (up to an arbitrarily small error)
\begin{equation}
\label{eq:multi-probe-signal}
    \frac{d\rho_\tl}{dt} = -i\left[\omega m \norm{H - \mS} Z_\tl, \rho_\tl\right]. 
\end{equation}
\end{lemma}
\begin{proof}
We will use a result from \cite{zhou2018achieving}, which states that if for a two-dimensional QEC code in $\mH_P \otimes \mH_A$, the QEC condition 
\begin{equation}
    \braket{0_\tl|S|0_\tl} = \braket{1_\tl|S|1_\tl},\quad \braket{0_\tl|S|1_\tl} = 0,\quad \forall S \in \mS = {\rm span}\{\id,L_i^{\dagger},L_i,L_i^{\dagger} L_j,\;\forall i,j\}, 
\end{equation}
is satisfied, then there exists a QEC strategy such that a logical state $\rho_\tl$ evolves as (up to an arbitrarily small error)
\begin{equation}
    \frac{d\rho_\tl}{dt} = -i\left[\omega \frac{\braket{0_\tl|H|0_\tl} - \braket{1_\tl|H|1_\tl}}{2} Z_\tl, \rho_\tl\right]. 
\end{equation}
In the multi-probe case, we only need to replace $\mS$ with a multi-probe Lindblad span $\mS_m$, where 
\begin{equation}
    \mS_m = {\rm span}\{\id,L_i^{\dagger(\ell)},L_i^{(\ell)},L_i^{\dagger(\ell)} L_j^{(\ell')},\;\forall i,j,\ell,\ell'\}, 
\end{equation}
and the single-probe Hamiltonian $H$ with the multi-probe Hamiltonian $\sum_{\ell=1}^m H^{(\ell)}$. Noting that $\mS_m$ contains at most $2$-local operators, i.e., operators that act as identity on all but two probes, and using \eqref{eq:2-local} and 
\begin{equation}
    \trace((Q_0\otimes Q_0)S_2) = \trace((Q_1\otimes Q_1)S_2) = 0,\quad \forall S_2 \in \mS_2 = {\rm span}\{S\otimes \id,\id\otimes S,L_i^\dagger\otimes L_j,\forall i,j,S\in\mS\}
\end{equation}
from the definition of $Q_k$, the multi-probe QEC condition translates to 
\begin{equation}
    \trace((\rho_0\otimes \rho_0)S_2) = \trace((\rho_1\otimes \rho_1)S_2),\quad \forall S_2 \in \mS_2 = {\rm span}\{S\otimes \id,\id\otimes S,L_i^\dagger\otimes L_j,\forall i,j,S\in\mS\}. 
\end{equation}
The above equation holds true because for all $S \in \mS$,
\begin{equation}
    \trace((\rho_0\otimes \rho_0-\rho_1\otimes \rho_1)(S\otimes \id)) = \trace((\rho_0\otimes \rho_0-\rho_1\otimes \rho_1)(\id \otimes S)) = \trace((\rho_0-\rho_1)S) = 0, 
\end{equation}
and for all $i,j$, 
\begin{equation}
    \trace((\rho_0\otimes \rho_0-\rho_1\otimes \rho_1)(L_i^\dagger\otimes L_j)) = 
    \frac{1}{2}\trace(( (\rho_0+\rho_1)\otimes (\rho_0-\rho_1) + (\rho_0-\rho_1)\otimes (\rho_0+\rho_1) )(L_i^\dagger\otimes L_j)) = 0.  
\end{equation}
Finally, we have \eqref{eq:multi-probe-signal}, noticing that 
\begin{equation}
    \frac{1}{2}\left(\braket{0_\tl|\sum_{\ell=1}^m H^{(\ell)}|0_\tl} - \braket{1_\tl|\sum_{\ell=1}^m H^{(\ell)}|1_\tl}\right) = \frac{1}{2} m \trace((\rho_0-\rho_1)H) = m \norm{H-\mS}. 
\end{equation}
\end{proof}

\begin{figure}[tb]
    \centering
    \includegraphics[width=0.5\textwidth]{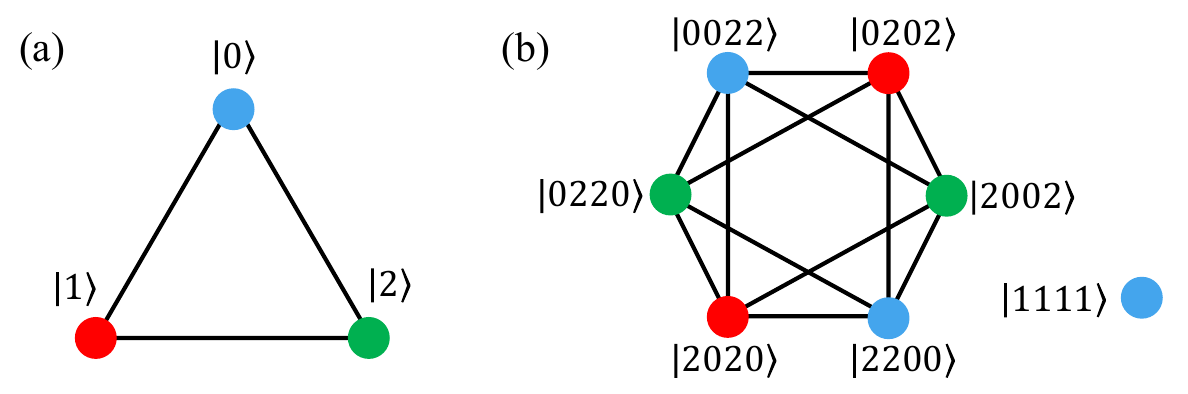}
    \caption{Dimension of ancilla. (a) The dimension of ancilla in the $(1,1)$-qutrit code is $3$. Here we draw a graph where each vertex represents the probe state  $\{\ket{0},\ket{1},\ket{2}\}$ and there is an edge between two different states $\ket{w}$ and $\ket{w'}$ if and only if there exists $S\in\mS$ such that $\braket{w|S|w'} \neq 0$. We use the color of vertices to represent the (orthonormal) basis of the ancilla. We require adjacent vertices to have different colors which is a sufficient condition for QEC. Then at least $3$ colors are required. (b) The dimension of ancilla in the $(4,1)$-qutrit code is $3$. Here we draw a graph where each vertex represents the probe state $\{\ket{0022},\ldots,\ket{1111}\}$ (where $\ldots$ includes all permutations of $\ket{0022}$) and there is an edge between two different states $\ket{w}$ and $\ket{w'}$ if and only if there exists $S_2\in\mS_2$ and $\ell\neq \ell'$ such that $\braket{w|S_2^{(\ell,\ell')}|w'} \neq 0$, where $S_2^{(\ell,\ell')}$ means $S_2$ acting on the $\ell$-th and $\ell'$-th probes (or equivalently, if and only if $w$ can be obtained from $w'$ by swapping two sites with different states). We require adjacent vertices to have different colors which is a sufficient condition for QEC. Then at least $3$ colors are required. }
    \label{fig:qutrit}
\end{figure}

\subsection{\texorpdfstring{$(4, 1)$}{(4,1)}-qutrit code}

We aim at finding a multi-probe QEC code in $\mH_P^{\otimes m}\otimes \mH_A$ that satisfies the sufficient condition in \lemmaref{lemma:multi-probe}. Assume $m = 4$. Since $\rho_1$ is a rank-one operator, we can define 
\begin{equation}
    \ket{1_\tl} = \ket{1111}_{P^{\otimes 4}}\ket{0}_A,
\end{equation}
which satisfies $\trace_{\backslash\{\ell,\ell'\}}(\ket{1_\tl}\bra{1_\tl}) = \rho_1 \otimes \rho_1$ (and $Q_1 = 0$). On the other hand, we define  \begin{equation}
\ket{0_\tl} = \frac{1}{\sqrt{6}} \big(  ( \ket{0022}_{P^{\otimes 4}}+\ket{2200}_{P^{\otimes 4}})\ket{0}_A + (\ket{0202}_{P^{\otimes 4}}+\ket{2020}_{P^{\otimes 4}})\ket{1}_A + (\ket{0220}_{P^{\otimes 4}}+\ket{2002}_{P^{\otimes 4}})\ket{2}_A\big),
\end{equation}
where the numbers of appearances of $\ket{0}$ and $\ket{2}$ are tuned so that $\trace_{\backslash\{\ell,\ell'\}}(\ket{0_\tl}\bra{0_\tl}) = \rho_0 \otimes \rho_0 - Q_0 \otimes Q_0$ and $Q_0 = \frac{1}{\sqrt{12}}(\ket{0}\bra{0}-\ket{2}\bra{2})$. Note that we also use additional ancillas so that $\trace_{\backslash\{\ell,\ell'\}}(\ket{0_\tl}\bra{0_\tl})$ contains no non-diagonal terms $\ket{02}\bra{20}$ and $\ket{20}\bra{02}$. The minimum dimension of ancillas required is $3$, as illustrated in \figref{fig:qutrit}. Finally, we have $\trace_{\backslash\{\ell,\ell'\}}(\ket{0_\tl}\bra{1_\tl}) = 0$ because $\ket{1}$ is perpendicular to ${\rm span}\{\ket{0},\ket{2}\}$. Therefore, \eqref{eq:2-local} is satisfied by the $(4,1)$-qutrit code. Using \lemmaref{lemma:multi-probe}, we obtain the logical dynamics
\begin{equation}
    \frac{d\rho_\tl}{dt} = - i \left[\frac{\omega m}{2}Z_\tl,\rho_\tl\right]. 
\end{equation}

\subsection{Advantage of a small ancilla}

The performance of the $(1,1)$-qutrit code and the $(4,1)$-qutrit code are equally optimal in the following sense: if we have $N$ probes that evolves for a probing time $t$, then the QEC strategy using both encodings (where arbitrary numbers of noiseless ancillas are allowed) can prepare a quantum state $\rho_\omega(t)$ that has QFI (up to an arbitrarily small error)
\begin{equation}
    F(\rho_\omega(t)) = 4\norm{H-\mS}^2N^2t^2  = N^2t^2. 
\end{equation}

In practice, however, the $(4,1)$-qutrit code is superior to the $(1,1)$-qutrit code when the noiseless ancilla is not considered as a free resource, which is true because it is must be built of exceedingly stable qubits or noisy qubits under \sisi{code concatenation}. In order to see the advantage of the $(4,1)$-qutrit code, consider the following overly-simplified but illustrative way of resource counting: a qutrit probe and a qutrit noiseless ancilla are considered as two equivalent units of resource and other resources (e.g., QEC operations) are considered free. Then using the $(1,1)$-qutrit code, the achievable QFI using $N_{\textsc{u}}$ units of resource in time $t$ is (assuming $N_{\textsc{u}}$ is a multiple of $2$) 
\begin{equation}
    F(\rho_\omega(t)) = \frac{1}{4}N_{\textsc{u}}^2t^2, 
\end{equation} 
and the optimal initial state is 
\begin{equation}
    \ket{\psi_{\rm in}} = \frac{1}{\sqrt{2}}\left(\ket{0_\tl}^{\otimes N_{\textsc{u}}/2}+\ket{1_\tl}^{\otimes N_{\textsc{u}}/2}\right). 
\end{equation}
On the other hand, using the $(4,1)$-qutrit code, the achievable QFI using $N_{\textsc{u}}$ units of resource in time $t$ is (assuming $N_{\textsc{u}}$ is a multiple of $5$)  
\begin{equation}
    F(\rho_\omega(t)) = \frac{16}{25}N_{\textsc{u}}^2t^2, 
\end{equation} 
and the optimal initial state is 
\begin{equation}
    \ket{\psi_{\rm in}} = \frac{1}{\sqrt{2}}\left(\ket{0_\tl}^{\otimes 4N_{\textsc{u}}/5}+\ket{1_\tl}^{\otimes 4N_{\textsc{u}}/5}\right). 
\end{equation}
Clearly, the smaller the ratio between the number of noiseless ancillary qutrits and the number of probe qutrits is, the larger QFI of the final state is. The reason is the QFI is a function of the number of probes only and the more probes there are, the larger the QFI will be. In the case where noiseless ancillas and probes are equivalent units of resource, the QFI of ancilla-free strategy is at most 4 times the QFI of ancilla-assisted strategy; while in practice when noiseless ancillas are more expensive than probes, the advantage of ancilla-free strategies will be even more significant.

\section{Upper bound on the dimension of ancilla in small-ancilla codes}
\label{app:graph}

Here we put an upper bound on the highest $\dim(\mH_A)$ needed, the dimension of the noiseless ancilla required in the small-ancilla code (\eqref{eq:code-few-ancilla}). $\dim(\mH_A) = \abs{\frakI}$, where the set $\frakI := \{\fraki_k(w),\forall k=0,1, w\in W_0\cup W_1\}$ should be large enough such that $\fraki_k(w)$ (for $k=0,1$) can be chosen to be any integer functions that satisfy $\fraki_k(w)\neq \fraki_k(w')$ when $w$ and $w'$ are different on exactly two sites. 
Calculating the minimum number of values of $\fraki_0(w)$ required so that $\fraki_0(w)\neq \fraki_0(w')$ when $w$ and $w'$ are different on exactly two sites is essentially a graph coloring problem (see a qutrit example in \figbref{fig:qutrit}), where the elements of $W_0$ are vertices of the graph, and there is an edge between $w$ and $w'$ as long as $w'$ can be obtained from $w$ by swapping two sites with different numbers. We want to know the minimum number of colors required to color the vertices of the graph such that every two adjacent vertices have different colors. Note that the graph is regular---every vertex has the same degree (the number of vertices adjacent to it), which we denote by $\deg(W_0)$. Similar statements also hold for $\fraki_1(w)$. Then, the minimum value of $\dim(\mH_A) = \abs{\frakI}$ required is at most $\max\{\deg(W_0)+1,\deg(W_1)+1\}$, from the Brook's theorem~\cite{brooks1941colouring}. For $k=0$, we have 
\begin{equation}
\deg(W_0)+1 = \sum_{i\neq j \leq d_0 - 1} m_im_j+1 \leq d^2 m^2, 
\end{equation}
and similarly,
\begin{equation}
\deg(W_1)+1 = \sum_{i\neq j \geq d_0} m_im_j+1 \leq d^2 m^2. 
\end{equation}
Therefore we only need an exponentially small ancilla
\begin{equation}
\dim(\mH_A) \leq d^2 m^2 = o(d^m),
\end{equation}
compared to the system of $m$ probes.

\section{Logical master equations: logical Pauli-Z signal under dephasing noise}
\label{app:dephasing}

Consider a logical quantum state $\rho_\tl$ that contains $m$ probes satisfying $\rho_\tl = P \rho_\tl P$ at time $t$. Without QEC, the quantum state at time $t + dt$ satisfies 
\begin{equation}
    \rho_\tl + d\rho_\tl = \rho_\tl 
+ \left( -i\left[\omega \sum_{\ell=1}^m H^{(\ell)},\rho_\tl\right] + \sum_{i=1}^r \sum_{\ell=1}^m L_i^{(\ell)}\rho_\tl L_i^{(\ell) \dagger }
- \frac{1}{2}\left\{L_i^{(\ell)\dagger}L_i^{(\ell)},\rho_\tl\right\}\right) dt + O(dt^2). 
\end{equation}
If we apply the QEC operation $\mP + \mR\circ\mP_\perp$ on $\rho_\tl$ instantaneously at time $t + dt$, we have 
\begin{equation}
    \rho_\tl + d\rho_\tl 
    = \rho_\tl + (\mP + \mR\circ\mP_\perp)\Bigg(-i\left[\omega \sum_{\ell=1}^m H^{(\ell)},\rho_\tl\right] + \sum_{i=1}^r \sum_{\ell=1}^m L_i^{(\ell)}\rho_\tl L_i^{(\ell) \dagger }
- \frac{1}{2}\left\{L_i^{(\ell)\dagger}L_i^{(\ell)},\rho_\tl\right\}\Bigg) dt + O(dt^2). 
\end{equation}
In the limit of infinitely fast QEC (i.e., $dt \rightarrow 0$), we derive~\eqref{eq:QEC-operation}: 
\begin{equation}
    \frac{d\rho_\tl}{dt} = -i\left[\omega \mP\left(\sum_{\ell=1}^m H^{(\ell)}\right),\rho_\tl\right] + \sum_{i=1}^r \sum_{\ell=1}^m \mP\left(L_i^{(\ell)}\rho_\tl L_i^{(\ell) \dagger }\right) 
+ \mR\left(\mP_\perp\left(L_i^{(\ell)}\rho_\tl L_i^{(\ell)\dagger}\right)\right)
- \frac{1}{2}\left\{\mP\left(L_i^{(\ell)\dagger}L_i^{(\ell)}\right),\rho_\tl\right\}. 
\end{equation}
Note that $\mP \circ \mR = \mR$ guarantees that the quantum state always stays in the logical space. 
Furthermore, we assume $\frakL_0 \perp \frakL_1$, where $\frakL_k = {\rm span}\{\ket{k_\tl},L_i^{(\ell)}\ket{k_\tl},\forall i,\ell\}$ for $k=0,1$, and \sisinew{the recovery channel
\begin{equation}
\label{eq:recovery}
    \mR(\cdot) = \sum_p (\ket{0_\tl}\bra{R_p}+\ket{1_\tl}\bra{S_p})(\cdot) (\ket{R_p}\bra{0_\tl}+\ket{S_p}\bra{1_\tl}) 
    +  \trace((\cdot)(\Pi_0-\tPi_R)) \ket{0_\tl}\bra{0_\tl}
    +  \trace((\cdot)(\id - \Pi_0 - \tPi_S)) \ket{1_\tl}\bra{1_\tl} ,
\end{equation}
where $\{\ket{R_p},\forall p\}$, $\{\ket{S_p},\forall p\}$ are two sets of vectors such that $\Pi_0 - \tPi_{R}$ and $\Pi_1 - \tPi_{S}$ are positive semidefinite. We use $\Pi_{0,1}$ to denote projections onto $\frakL_{0,1}$ and we define $\tPi_{R} := \sum_p \ket{R_p}\bra{R_p}$ and $\tPi_{S} := \sum_p \ket{S_p}\bra{S_p}$.}
Then by considering the input and output of the basis of the logical qubit operators $\{\ket{0_\tl}\bra{0_\tl},\ket{0_\tl}\bra{1_\tl},\ket{1_\tl}\bra{0_\tl},\ket{1_\tl}\bra{1_\tl}\}$, we can see that the logical dynamic is described by a dephasing noise and a Pauli-Z rotation, that is, for any $\rho_\tl = P \rho_\tl P$, 
\begin{equation}
\frac{d\rho_\tl}{dt} = -i\left[\frac{\omega \trace(\sum_\ell H^{(\ell)} Z_\tl)}{2} Z_\tl + \frac{\beta_\tl(\mR)}{2} Z_\tl,\rho_\tl\right] + \frac{\gamma_\tl(\mR)}{2}(Z_\tl \rho_\tl Z_\tl - \rho_\tl),
\end{equation}
where $Z_\tl = \ket{0_\tl}\bra{0_\tl} - \ket{1_\tl}\bra{1_\tl}$, 
\begin{equation}
\gamma_\tl(\mR) = -\frac{1}{2}(x+x^*),\quad 
\beta_\tl(\mR) = -\frac{1}{2i}(x-x^*),
\end{equation}
and 
\begin{multline}
x= \sum_{i=1}^r \sum_{\ell=1}^m \bra{0_\tl} L_i^{(\ell)} \ket{0_\tl} \bra{1_\tl} L_i^{(\ell)\dagger} \ket{1_\tl}- \frac{1}{2}(\bra{0_\tl} L_i^{(\ell)\dagger} L_i^{(\ell)} \ket{0_\tl} + \bra{1_\tl} L_i^{(\ell)\dagger} L_i^{(\ell)} \ket{1_\tl}) \\ + \sum_{p}\bra{R_p}(\id - \ket{0_\tl}\bra{0_\tl})(L_i^{(\ell)} \ket{0_\tl}\bra{1_\tl} L_i^{(\ell)\dagger} )(\id - \ket{1_\tl}\bra{1_\tl})\ket{S_p}. 
\end{multline}
In particular, we are going to pick the optimal $\mR$ of the form \eqref{eq:recovery} such that $\gamma_\tl(\mR)$ is minimized, and we define 
\begin{multline}
\gamma_\tl := \min_{\mR  \text{ s.t.~\eqref{eq:recovery}}} \gamma_\tl(\mR) = \sum_{i=1}^r \sum_{\ell=1}^m - \Re[\bra{0_\tl} L_i^{(\ell)} \ket{0_\tl} \bra{1_\tl} L_i^{(\ell)\dagger} \ket{1_\tl}] +  \frac{1}{2}(\bra{0_\tl} L_i^{(\ell)\dagger} L_i^{(\ell)} \ket{0_\tl} + \bra{1_\tl} L_i^{(\ell)\dagger} L_i^{(\ell)} \ket{1_\tl}) \\ - \norm{ \sum_{i=1}^r \sum_{\ell=1}^m (\id - \ket{0_\tl}\bra{0_\tl})(L_i^{(\ell)} \ket{0_\tl}\bra{1_\tl} L_i^{(\ell)\dagger} )(\id - \ket{1_\tl}\bra{1_\tl})}_1, 
\end{multline}
where we use $\max_{U:U^\dagger U = I} \Re[\trace(BU)] = \norm{B}_1$ for any $B$ whose number of columns are no smaller than number of rows.  

\section{Achieving the optimal HL: Proofs and Numerics 
}
\label{app:limit}

\subsection{Overview}

In this appendix, we prove \eqref{eq:limit-1} and \eqref{eq:limit-2} for each code, which implies the achievability of the optimal HL with respect to the number of probes, as discussed in the main text. \sisi{We also present a numerical simulation result at the end of this section, demonstrating the decaying of $\gamma_\tl$ with respect to $m$ in a qutrit example. }

First, we note that that for both codes, 
\begin{gather}
    \lim_{m\rightarrow \infty}\tilde \rho_0 = \sum_{i=0}^{d_0-1} \lim_{m\rightarrow \infty} \frac{m_i}{m}\ket{i}_P\bra{i}_P = \rho_0, \\
    \lim_{m\rightarrow \infty}\tilde \rho_1 = \sum_{i=d_0}^{d-1} \lim_{m\rightarrow \infty} \frac{m_i}{m}\ket{i}_P\bra{i}_P = \rho_1. 
\end{gather}
Then \eqref{eq:limit-1} is proven, using \eqref{eq:rho-condition-1}. 

We now briefly sketch the proof of \eqref{eq:limit-2} and provide the details later. (Note that we will need to use the fact that, for both codes, $\frakL_0 \perp \frakL_1$, where $\frakL_k = {\rm span}\{\ket{k_\tl},L_i^{(\ell)}\ket{k_\tl},\forall i,\ell\}$ for $k=0,1$.) 
Using the gauge constraints (\eqref{eq:gauge-1} and \eqref{eq:gauge-2}), we have 
\begin{equation}
\label{eq:gamma-code}
    \gamma_\tl = m \bigg(\sum_{i=1}^r {\mu}_i - \Re[b_{0,i}^*b_{1,i}] + \frac{1}{2}(a_{0,ii}+a_{1,ii})\bigg) - \norm{B}_1,
\end{equation}
where $b_{k,i} = \braket{k_\tl|L_i^{(\ell)}|k_\tl}$, $a_{k,ij} = \braket{k_\tl|L_i^{(\ell)\dagger}L_j^{(\ell)}|k_\tl} - \mu_{ij}\delta_{ij}$ for all $k=0,1$,  $i,j=1,\ldots,r$, and 
\begin{equation}
\label{eq:B}
B=\sum_{i=1}^r \sum_{\ell=1}^m (L_i^{(\ell)} - b_{0,i})\ket{0_\tl}\bra{1_\tl}(L_i^{(\ell)} - b_{1,i})^\dagger.     
\end{equation}
The second and third terms in \eqref{eq:gamma-code} contribute to $O(1)$ in $\gamma_\tl$ because $b_{k,i}$ and $a_{k,ij}$ are $O(1/m)$ from \eqref{eq:rounding}. 
\sisilong{We will prove later that the fourth term $\norm{B}_1$ is equal to $m \sum_{i=1}^r {\mu}_i$ up to a constant error, i.e, we will have the following theorem 
\begin{theorem}[Constant error rate]
\label{thm:lower}
For the small-ancilla code, the operator $B$ defined in \eqref{eq:B} satisfies 
\begin{equation}
\label{eq:lower-3}
    \norm{B}_1 \geq m \sum_{i=1}^r {\mu}_i - O(1), 
\end{equation} 
which implies $\gamma_\tl = O(1)$ for the small-ancilla code. 
For the ancilla-free random code, the above inequality holds true with probability $1 - e^{-\Omega(m)}$, which implies $\gamma_\tl = O(1)$ for the ancilla-free random code with probability $1 - e^{-\Omega(m)}$.
\end{theorem}

Below, we will prove \eqref{eq:lower-3} for the small-ancilla code in \appref{app:limit-1} and for the ancilla-free random code in \appref{app:limit-2}, respectively, which complete the proof of \thmref{thm:lower}. In these proofs,  
we will use the H\"older's inequality~\cite{baumgartner2011inequality} 
\begin{equation}
    \norm{B}_1 \geq \frac{\abs{\trace(B V)}}{\norm{V}_\infty},
\end{equation} 
where $\norm{V}_\infty$ is the spectral norm (maximum singular value) of $V$ and we will find an approximately unitary matrix $V$ such that 
\begin{equation}
    \abs{\trace(B V)} = m \sum_{i=1}^r {\mu}_i - O(1)  \text{~~and~~} \norm{V}_\infty = 1 + O(1/m).  
\end{equation}

Note that we use some of the proof techniques we use in \appref{app:limit-2} and  \appref{app:limit-1} are identical. However, we repeat them for the completeness and self-consistency of the proofs.}

\subsection{Proof of \texorpdfstring{\eqref{eq:lower-3}}{Eq.(E5)} for the small-ancilla code}
\label{app:limit-1}

The small-ancilla code (\eqref{eq:code-few-ancilla}) is symmetric (i.e., permutation-invariant) in the probe system. For all $\ell$ and $\ell' \neq \ell$, we have  
\begin{align}
\label{eq:C2}
b_{0,i} &:= \braket{0_\tl|L_i^{(\ell)}|0_\tl} 
= \trace\left(\sum_{k=0}^{d_0-1} \left(\frac{m_k}{m}-\lambda_k\right) \ket{k}\bra{k} \cdot L_i\right) = O\left(\frac{1}{m}\right),\\
\label{eq:C3}
b_{1,i} &:= \braket{1_\tl|L_i^{(\ell)}|1_\tl} 
= \trace\left(\sum_{k=d_0}^{d-1} \left(\frac{m_k}{m}-\lambda_k\right) \ket{k}\bra{k} \cdot L_i\right) = O\left(\frac{1}{m}\right),\\
\label{eq:C4}
a_{0,ij} &:= \braket{0_\tl|L_i^{(\ell)\dagger}L_j^{(\ell)}|0_\tl} - \mu_{i}\delta_{ij} = \trace\left(\sum_{k=0}^{d_0-1} \left(\frac{m_k}{m}-\lambda_k\right) \ket{k}\bra{k} \cdot L_i^{\dagger}L_j\right)   =  O\left(\frac{1}{m}\right),\\
\label{eq:C5}
a_{1,ij} &:= \braket{1_\tl|L_i^{(\ell)\dagger}L_j^{(\ell)}|1_\tl} - \mu_{i}\delta_{ij} = \trace\left(\sum_{k=d_0}^{d-1} \left(\frac{m_k}{m}-\lambda_k\right) \ket{k}\bra{k} \cdot L_i^{\dagger}L_j\right) =  O\left(\frac{1}{m}\right),\\
\label{eq:C6}
\eta_{0,ij} &:= \braket{0_\tl|L_{i}^{(\ell)\dagger}L_{j}^{(\ell')}|0_\tl} = \frac{m b_{0,i}^*b_{0,j}}{m-1} - \sum_{k=0}^{d_0-1} \frac{m_k}{m(m-1)} \bra{k}L_{i}^{\dagger}\ket{k}\bra{k}L_{j}\ket{k} = O\left(\frac{1}{m}\right),\\
\label{eq:C7}
\eta_{1,ij} &:= \braket{1_\tl|L_{i}^{(\ell)\dagger}L_{j}^{(\ell')}|1_\tl} = \frac{m b_{1,i}^*b_{1,j}}{m-1} - \sum_{k=d_0}^{d-1} \frac{m_k}{m(m-1)} \bra{k}L_{i}^{\dagger}\ket{k}\bra{k}L_{j}\ket{k} = O\left(\frac{1}{m}\right). 
\end{align}
To derive the equations above, we use the gauge constraints (\eqref{eq:gauge-1} and \eqref{eq:gauge-2}) and the fact that (we will omit the subscript ${}_P$ denoting the probe systems in this appendix)
\begin{equation}
    \trace_{\backslash\{\ell,\ell'\}}(\ket{0_\tl}\bra{0_\tl}) = \sum_{i,j=0}^{d_0-1} m_{ij} \ket{ij}\bra{ij}, \quad \trace_{\backslash\{\ell,\ell'\}}(\ket{1_\tl}\bra{1_\tl}) = \sum_{i,j=d_0}^{d-1} m_{ij} \ket{ij}\bra{ij}, 
\end{equation}
where 
\begin{equation}
m_{ij} = 
\frac{1}{m(m-1)} \times
\begin{cases}
m_i m_j, & i\neq j,\\
m_i (m_i - 1),& i = j. 
\end{cases}
\end{equation}
In particular, it implies 
\begin{equation}
\label{eq:1-local}
        \tilde\rho_0 = \trace_{\backslash\{\ell\}}(\ket{0_\tl}\bra{0_\tl}) = \sum_{i=0}^{d_0-1} \frac{m_{i}}{m} \ket{i}\bra{i}, \quad 
        \tilde\rho_1 =  \trace_{\backslash\{\ell\}}(\ket{1_\tl}\bra{1_\tl}) = \sum_{i=d_0}^{d-1} \frac{m_{i}}{m} \ket{i}\bra{i}, 
\end{equation}

Note that the multi-probe QEC condition (\eqref{eq:sensing-condition-2}, taking $Q_k = 0$) approximately holds true because 
\begin{equation}
    \lim_{m\rightarrow \infty}\trace_{\backslash\{\ell,\ell'\}}(\ket{0_\tl}\bra{0_\tl}) = \rho_0 \otimes \rho_0, \quad \lim_{m\rightarrow \infty}\trace_{\backslash\{\ell,\ell'\}}(\ket{1_\tl}\bra{1_\tl}) = \rho_1 \otimes \rho_1,
\end{equation}
However, for a finite $m$, the multi-probe QEC condition is only approximate true (when choosing $Q_k = 0$) and $\gamma_\tl$ can be positive (except for special cases where all parameters $b_{k,i}$, $a_{k,ij}$ and $\eta_{k,ij}$ are exactly zero and then $\gamma_\tl = 0$, e.g., the repetition code and the $(4,1)$-qutrit code). We will now analyze the scaling of $\gamma_\tl$ with respect to $m$. 

As argued in the previous discussion, to show $\gamma_\tl = O(1)$, our goal is to prove $\norm{B}_1 \geq m \sum_{i=1}^r {\mu}_i - O(1)$. 
One possible way to put a lower bound on $B$ is to find a matrix $V$ and use 
\begin{equation}
\label{eq:lower}
\norm{B}_1 \geq \frac{|\trace(BV)|}{\norm{V}_\infty}. 
\end{equation}
\sisinewlong{
In particular, we will choose 
\begin{equation}
V = \sum_{i=1}^r \sum_{\ell=1}^{m-1} 
\ket{f_{i,\ell}} \bra{e_{i,\ell}},
\end{equation}
where for $i = 1,\ldots,r$ and $\ell = 1,\ldots, m-1$, 
\begin{equation}
\ket{e_{i,\ell}} = \frac{1}{\sqrt{m}}\sum_{\ell'=1}^{m} \exp\left(-i \frac{2\pi}{m}\ell \ell'\right) \ket{\widehat{J}_{0,i}^{(\ell')}},\quad 
\ket{f_{i,\ell}} = \frac{1}{\sqrt{m}}\sum_{\ell'=1}^{m} \exp\left(-i \frac{2\pi}{m}\ell \ell'\right) \ket{\widehat{J}_{1,i}^{(\ell')}},
\end{equation}
and for $i = 1,\ldots,r$, $\ell = 1,\ldots, m$ and $k = 0,1$, 
\begin{equation}
\ket{\widehat{J}_{k,i}^{(\ell)}} = \frac{1}{\sqrt{\mu_i}}(L_i^{(\ell)} - b_{k,i})\ket{k_\tl}. 
\end{equation}
Intuitively, $\{\ket{\widehat{J}_{0,i}^{(\ell)}}\}$ and $\{\ket{\widehat{J}_{1,i}^{(\ell)}}\}$ are two sets of approximately orthonormal states when $m \gg 1$. They satisfy, for all $\ell$, 
\begin{equation}
\label{eq:condition-J-small-ancilla-1}
\braket{\widehat{J}_{k,i}^{(\ell)}|\widehat{J}_{k,j}^{(\ell)}}  = \delta_{ij} - \frac{b_{k,i}^*b_{k,j}}{\sqrt{\mu_i\mu_j}} = \delta_{ij} + O\left(\frac{1}{m^{2}}\right). 
\end{equation}
and for $\ell\neq\ell'$, 
\begin{equation}
\label{eq:condition-J-small-ancilla-2}
\braket{\widehat{J}_{k,i}^{(\ell)}|\widehat{J}_{k,j}^{(\ell')}} = \frac{1}{\sqrt{\mu_i\mu_j}} \big(\eta_{k,ij} - {b_{k,i}^*b_{k,j}}\big) = O\left(\frac{1}{m}\right).
\end{equation}
$\{\ket{e_{i,\ell}}\}$ and $\{\ket{f_{i,\ell}}\}$ are Fourier transforms of $\{\ket{\widehat{J}_{0,i}^{(\ell)}}\}$ and $\{\ket{\widehat{J}_{1,i}^{(\ell)}}\}$, except that we do not define $\ket{e_{i,\ell}}$ and $\ket{f_{i,\ell}}$ when $\ell = m$ (while traditionally, the $\ell = m$ terms are also used in Fourier transform). They satisfy, for all $1 \leq \ell \leq m - 1$,  
\begin{equation}
\braket{e_{i,\ell}|e_{j,\ell'}}  = \delta_{\ell\ell'} \left(\delta_{ij} - \frac{1}{\sqrt{\mu_i\mu_j}} \eta_{0,ij}\right),\quad 
\braket{f_{i,\ell}|f_{j,\ell'}}  = \delta_{\ell\ell'} \left(\delta_{ij} - \frac{1}{\sqrt{\mu_i\mu_j}} \eta_{1,ij}\right).
\end{equation}
Note that for $\ell \neq \ell'$ and any $i,j$, unlike $\ket{\widehat{J}_{k,i}^{(\ell)}}$ and $\ket{\widehat{J}_{k,j}^{(\ell')}}$ which are approximately orthogonal to each other, $\ket{e_{i,\ell}}$ and $\ket{e_{j,\ell'}}$ (or $\ket{f_{i,\ell}}$ and $\ket{f_{j,\ell'}}$) are strictly orthogonal to each other, which is a key property we will use later. 

Our goal is to prove that
\begin{lemma}
\label{lemma:V-small-ancilla}
For the small-ancilla code, $B=\sum_{i=1}^r \sum_{\ell=1}^m (L_i^{(\ell)} - b_{0,i})\ket{0_\tl}\bra{1_\tl}(L_i^{(\ell)} - b_{1,i})^\dagger$ and $V = \sum_{i=1}^r \sum_{\ell=1}^{m-1} 
\ket{f_{i,\ell}} \bra{e_{i,\ell}}$ defined above satisfy 
\begin{equation}
\label{eq:lower-2}
\abs{\trace(BV)} = m \sum_{i=1}^r \mu_i + O(1), \quad \text{and}\quad \norm{V}_\infty \leq 1 + O\left(\frac{1}{m}\right). 
\end{equation}
Then 
\begin{equation}
    \norm{B}_1 \geq \frac{\abs{\trace(BV)}}{\norm{V}_\infty} \geq m \sum_{i=1}^r \mu_i + O(1). 
\end{equation}
\end{lemma}

\begin{proof}

First, by direct calculations and using the definitions of $\ket{e_{i,\ell}}$ and $\ket{f_{i,\ell}}$, we have 
\begin{align}
V 
= \sum_{i=1}^{r} \sum_{\ell=1}^{m-1} \ket{f_{i,\ell}}\bra{e_{i,\ell}} 
= \sum_{i=1}^{r} \left( \sum_{\ell=1}^m \frac{m-1}{m}\ket{\widehat{J}_{1,i}^{(\ell)}}  \bra{\widehat{J}_{0,i}^{(\ell)}} + \sum_{\ell\neq \ell',\ell,\ell'=1}^m \frac{-1}{m}\ket{\widehat{J}_{1,i}^{(\ell)}}  \bra{\widehat{J}_{0,i}^{(\ell')}} \right).
\end{align}
Furthermore, from 
\eqref{eq:condition-J-small-ancilla-1} and 
\eqref{eq:condition-J-small-ancilla-2}, we have 
\begin{equation}
    \sum_{i=1}^r\bra{\widehat{J}_{1,i}^{(\ell)}}  B \ket{\widehat{J}_{0,i}^{(\ell)}}
    = \sum_{i} \mu_i + O\left(\frac{1}{m}\right)\quad
    \text{and}\quad
    \sum_{i=1}^r\bra{\widehat{J}_{1,i}^{(\ell)}}  B \ket{\widehat{J}_{0,i}^{(\ell')}}
    = O\left(\frac{1}{m}\right),\quad \text{when }\ell\neq\ell'.
\end{equation}
Then we have 
\begin{align}
\label{eq:B-real}
\trace(BV) = \sum_{i,\ell} \bra{e_{i,\ell}} B \ket{f_{i,\ell}}= m \sum_{i} \mu_i + O(1),
\end{align}
proving the first equation in \eqref{eq:lower-2}. 

Meanwhile, 
\begin{align}
\norm{V}_\infty^2 &= \big\|V^\dagger V\big\|_\infty = \bigg\|\sum_{i,\ell,j,\ell'} \ket{e_{i,\ell}}\braket{f_{i,\ell}|f_{j,\ell'}}\bra{e_{j,\ell'}}\bigg\|_\infty\\
&= \bigg\|\sum_{i,j,\ell} \left(\delta_{ij} - \frac{1}{\sqrt{\mu_i\mu_j}} \eta_{1,ij}\right) \ket{e_{i,\ell}}\bra{e_{j,\ell}}\bigg\|_\infty\\
&\leq \max_{\ell \in [1,m]} \left( \Big\|\sum_{i} \ket{e_{i,\ell}}  \bra{e_{i,\ell}}\Big\|_\infty + \Big\|\sum_{i,j} \frac{\eta_{1,ij}}{\sqrt{\mu_i\mu_{j}}} \ket{e_{i,\ell}}  \bra{e_{j,\ell}}\Big\|_\infty \right), 
\end{align}
where in the last step we use the fact that $\ket{e_{i,\ell}}$ and $\ket{e_{j,\ell'}}$ are strictly orthogonal to each other for any $\ell \neq \ell'$ and $i,j$. These two terms (for any $\ell$) can be bounded through the following. Choosing an arbitrary orthonormal set of vectors $\{\ket{\widehat{\v{e}}_i}\}_{i=1}^m$, the first term satisfies
\begin{align}
\Big\|\sum_{i} \ket{e_{i,\ell}}  \bra{e_{i,\ell}}\Big\|_\infty 
&= \Big\|\Big(\sum_{i} \ket{e_{i,\ell}}\bra{\widehat{\v{e}}_i}\Big)\Big(\sum_{j} \ket{\widehat{\v{e}}_j}\bra{e_{j,\ell}}\Big)\Big\|_\infty \\
&= \Big\| \Big(\sum_{i} \ket{\widehat{\v{e}}_i}\bra{e_{i,\ell}}\Big)\Big(\sum_{j} \ket{e_{j,\ell}}\bra{\widehat{\v{e}}_j}\Big) \Big\|_\infty \\
&= \Big\| \sum_{ij} \left( \delta_{ij} - \frac{\eta_{0,ij}}{\sqrt{\mu_i\mu_j}} \right) \ket{\widehat{\v{e}}_i}\bra{\widehat{\v{e}}_j}\Big\|_\infty \leq 1 + O\left(\frac{1}{m}\right),\label{eq:approximate-proj}
\end{align}
where we use $\norm{A^\dagger A}_\infty = \norm{AA^\dagger}_\infty$ in the second step. The second term satisfies 
\begin{align}
\Big\|\sum_{i,j} \frac{\eta_{1,ij}}{\sqrt{\mu_i\mu_{j}}} \ket{e_{i,\ell}}  \bra{e_{i',\ell}}\Big\|_\infty 
&= \max_{\text{unit vector }\ket{v}}  \abs{ \sum_{i,j} \frac{\eta_{1,ij}}{\sqrt{\mu_i\mu_{j}}} \braket{v|e_{i,\ell}}  \braket{e_{j,\ell}|v} } \\
&\leq \sum_{i,j} \abs{ \frac{\eta_{1,ij}}{\sqrt{\mu_i\mu_{j}}} } \abs{\delta_{ij} - \frac{1}{\sqrt{\mu_i\mu_j}} \eta_{0,ij} } = O\left(\frac{1}{m}\right). 
\end{align}
Therefore, 
\begin{equation}
\norm{V}_\infty^2 \leq 1 + O\left(\frac{1}{m}\right),
\end{equation}
proving the second inequality in \eqref{eq:lower-2}. 

\end{proof}

To conclude, now we have, from \eqref{eq:lower-2}, $\norm{B}_1 = m \sum_i \mu_i  + O(1)$. Then 
\begin{equation}
\gamma_\tl
= m \sum_{i=1}^r {\mu}_i - \Re[b_{0,i}^*b_{1,i}] + \frac{1}{2}(a_{0,ii}+a_{1,ii}) - \norm{B}_1 =  m \sum_{i=1}^r {\mu}_i - \norm{B}_1  + O(1) = O(1). 
\end{equation}

Finally, we remark that the proof above implicitly provides one construction of an asymptotically optimal recovery channel $\mR_{\rm opt}$ of the form in \eqref{eq:recovery} for achieving the optimal HL using the small-ancilla code, where $p = (i,\ell)$ for $1 \leq i \leq r$ and $1\leq \ell \leq m-1$, 
\begin{equation}
\label{eq:optimal-recovery}
\ket{R_{(i,\ell)}} = \frac{\ket{e_{i,\ell}}}{\sqrt{\norm{\sum_{j,k}\ket{e_{j,k}}\bra{e_{j,k}}}_\infty}},\quad 
\ket{S_{(i,\ell)}} = \frac{\ket{f_{i,\ell}}}{\sqrt{\norm{\sum_{j,k}\ket{f_{j,k}}\bra{f_{j,k}}}_\infty}}.
\end{equation}
From the proof above, we have $\Re[\trace(BV)] = m\sum_{i} \mu_i + O(1)$ from \eqref{eq:B-real}. Moreover, using exactly the same technique in proving $\norm{V}_\infty = \norm{\sum_{i=1}^r \sum_{\ell=1}^{m-1} \ket{e_{i,\ell}}\bra{f_{i,\ell}}}_\infty \leq 1 + O(1/m)$, we can also prove 
\begin{gather}
\bigg\| \sum_{i=1}^r \sum_{\ell=1}^{m-1} \ket{e_{i,\ell}}\bra{e_{i,\ell}}\bigg\|_\infty \leq 1 + O\left(\frac{1}{m}\right),\label{eq:ee-matrix-upper}\\
\bigg\| \sum_{i=1}^r \sum_{\ell=1}^{m-1} \ket{f_{i,\ell}}\bra{f_{i,\ell}}\bigg\|_\infty \leq 1 + O\left(\frac{1}{m}\right).\label{eq:ff-matrix-upper}
\end{gather}
One can then verify that 
\begin{align}
    \gamma_\tl(\mR_{\rm opt}) 
    &= \sum_{i=1}^r \sum_{\ell=1}^m - \Re[\bra{0_\tl} L_i^{(\ell)} \ket{0_\tl} \bra{1_\tl} L_i^{(\ell)\dagger} \ket{1_\tl}] + \frac{1}{2}(\bra{0_\tl} L_i^{(\ell)\dagger} L_i^{(\ell)} \ket{0_\tl} + \bra{1_\tl} L_i^{(\ell)\dagger} L_i^{(\ell)} \ket{1_\tl})\\  & \qquad \qquad \qquad \qquad  - \sum_{i,\ell}\Re[\bra{R_{(i,\ell)}}(\id - \ket{0_\tl}\bra{0_\tl})(L_i^{(\ell)} \ket{0_\tl}\bra{1_\tl} L_i^{(\ell)\dagger} )(\id - \ket{1_\tl}\bra{1_\tl})\ket{S_{(i,\ell)}}]\\
    &= m \sum_{i=1}^r {\mu}_i - \Re[b_{0,i}^*b_{1,i}] + \frac{1}{2}(a_{0,ii}+a_{1,ii}) - \frac{\Re[\trace(BV)]}{\norm{\sum_{j,k}\ket{e_{j,k}}\bra{e_{j,k}}}_\infty \norm{\sum_{j,k}\ket{f_{j,k}}\bra{f_{j,k}}}_\infty}\\
    &\leq m \sum_{i=1}^r {\mu}_i + O(1) - \frac{m\sum_{i} \mu_i + O(1)}{1 + O(1/m)} = O(1). 
\end{align}
}

\subsection{Proof of \texorpdfstring{\eqref{eq:lower-3}}{Eq.(E5)} for the ancilla-free random code}
\label{app:limit-2}

For the ancilla-free random code, we again define for all $\ell$ and $\ell' \neq \ell$: 
\begin{align}
\label{eq:D2}
b_{0,i} &:= \braket{0_\tl|L_i^{(\ell)}|0_\tl} 
= \trace\left(\sum_{k=0}^{d_0-1} \left(\frac{m_k}{m}-\lambda_k\right) \ket{k}\bra{k} \cdot L_i\right) = O\left(\frac{1}{m}\right),\\
\label{eq:D3}
b_{1,i} &:= \braket{1_\tl|L_i^{(\ell)}|1_\tl} 
= \trace\left(\sum_{k=d_0}^{d-1} \left(\frac{m_k}{m}-\lambda_k\right) \ket{k}\bra{k} \cdot L_i\right)  = O\left(\frac{1}{m}\right),\\
\label{eq:D4}
a_{0,ij} &:= \braket{0_\tl|L_i^{(\ell)\dagger}L_j^{(\ell)}|0_\tl} - \mu_{i}\delta_{ij} = \trace\left(\sum_{k=0}^{d_0-1} \left(\frac{m_k}{m}-\lambda_k\right) \ket{k}\bra{k} \cdot L_i^{\dagger}L_j\right) = O\left(\frac{1}{m}\right),\\
\label{eq:D5}
a_{1,ij} &:= \braket{1_\tl|L_i^{(\ell)\dagger}L_j^{(\ell)}|1_\tl} - \mu_{i}\delta_{ij} = \trace\left(\sum_{k=d_0}^{d-1} \left(\frac{m_k}{m}-\lambda_k\right) \ket{k}\bra{k} \cdot L_i^{\dagger}L_j\right) = O\left(\frac{1}{m}\right),\\
\label{eq:D6}
\begin{split}
\eta^{(\ell,\ell')}_{0,ij} &:= \braket{0_\tl|L_{i}^{(\ell)\dagger}L_{j}^{(\ell')}|0_\tl} \!=\! \frac{m b_{0,i}^*b_{0,j}}{m-1} \!-\! \sum_{k=0}^{d_0-1} \frac{m_k}{m(m-1)} \bra{k}L_{i}^{\dagger}\ket{k}\bra{k}L_{j}\ket{k} \! +\! \sum_{\substack{k\neq k',\\ k,k'=0}}^{d_0-1} \Delta^{(\ell,\ell')}_{0,kk'} \bra{k'}L_{i}^{\dagger}\ket{k}\bra{k}L_{j}\ket{k'}\\
&=: \eta_{0,ij} + \teta^{(\ell,\ell')}_{0,ij} = O\left(\frac{1}{m}\right) + \teta^{(\ell,\ell')}_{0,ij},
\end{split}\\
\label{eq:D7}
\begin{split}
\eta^{(\ell,\ell')}_{1,ij} &:= \braket{1_\tl|L_{i}^{(\ell)\dagger}L_{j}^{(\ell')}|1_\tl} \!=\! \frac{m b_{1,i}^*b_{1,j}}{m-1} \!-\! \sum_{k=d_0}^{d-1} \frac{m_k}{m(m-1)} \bra{k}L_{i}^{\dagger}\ket{k}\bra{k}L_{j}\ket{k}  \!+\! \sum_{\substack{k\neq k',\\ k,k'=d_0}}^{d-1} \Delta^{(\ell,\ell')}_{1,kk'} \bra{k'}L_{i}^{\dagger}\ket{k}\bra{k}L_{j}\ket{k'}\\
&=: \eta_{1,ij} + \teta^{(\ell,\ell')}_{1,ij} = O\left(\frac{1}{m}\right) + \teta^{(\ell,\ell')}_{1,ij},
\end{split}
\end{align}
where $\Delta^{(\ell,\ell')}_{0,kk'}$ and $\Delta^{(\ell,\ell')}_{1,kk'}$ are defined through
\begin{gather}
\label{eq:def-Delta}
    \trace_{\backslash\{\ell,\ell'\}}(\ket{0_\tl}\bra{0_\tl}) = \sum_{i,j=0}^{d_0-1} m_{ij} \ket{ij}\bra{ij} + \sum_{\substack{i\neq j,\\ i,j=0}}^{d_0-1}\Delta^{(\ell,\ell')}_{0,ij} \ket{ij}\bra{ji},\\ 
    \trace_{\backslash\{\ell,\ell'\}}(\ket{1_\tl}\bra{1_\tl}) = \sum_{i,j=d_0}^{d-1} m_{ij} \ket{ij}\bra{ij} + \sum_{\substack{i\neq j,\\i,j=d_0}}^{d-1} \Delta^{(\ell,\ell')}_{1,ij} \ket{ij}\bra{ji}.  
\end{gather}
Note that \eqref{eq:D2}-\eqref{eq:D5} and  \eqref{eq:C2}-\eqref{eq:C5} are identical; while \eqref{eq:D6} and \eqref{eq:D7} have additional terms $\teta^{(\ell,\ell')}_{0,ij}$, $\teta^{(\ell,\ell')}_{1,ij}$ compared to \eqref{eq:C6} and \eqref{eq:C7} that are functions of $\ell,\ell'$ because the ancilla-free random code may not be symmetric like the small-ancilla code. 

\sisinewlong{
In order to show $\norm{B}_1 \geq m \sum_{i=1}^r {\mu}_i - O(1)$ with high probability, we first prove \lemmaref{lemma:Delta} which shows that $\Delta^{(\ell,\ell')}_{k,ij}$ is small with high probability. Using \lemmaref{lemma:Delta}, we will later show in \lemmaref{lemma:V-ancilla-free} that the contribution from $\Delta^{(\ell,\ell')}_{k,ij}$ is negligible with high probability and prove \eqref{eq:lower-3}. 

\begin{lemma}
\label{lemma:Delta}
$\Delta^{(\ell,\ell')}_{0,ij}$ (or $\Delta^{(\ell,\ell')}_{1,ij}$) defined above for all $\ell \neq \ell'$ and $i\neq j \in [0,d_0-1]$ (or $i\neq j \in [d_0,d-1]$) satisfies
\begin{equation}
    \prob\left(\abs{\Delta_{k,ij}^{(\ell,\ell')}} > \frac{1}{m^3},\exists k,i,j,\ell,\ell'\right) 
    = e^{-\Omega(m)}.
\end{equation}
\end{lemma}

\begin{proof}

First, note that if $d_0 = 1$, the codeword $\ket{0_\tl} = \ket{0}^{\otimes m}$ is the same as in the small-ancilla code and $\Delta_{0,ij}^{(\ell,\ell')} = 0$. Now we assume $d_0 > 1$. Then 
from
\begin{gather}
    \trace_{\backslash\{\ell,\ell'\}}(\ket{0_\tl}\bra{0_\tl}) = \sum_{i,j=0}^{d_0-1} m_{ij} \ket{ij}\bra{ij} + \sum_{\substack{i\neq j,\\i,j=0}}^{d_0-1}\Delta^{(\ell,\ell')}_{0,ij} \ket{ij}\bra{ji},
\end{gather}
we have
\begin{equation}
    \Delta^{(\ell,\ell')}_{0,ij} = \frac{1}{\abs{W_0}}\sum_{\widetilde{w} \in \widetilde{W}_{0,ij}} e^{-i \theta_{ij\widetilde{w}}} e^{i \theta_{ji\widetilde{w}}},
\end{equation}
and $\widetilde{W}_{0,ij}$  
for $i\neq j$ and $i,j \in \{0,\cdots,d_0-1\}$ 
is the set of strings of length $m-2$ which contains $m_i-1$ $i$'s, $m_j-1$ $j$'s and $m_{l}$ $l$'s for all $l \neq i,j$ and $l \in \{0,\cdots,d_0-1\}$.
We have
\begin{gather}
    \bE[\Delta^{(\ell,\ell')}_{0,ij}] = \frac{1}{\abs{W_0}}\sum_{\widetilde{w} \in \widetilde{W}_{0,ij}} \bE[e^{-i \theta_{ij\widetilde{w}}}] \bE[e^{i \theta_{ji\widetilde{w}}}] = 0, 
    \\
    \bE[\abs{\Delta^{(\ell,\ell')}_{0,ij}}^2] = \frac{1}{\abs{W_0}^2}\sum_{\widetilde{w},\widetilde{w}' \in \widetilde{W}_{0,ij}} \bE[e^{-i \theta_{ij\widetilde{w}}}e^{i \theta_{ji\widetilde{w}}}e^{i \theta_{ij\widetilde{w}'}}e^{-i \theta_{ji\widetilde{w}'}}]  = \frac{\abs{\widetilde{W}_{0,ij}}}{\abs{W_0}^2} \leq \frac{1}{\abs{W_0}}. 
\end{gather}
Note that $\abs{W_0} = \frac{m!}{m_0! \cdots m_{d_0-1}!} = e^{\Omega(m)}$. 
Using the Chebyshev's inequality, we have for all $i\neq j$, 
$    \prob\left(\Re[\Delta^{(\ell,\ell')}_{0,ij}] > \frac{1}{\sqrt{2} m^3}\right)  \leq\frac{2m^6}{\abs{W_0}}$, 
$    \prob\left(\Im[\Delta^{(\ell,\ell')}_{0,ij}] > \frac{1}{\sqrt{2}m^3}\right)\leq \frac{2m^6}{\abs{W_0}}$.
Then the union bound implies, 
\begin{equation}
    \prob\left(\abs{\Delta^{(\ell,\ell')}_{0,ij}} > \frac{1}{m^3}\right) \leq
    \begin{cases}
        \frac{4m^6}{\abs{W_0}} = e^{-\Omega(m)} & d_0 > 1,\\ 
         0  & d_0 = 1. 
    \end{cases}
\end{equation}
Similarly, 
we can show that 
\begin{equation}
    \prob\left(\abs{\Delta^{(\ell,\ell')}_{1,ij}} > \frac{1}{m^3}\right) \leq
    \begin{cases}
        \frac{4m^6}{\abs{W_1}} = e^{-\Omega(m)} & d-d_0 > 1,\\ 
         0  & d-d_0 = 1. 
    \end{cases}
\end{equation}
The total number of different $(k,i,j,\ell,\ell')$ are 
\begin{equation}
    \frac{m(m-1)}{2}{d_0(d_0-1)}{} + \frac{m(m-1)}{2}{(d-d_0)(d-d_0-1)}{} \leq m^2 d^2. 
\end{equation}
Using the union bound again, we have 
\begin{equation}
    \prob\left(\abs{\Delta_{k,ij}^{(\ell,\ell')}} > \frac{1}{m},\exists k,i,j,\ell,\ell'\right) \leq \begin{cases}
        \frac{4m^2}{\min\{\abs{W_0},\abs{W_1}\}}m^6d^2 = e^{-\Omega(m)} & d_0 > 1,\; d-d_0 > 1, \\
        \frac{4m^2}{\abs{W_1}\}}m^6d^2 = e^{-\Omega(m)} & d_0 = 1,\; d-d_0 > 1,\\
        \frac{4m^2}{\abs{W_0}\}}m^6d^2 = e^{-\Omega(m)} & d_0 > 1,\; d-d_0 = 1,\\ 
        0 & d_0 = d-d_0 = 1.\\
    \end{cases}
\end{equation}
\end{proof}

Now we prove \lemmaref{lemma:V-ancilla-free} which implies \eqref{eq:lower-3}, adopting a similar approach as we used in \lemmaref{lemma:V-small-ancilla}.

\begin{lemma}
\label{lemma:V-ancilla-free}
For the ancilla-free random code and $B=\sum_{i=1}^r \sum_{\ell=1}^m (L_i^{(\ell)} - b_{0,i})\ket{0_\tl}\bra{1_\tl}(L_i^{(\ell)} - b_{1,i})^\dagger$, there exists a matrix $V$ such that 
\begin{equation}
\label{eq:lower-4}
\abs{\trace(BV)} = m \sum_{i=1}^r \mu_i + O(1), \quad \text{and}\quad \norm{V}_\infty \leq 1 + O\left(\frac{1}{m}\right),
\end{equation}
with probability $1 - e^{-\Omega(m)}$. Then 
\begin{equation}
    \norm{B}_1 \geq \frac{\abs{\trace(BV)}}{\norm{V}_\infty} \geq m \sum_{i=1}^r \mu_i + O(1). 
\end{equation}
\end{lemma}

\begin{proof}
Here we assume for all $k,i\neq j,\ell\neq\ell'$, 
\begin{equation}
    \abs{\Delta_{k,ij}^{(\ell,\ell')}} < \frac{1}{m^3},
\end{equation}
which is true with probability $1 - e^{-\Omega(m)}$ from \lemmaref{lemma:Delta}. It implies that 
\begin{equation}
    \teta^{(\ell,\ell')}_{k,ij} = O\left(\frac{1}{m^3}\right),\quad \forall k,i,j,\ell\neq \ell'. 
\end{equation}
Let 
\begin{equation}
V = \sum_{i=1}^r \sum_{\ell=1}^{m-1} 
\ket{f_{i,\ell}} \bra{e_{i,\ell}},
\end{equation}
where for $i = 1,\ldots,r$ and $\ell = 1,\ldots, m-1$, 
\begin{equation}
\ket{e_{i,\ell}} = \frac{1}{\sqrt{m}}\sum_{\ell'=1}^{m} \exp\left(-i \frac{2\pi}{m}\ell \ell'\right) \ket{\widehat{J}_{0,i}^{(\ell')}},\quad 
\ket{f_{i,\ell}} = \frac{1}{\sqrt{m}}\sum_{\ell'=1}^{m} \exp\left(-i \frac{2\pi}{m}\ell \ell'\right) \ket{\widehat{J}_{1,i}^{(\ell')}},
\end{equation}
and for $i = 1,\ldots,r$, $\ell = 1,\ldots, m$ and $k = 0,1$, 
\begin{equation}
\ket{\widehat{J}_{k,i}^{(\ell)}} = \frac{1}{\sqrt{\mu_i}}(L_i^{(\ell)} - b_{k,i})\ket{k_\tl}. 
\end{equation}
They satisfy, for all $\ell$, 
\begin{equation}
\label{eq:condition-J-ancilla-free-1}
\braket{\widehat{J}_{k,i}^{(\ell)}|\widehat{J}_{k,j}^{(\ell)}}  = \delta_{ij} - \frac{b_{k,i}^*b_{k,j}}{\sqrt{\mu_i\mu_j}} = \delta_{ij} + O\left(\frac{1}{m^2}\right). 
\end{equation}
and for $\ell\neq\ell'$, 
\begin{equation}
\label{eq:condition-J-ancilla-free-2}
\braket{\widehat{J}_{k,i}^{(\ell)}|\widehat{J}_{k,j}^{(\ell')}} = \frac{1}{\sqrt{\mu_i\mu_j}} \big(\eta_{k,ij} - {b_{k,i}^*b_{k,j}} + \teta^{(\ell,\ell')}_{k,ij}\big) = O\left(\frac{1}{m}\right). 
\end{equation}
$\ket{e_{i,\ell}}$ and $\ket{f_{i,\ell}}$ satisfy, for all $1 \leq \ell \leq m - 1$,  
\begin{equation}
\braket{e_{i,\ell}|e_{j,\ell'}}  = \delta_{\ell\ell'} \left(\delta_{ij} - \frac{1}{\sqrt{\mu_i\mu_j}} \eta_{0,ij}\right) + \tepsilon^{(\ell,\ell')}_{0,ij},\quad 
\braket{f_{i,\ell}|f_{j,\ell'}}  = \delta_{\ell\ell'} \left(\delta_{ij} - \frac{1}{\sqrt{\mu_i\mu_j}} \eta_{1,ij}\right) + \tepsilon^{(\ell,\ell')}_{1,ij},
\end{equation}
where for $k = 0,1$, 
\begin{equation}
\tepsilon^{(\ell,\ell')}_{k,ij} = \frac{1}{m} \sum_{l,l'=1}^m \exp\left({-i\frac{2\pi}{m} (\ell l-\ell' l') } \right) \teta^{(l,l')}_{k,i,j} = O\left(\frac{1}{m^2}\right). 
\end{equation}
Note that the new terms $\tepsilon^{(\ell,\ell')}_{k,ij}$ appear in the case of the ancilla-free random code, compared to the previous small-ancilla case. However, as we will show later, the contribution of $\tepsilon^{(\ell,\ell')}_{k,ij}$ is vanishingly small.

First, by direct calculations and using the definitions of $\ket{e_{i,\ell}}$ and $\ket{f_{i,\ell}}$, we have 
\begin{equation}
V = \sum_{i=1}^{r} \sum_{\ell=1}^{m-1} \ket{f_{i,\ell}}\bra{e_{i,\ell}} 
= \sum_{i=1}^{r} \left( \sum_{\ell=1}^m \frac{m-1}{m}\ket{\widehat{J}_{1,i}^{(\ell)}}  \bra{\widehat{J}_{0,i}^{(\ell)}} + \sum_{\ell\neq \ell',\ell,\ell'=1}^m \frac{-1}{m}\ket{\widehat{J}_{1,i}^{(\ell)}}  \bra{\widehat{J}_{0,i}^{(\ell')}} \right). 
\end{equation}
Furthermore, from 
\eqref{eq:condition-J-ancilla-free-1} and 
\eqref{eq:condition-J-ancilla-free-2}, we have 
\begin{equation}
    \sum_{i=1}^r\bra{\widehat{J}_{1,i}^{(\ell)}}  B \ket{\widehat{J}_{0,i}^{(\ell)}}
    = \sum_{i} \mu_i + O\left(\frac{1}{m}\right), \quad
    \text{and}\quad
    \sum_{i=1}^r\bra{\widehat{J}_{1,i}^{(\ell)}}  B \ket{\widehat{J}_{0,i}^{(\ell')}}
    = O\left(\frac{1}{m}\right),\quad \text{when }\ell\neq\ell'.
\end{equation}
Then we have 
\begin{align}
\label{eq:B-real-random}
\trace(BV) = \sum_{i,\ell} \bra{e_{i,\ell}} B \ket{f_{i,\ell}}= m \sum_{i} \mu_i + O(1),
\end{align}
proving the first equation in \eqref{eq:lower-4}.

Meanwhile, we want to prove an upper bound of $1 + O(1/m)$ on 
\begin{align}
\norm{V}_\infty^2 &= \big\|V^\dagger V\big\|_\infty = \bigg\|\sum_{i,\ell,j,\ell'} \ket{e_{i,\ell}}\braket{f_{i,\ell}|f_{j,\ell'}}\bra{e_{j,\ell'}}\bigg\|_\infty\\ 
&= \bigg\|\sum_{i,j,\ell,\ell'} \left(\delta_{\ell\ell'} \left(\delta_{ij} - \frac{1}{\sqrt{\mu_i\mu_j}} \eta_{1,ij} \right) + \tepsilon^{(\ell,\ell')}_{1,ij} \right) \ket{e_{i,\ell}}\bra{e_{j,\ell'}}\bigg\|_\infty\\ 
& = \max_{\text{unit vector }\ket{v}} \sum_{i,j,\ell,\ell'} \left(\delta_{\ell\ell'} \left( \delta_{ij} - \frac{1}{\sqrt{\mu_i\mu_j}} \eta_{1,ij} \right) + \tepsilon^{(\ell,\ell')}_{1,ij} \right) \braket{v|e_{i,\ell}}\braket{e_{j,\ell'}|v}. 
\end{align}
Let $\ket{v} = \sum_{i,\ell} v_{i,\ell} \ket{e_{i,\ell}}$ be an arbitrary unitary vector. Then using 
\begin{align}
    1 &= \braket{v|v} = \sum_{i,j,\ell,\ell'}  v^*_{i,\ell}v_{j,\ell'} \braket{e_{i,\ell}|e_{j,\ell'}} = 
    \sum_{i,\ell}  |v_{i,\ell}|^2  
    - \sum_{i,j,\ell}  v^*_{i,\ell}v_{j,\ell}  \frac{1}{\sqrt{\mu_i\mu_j}} \eta_{0,ij} + 
    \sum_{i,j,\ell,\ell'}  v^*_{i,\ell}v_{j,\ell'}  \tepsilon^{(\ell,\ell')}_{0,ij}\\
    & \geq \sum_{i,\ell}  |v_{i,\ell}|^2  
    - \sum_{i,j,\ell}  \frac{|v_{i,\ell}|^2 + |v_{j,\ell}|^2}{2}  \frac{1}{\sqrt{\mu_i\mu_j}} |\eta_{0,ij}| - 
    \sum_{i,j,\ell,\ell'}  \frac{|v_{i,\ell}|^2 + |v_{j,\ell'}|^2}{2}   |\tepsilon^{(\ell,\ell')}_{0,ij}|\\
    & =: \sum_{i,\ell}  |v_{i,\ell}|^2  \left(1 + \nu_{i,\ell}\right), \quad \text{where }|\nu_{i,\ell}| = O\left(\frac{1}{m}\right), 
\end{align}
and 
\begin{align}
&\quad \sum_{i,j,\ell,\ell'} \left(\delta_{\ell\ell'} \left( \delta_{ij} - \frac{1}{\sqrt{\mu_i\mu_j}} \eta_{1,ij}\right)  + \tepsilon^{(\ell,\ell')}_{1,ij}\right) \braket{v|e_{i,\ell}}\braket{e_{j,\ell'}|v} \\
&= 
\sum_{i,j,\ell,\ell',k,k',q,q'} v_{k,q}^* v_{k',q'}  \left(\delta_{\ell\ell'} \left( \delta_{ij} - \frac{1}{\sqrt{\mu_i\mu_j}} \eta_{1,ij} \right)  + \tepsilon^{(\ell,\ell')}_{1,ij} \right) 
\left(\delta_{q\ell} \left( \delta_{ki} - \frac{1}{\sqrt{\mu_k\mu_i}} \eta_{0,ki} \right)  + \tepsilon^{(q,\ell)}_{0,ki}\right)
\nonumber \\ & \qquad \qquad \qquad  \qquad \qquad \qquad  \qquad \qquad \qquad  \qquad \qquad \qquad  \times \left(\delta_{\ell'q'} \left( \delta_{jk'} - \frac{1}{\sqrt{\mu_j\mu_{k'}}} \eta_{0,jk'} \right)  + \tepsilon^{(\ell',q')}_{0,jk'}\right)\\
&\leq \sum_{k,k',q,q'}\frac{|v_{k,q}|^2 + |v_{k',q'}|^2}{2} \Bigg| \sum_{i,j,\ell,\ell'} \left(\delta_{\ell\ell'}  \left( \delta_{ij} - \frac{1}{\sqrt{\mu_i\mu_j}} \eta_{1,ij} \right)  + \tepsilon^{(\ell,\ell')}_{1,ij} \right) 
\left(\delta_{q\ell} \left( \delta_{ki} - \frac{1}{\sqrt{\mu_k\mu_i}} \eta_{0,ki} \right)  + \tepsilon^{(q,\ell)}_{0,ki}\right) 
\nonumber \\ & \qquad \qquad \qquad  \qquad \qquad \qquad  \qquad \qquad \qquad  \qquad \qquad \qquad  \times \left(\delta_{\ell'q'} \left( \delta_{jk'} - \frac{1}{\sqrt{\mu_j\mu_{k'}}} \eta_{0,jk'} \right)  + \tepsilon^{(\ell',q')}_{0,jk'}\right) \Bigg|\\
&=: \sum_{i,\ell} |v_{i,\ell}|^2 \left(1 + \tilde\nu_{i,\ell}\right), \quad \text{where }|\tilde\nu_{i,\ell}| = O\left(\frac{1}{m}\right), 
\end{align}
we have 
\begin{align}
    \norm{V}_\infty^2 
    &\leq \max_{\text{unit vector }\ket{v}} \sum_{i,\ell} |v_{i,\ell}|^2 \left(1 + \tilde\nu_{i,\ell}\right) \leq  \max_{\text{unit vector }\ket{v}} \sum_{i,\ell} |v_{i,\ell}|^2 \left(1 + \tilde\nu_{i,\ell} + \nu_{i,\ell} - \nu_{i,\ell}\right)\\
    &\leq 1 + \max_{\text{unit vector }\ket{v}} \sum_{i,\ell} |v_{i,\ell}|^2 \abs{ 1 + \nu_{i,\ell}}  \abs{ \frac{\tilde\nu_{i,\ell}- \nu_{i,\ell}}{1 + \nu_{i,\ell}} } = 1 + O\left(\frac{1}{m}\right),
\end{align}
where the last step follows from $\tilde\nu_{i,\ell} = O(1/m)$, 
$\nu_{i,\ell} = O(1/m)$ and $\sum_{i,\ell} |v_{i,\ell}|^2 (1 + \nu_{i,\ell} ) = 1$. 
\end{proof}

Finally, we remark that $\ket{e_{i,\ell}}$ and $\ket{f_{i,\ell}}$ in the proof above generate an asymptotically optimal recovery channel $\mR_{\rm opt}$ of the form in \eqref{eq:recovery} with probability $1 - e^{-\Omega(m)}$ for achieving the optimal HL using the ancilla-free random code, where $p = (i,\ell)$ for $1 \leq i \leq r$ and $1\leq \ell \leq m-1$, 
\begin{equation}
\label{eq:optimal-recovery-random}
\ket{R_{(i,\ell)}} = \frac{\ket{e_{i,\ell}}}{\sqrt{\norm{\sum_{j,k}\ket{e_{j,k}}\bra{e_{j,k}}}_\infty}},\quad 
\ket{S_{(i,\ell)}} = \frac{\ket{f_{i,\ell}}}{\sqrt{\norm{\sum_{j,k}\ket{f_{j,k}}\bra{f_{j,k}}}_\infty}}.
\end{equation}
As long as 
\begin{equation}
    \abs{\Delta_{k,ij}^{(\ell,\ell')}} < \frac{1}{m^3},\quad  \forall k,i\neq j,\ell\neq\ell' 
\end{equation}
which is true with probability $1 - e^{-\Omega(m)}$, we have $\Re[\trace(BV)] = m\sum_{i} \mu_i + O(1)$ from \eqref{eq:B-real-random}. Moreover, using exactly the same technique in proving $\norm{V}_\infty = \norm{\sum_{i=1}^r \sum_{\ell=1}^{m-1} \ket{e_{i,\ell}}\bra{f_{i,\ell}}}_\infty \leq 1 + O(1/m)$, we can also prove 
\begin{gather}
\bigg\| \sum_{i=1}^r \sum_{\ell=1}^{m-1} \ket{e_{i,\ell}}\bra{e_{i,\ell}}\bigg\|_\infty \leq 1 + O\left(\frac{1}{m}\right),\\
\bigg\| \sum_{i=1}^r \sum_{\ell=1}^{m-1} \ket{f_{i,\ell}}\bra{f_{i,\ell}}\bigg\|_\infty \leq 1 + O\left(\frac{1}{m}\right).
\end{gather}
One can then verify that 
\begin{align}
    \gamma_\tl(\mR_{\rm opt}) 
    &= \sum_{i=1}^r \sum_{\ell=1}^m - \Re[\bra{0_\tl} L_i^{(\ell)} \ket{0_\tl} \bra{1_\tl} L_i^{(\ell)\dagger} \ket{1_\tl}] + \frac{1}{2}(\bra{0_\tl} L_i^{(\ell)\dagger} L_i^{(\ell)} \ket{0_\tl} + \bra{1_\tl} L_i^{(\ell)\dagger} L_i^{(\ell)} \ket{1_\tl})\\  & \qquad \qquad \qquad \qquad  - \sum_{i,\ell}\Re[\bra{R_{(i,\ell)}}(\id - \ket{0_\tl}\bra{0_\tl})(L_i^{(\ell)} \ket{0_\tl}\bra{1_\tl} L_i^{(\ell)\dagger} )(\id - \ket{1_\tl}\bra{1_\tl})\ket{S_{(i,\ell)}}]\\
    &= m \sum_{i=1}^r {\mu}_i - \Re[b_{0,i}^*b_{1,i}] + \frac{1}{2}(a_{0,ii}+a_{1,ii}) - \frac{\Re[\trace(BV)]}{\norm{\sum_{j,k}\ket{e_{j,k}}\bra{e_{j,k}}}_\infty \norm{\sum_{j,k}\ket{f_{j,k}}\bra{f_{j,k}}}_\infty}\\
    &\leq m \sum_{i=1}^r {\mu}_i + O(1) - \frac{m\sum_{i} \mu_i + O(1)}{1 + O(1/m)} = O(1).  
\end{align}

}


\sisilong{
\subsection{Numerical simulation}


In \appref{app:limit-1} and \appref{app:limit-2}, we proved \thmref{thm:near-optimal}, that is, $\gamma_\tl = O(1)$ as $m$ increases for the small-ancilla code and (with high probability) the ancilla-free random code. As discussed in the main text, if we consider an input state 
\begin{equation}
\ket{\psi_{\rm in}} = \frac{1}{\sqrt{2}}(\ket{0_\tl} + \ket{1_\tl}),
\end{equation}
and encode $\ket{0_\tl}$ and $\ket{1_\tl}$ using $m = N$ probes. The output QFI is 
\begin{equation}
\label{eq:output-qfi}
F(\rho_\omega(t)) = N^2 \trace((\tilde\rho_0 - \tilde\rho_1)H)^2 t^2 e^{-2\gamma_\tl t}. 
\end{equation}
We proved that 
\begin{equation}
\lim_{N\rightarrow \infty}\trace((\tilde\rho_0 - \tilde\rho_1)H) = 2\norm{H-\mS},\quad \lim_{N\rightarrow \infty}\gamma_\tl = \overline{\gamma}_\tl< \infty, 
\end{equation}
which guarantees the optimal HL with respect to $N$ when $N\rightarrow \infty$ and $t \ll 1/\overline{\gamma}_\tl$. However, the performance of our codes in the regime where the number of probes $N$ is small remains to be investigated. To this end, we consider a concrete example in this section and numerically plot the gap between our codes and the optimal HL.

\begin{figure}[btp]
\centering
\includegraphics[width=0.5\textwidth]{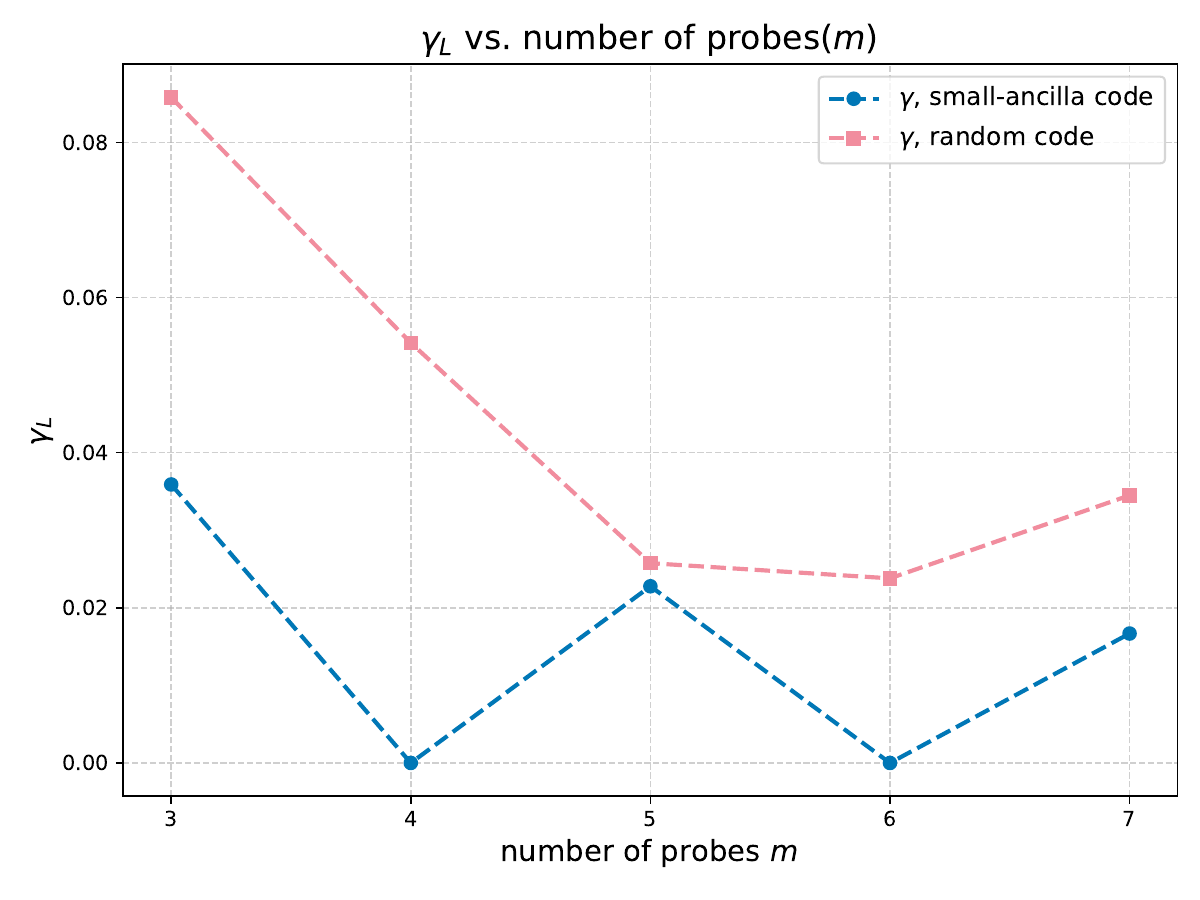}
 \caption{\label{fig:gamma}
{\sisi{Error rate $\gamma_\tl$ as a function of number of probes in the codeword for the small-ancilla code and the ancilla-free random code. We assume the probe evolution is given by the qutrit example discussed in \appref{app:qutrit}. Here, the random phases of the random code is picked as the optimal one among $50$ samples of random phases that pass a filtering test.}}
 }
\end{figure}

\begin{figure}[btp]
\centering
\includegraphics[width=0.8\textwidth]{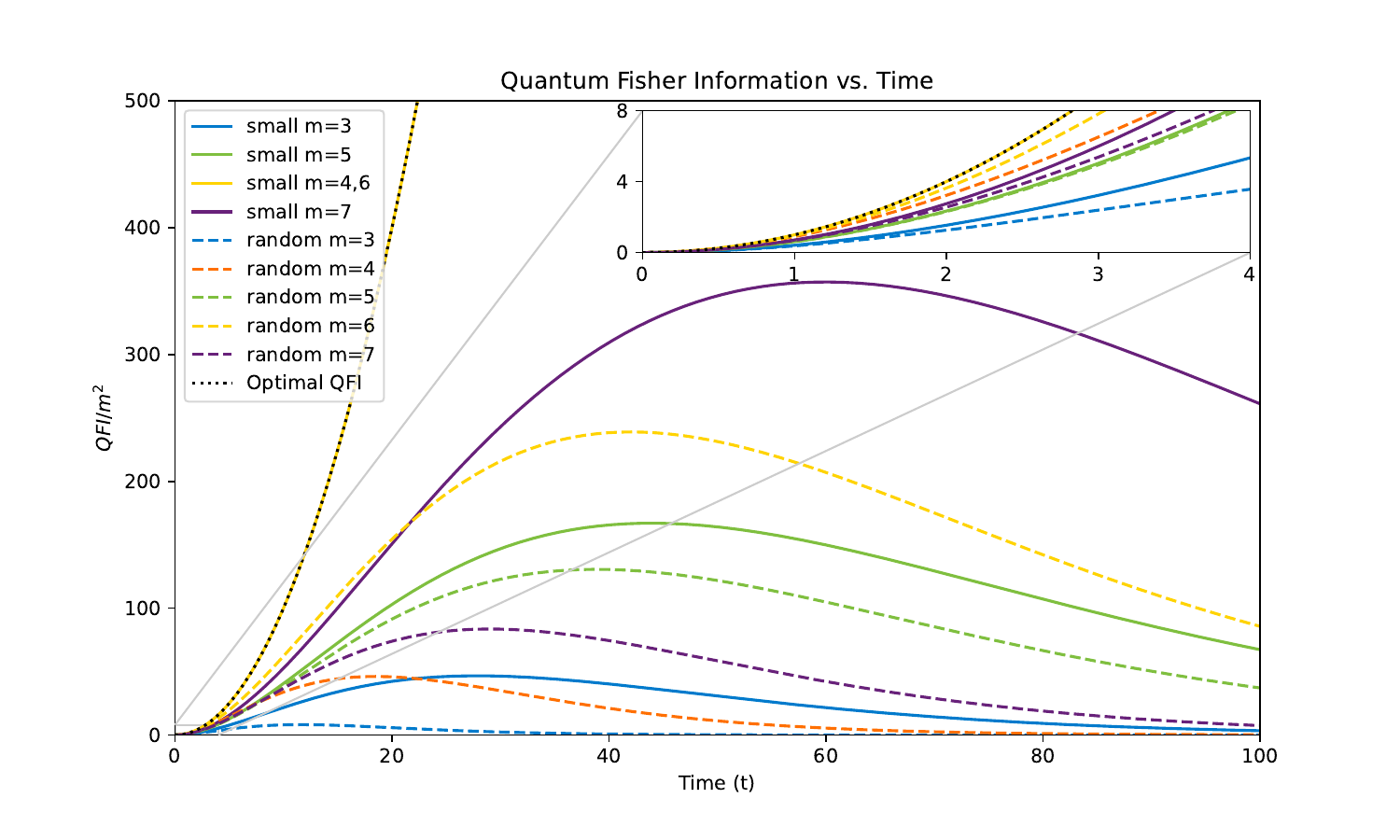}
 \caption{\label{fig:qfi}
{\sisi{Comparison between the optimal HL and the output QFI of our (effectively) ancilla-free QEC protocols. Here the number of probes in the entire state is equal to the number of probes in each codeword $N = m$. We assume the probe evolution is given by the qutrit example discussed in \appref{app:qutrit}. We plot the ideal optimal QFI $F(\rho_\omega(t))/m^2$ in the noiseless case (in dotted lines) and the output QFIs of the QEC protocol using the small-ancilla code (in solid lines) and the ancilla-free random code (in dashed lines). $m=4,6$ overlaps with the optimal QFI because $\gamma_\tl = 0$. The time where the QFI started to decrease has a scaling of $1/\gamma_\tl$. Here, the random phases of the random code is picked as the optimal one among $50$ samples of random phases that pass a filtering test, the same as in \figref{fig:gamma}.
}}
 }
\end{figure}

In the numerical simulation, we use the qutrit model introduced in \appref{app:qutrit}, where 
\begin{equation}
H = \begin{pmatrix}
1 & 0 & 0 \\
0 & -1 & -1 \\
0 & -1 & -1 \\
\end{pmatrix},\quad 
L = 
\begin{pmatrix}
0 & 1 & 1 \\
0 & 0 & 1 \\
0 & 0 & 0 \\
\end{pmatrix},  
\end{equation}
and 
\begin{equation}
    \rho_{0} = \frac{1}{2}\left(\ket{0}\bra{0} + \ket{2}\bra{2}\right),\quad 
    \rho_{1} = \ket{1}\bra{1}. 
\end{equation}
We provided an exact small-ancilla QEC code ($\gamma_\tl = 0$) for $m = 4$ probes in \appref{app:qutrit}. In general, $\gamma_\tl$ is not necessarily zero for other probe numbers, or for the random code. Here, we take $N = m = 3,4,5,6,7$ and plot the corresponding $\gamma_\tl$ for both the small-ancilla code and the ancilla-free random code in \figref{fig:gamma}. In particular, to pick the random phases in the random code, we sample multiple sets of random phases, perform a filtering test on them (a set of phases passes the test when $\abs{\Delta_{0,02}}$ defined in \eqref{eq:def-Delta} is below certain thresholds, which is deemed as a requirement for $\gamma_\tl$ to be small) and then pick the one that minimizes $\gamma_\tl$ among $50$ sets of phases that pass the filtering test. \figref{fig:gamma} shows for both the small-ancilla and the ancilla-free random code, $\gamma_\tl$ is small and has a decaying trend, which aligns with our asymptotic analysis and is in contrast with situations without QEC where $\gamma_\tl$ usually increases linearly with respect to $m$. Note that the curves of the ancilla-free random code are above those of the small-ancilla code because the small-ancilla code utilizes ancillas to remove extra non-zero terms contributing to $\gamma_\tl$, while the ancilla-free random code relies on the randomness of phases to reduce them. In particular, at $m=4,6$, there exists an exact small-ancilla code, i.e., $\gamma_\tl = 0$. Finally, we note that the reason that $\gamma_\tl$ at $m=7$ is larger than $\gamma_\tl$ at $m = 6$ for ancilla-free random codes might be there are more choices of random phases at $m = 7$ and we didn't collect a sufficiently large number of sets of random phases to find the optimal random phases at $m = 7$ due to the increasing running time of simulation as $m$ increases. We also plot the normalized QFI $F(\rho_\omega(t))/N^2$ as a function of $t$ in \figref{fig:qfi} (and we take $N=m$). We first plot the optimal ideal QFI in the noiseless case, which shows a perfect hyperbolic shape indicating the saturation of the optimal HL. We then plot the output QFI (\eqref{eq:output-qfi}) of the two metrological protocols using the small-ancilla code and the ancilla-free random code when $m=3,4,5,6,7$. Since the code is exact at $m=4,6$, the optimal HL is achieved. In other cases, the QFI first increases quadratically in the regime where $t \ll 1/\gamma_\tl$ but then decays as $t$ becomes larger. In particular, as $m$ increases, the performance of the QEC protocol does not decline, in contrast to the case without QEC. Also, note that the normalized QFI $F(\rho_\omega(t))/m^2$ is (in most cases) closer to the optimal HL in the small $t$ regime as $m$ increases (see the zoomed-in plot around $t = 0$ in \figref{fig:qfi}), which aligns with our asymptotic analysis.  


The source code for the numerical simulation is available at \href{https://github.com/argygianni/Ancilla-free-QECC-for-metrology}{\color{black}{https://github.com/argygianni/Ancilla-free-QECC-for-metrology}}}.


\section{Achieving the optimal SQL using ancilla-free QEC codes}
\label{app:SQL}

In this appendix, we first review existing results on achieving the optimal SQL when HNLS fails using ancilla-assisted QEC. Then we show that using the similar techniques that were presented in the main text, we can define two (effectively) ancilla-free codes: the small-ancilla code{$^\sql$} (similar to \eqref{eq:code-few-ancilla}) and the qubit-ancilla random code{$^\sql$} (similar to \eqref{eq:code-ancilla-free}) that achieve the optimal SQL. Note that in this appendix, we use $^\sql$ to distinguish the codes achieving the optimal SQL from those achieving the optimal HL when HNLS holds in the main text. 

\subsection{Optimal ancilla-assisted QEC}
\label{app:ancilla-assisted-SQL}

When HNLS fails (i.e, $H \in \mS$), the SQL coefficient has the upper bound~\cite{demkowicz2017adaptive,zhou2018achieving},
\begin{equation}
    \sup_{t > 0} \lim_{N\rightarrow \infty} \frac{F(\rho_\omega(t))}{Nt} \leq 4{o_1},
\end{equation}
where ${o_1}$ is defined in \eqref{eq:alpha}. It was shown in~\cite{zhou2019optimal} that for any $0 < \delta \leq 4 o_1$, there exists a two-dimensional ancilla-assisted QEC code that entangles one probe and two ancillas $\mH_P \otimes \mH_{A_1} \otimes \mH_{A_2}$, where $\dim(\mH_P) = \mH_{A_1} = d$ and $\dim(\mH_{A_2}) = 2$, that achieves 
\begin{equation}
    \sup_{t > 0} \lim_{N\rightarrow \infty} \frac{F(\rho_\omega(t))}{Nt} = 4{o_1} - \delta. 
\end{equation}
The ancilla-assisted code$^\sql$ is defined by 
\begin{equation}
\label{eq:code-ancilla-SQL}
    \ket{0^\sg_\tl} = \sum_{i=0}^{d-1} \sqrt{\lambda_i^{[0]}}\ket{\phi_i}_P\ket{\phi_i}_{A_1}\ket{0}_{A_2},\quad \ket{1^\sg_\tl} = \sum_{i=0}^{d-1} \sqrt{\lambda_i^{[1]}}\ket{\varphi_i}_P\ket{\varphi_i}_{A_1}\ket{1}_{A_2},
\end{equation}
where $\{\ket{\phi_i}\}_{i=0}^{d-1}$ and $\{\ket{\varphi_i}\}_{i=0}^{d-1}$ are two sets of orthonormal bases in $\mH_P$ (or $\mH_{A_1}$), $\sum_{i=0}^{d-1} \lambda_i^{[0]} = \sum_{i=0}^{d-1} \lambda_i^{[1]} = 1$ and $\lambda_i^{0,1} > 0$. Differing from the HL case, $\{\ket{\phi_i}\}$ and $\{\ket{\varphi_i}\}$ are not perpendicular to each other and the additional qubit ancilla $\mH_{A_2}$ is required to make sure that $\braket{0^\sg_\tl|1^\sg_\tl} = 0$. We do not specify the exact values of $\lambda_i^{[0],[1]}$, $\ket{\phi_i}$ or $\ket{\varphi_i}$ as they will not be used in the following discussion. Note that we use the superscript $^\sg$ to denote the single-probe case, in order to distinguish it from the multi-probe case that we discuss later. 

Following the discussion in \appref{app:dephasing} (and also in \cite{zhou2019optimal}), there exists a recovery operation for the ancilla-assisted code$^\sql$ above such that the logical master equation for the logical qubit consists of a Pauli-Z Hamiltonian and a dephasing noise, described by 
\begin{equation}
    \frac{d\rho^\sg_\tl}{dt} = -i\left[\frac{\omega \trace( H Z^\sg_\tl)}{2}Z^\sg_\tl + \beta^\sg_\tl Z^\sg_\tl,\rho^\sg_\tl\right] + \frac{\gamma^\sg_\tl}{2}(Z^\sg_\tl \rho^\sg_\tl Z^\sg_\tl - \rho^\sg_\tl),
\end{equation}
where $Z^\sg_\tl = \ket{0^\sg_\tl}\bra{0^\sg_\tl} - \ket{1^\sg_\tl}\bra{1^\sg_\tl}$, $\beta^\sg_\tl$ is a parameter-independent constant,  
\begin{equation}
\label{eq:gamma-sg-SQL}
\gamma^\sg_\tl = \sum_{i=1}^r - \Re[\bra{0^\sg_\tl} L_i \ket{0^\sg_\tl} \bra{1^\sg_\tl} L_i^{\dagger} \ket{1^\sg_\tl}] +  \frac{1}{2}(\bra{0^\sg_\tl} L_i^{\dagger} L_i \ket{0^\sg_\tl} + \bra{1^\sg_\tl} L_i^{\dagger} L_i \ket{1^\sg_\tl}) - \norm{B^\sg}_1, 
\end{equation}
and 
\begin{equation}
    B^\sg = \sum_{i=1}^r  (\id-\ket{0^\sg_\tl}\bra{0^\sg_\tl})(L_i \ket{0^\sg_\tl}\bra{1^\sg_\tl} L_i^{\dagger} )(\id-\ket{1^\sg_\tl}\bra{1^\sg_\tl}). 
\end{equation}
Using an initial state $\ket{\psi^\sg_{\rm in}} = \frac{1}{\sqrt{2}}(\ket{0^\sg_\tl} + \ket{1^\sg_\tl})$, the final state satisfies 
\begin{equation}
    \sup_{t > 0} \lim_{N\rightarrow \infty} \frac{F(\rho^\sg_\omega(t))}{Nt} = \frac{\trace(H Z^\sg_\tl)^2}{2\gamma^\sg_\tl} = 4{o_1} - \delta. 
\end{equation}
Our goal is to find ancilla-free codes $\{\ket{0_\tl},\ket{1_\tl}\}$ on $m$ probes such that 
\begin{equation}
    \trace\Big(\sum_\ell H^{(\ell)} Z_\tl\Big) = m \trace(H Z^\sg_\tl) + O(1),\quad \gamma_\tl = m \gamma^\sg_\tl + O(1). 
\end{equation}
Then letting $m = \Theta(N)$, $n = O(1)$, and $N = mn$, with an initial state $\ket{\psi_{\rm in}} = \frac{1}{\sqrt{2}}(\ket{0_\tl}^{\otimes n} + \ket{1_\tl}^{\otimes n})$, the final state satisfies 
\begin{equation}
\label{eq:near-optimal}
    \sup_{t > 0} \lim_{N\rightarrow \infty} \frac{F(\rho_\omega(t))}{Nt} = \lim_{m\rightarrow \infty} \frac{n}{N} \frac{(\trace(\sum_\ell H^{(\ell)} Z_\tl))^2}{2\gamma_\tl} =  \lim_{m\rightarrow \infty} \frac{1}{m} \frac{(m\trace(H Z^\sg_\tl))^2}{2m\gamma^\sg_\tl + O(1)} =  \frac{\trace(H Z^\sg_\tl)^2}{2\gamma^\sg_\tl}  = 4{o_1} - \delta.
\end{equation}
Note that we can exchange the order of $\sup_{t > 0}$ and $\lim_{N\rightarrow \infty}$ because of the uniform convergence of $F(\rho_\omega(t))$ when $m\rightarrow\infty$. 
We will show the existence of such multi-probe codes in the next section. 

\subsection{Optimal ancilla-free QEC codes}

Here we provide two QEC codes that achieve \eqref{eq:near-optimal} and utilize an exponentially small ancilla. We will first define the codes and then prove \eqref{eq:near-optimal}. \sisinewlong{The code structures are very similar to the small-ancilla code and the ancilla-free random code we used in the HL case. The main idea is still to use the multi-probe encoding so that the 2-local density operator of the code approximate the one in the single-probe case which achieves the optimal SQL coefficient. Specifically, we will prove the following theorem. 
\begin{theorem}[Achieving the optimal SQL]
\label{thm:near-optimal}
For the small-ancilla code$^\sql$ defined below, 
\begin{equation}
\label{eq:near-optimal-thm}
    \trace\Big(\sum_\ell H^{(\ell)} Z_\tl\Big) = m \trace(H Z^\sg_\tl) + O(1),\quad \gamma_\tl = m \gamma^\sg_\tl + O(1). 
\end{equation}
For the ancilla-free random code$^\sql$ defined below, the above equation holds with probability $1 - e^{-\Omega(m)}$. 
\end{theorem}
Note that readers might find some parts in this section similar to the corresponding parts in \appref{app:limit}, as we use some identical proof techniques in both appendixes. For the sake of completeness and self-consistency, we repeated these contents.
}

\subsubsection{The small-ancilla code\texorpdfstring{$^\sql$}{}}
\label{app:small-ancilla-sql}

We first define the small-ancilla code{$^\sql$} acting on $m\geq 3$ probes in the SQL case, where 
\begin{equation}
\label{eq:code-few-ancilla-SQL}
\begin{split}
    \ket{0_\tl} &= \frac{1}{\sqrt{\abs{W^\sql_0}}} \sum_{w \in W^\sql_0} \ket{\phi_w}_{P^{\otimes m}} \ket{\fraki^\sql_{0}(w)}_{A_1} \ket{0}_{A_2},\\
    \ket{1_\tl} &= \frac{1}{\sqrt{\abs{W^\sql_1}}} \sum_{w \in W^\sql_1} \ket{\varphi_w}_{P^{\otimes m}} \ket{\fraki^\sql_{1}(w)}_{A_1} \ket{1}_{A_2},
\end{split}
\end{equation}
Here we use $\ket{\phi_w}$ (or $\ket{\varphi_w}$) to denote $\ket{\phi_{w_1}}\otimes\cdots\otimes \ket{\phi_{w_m}}$ (or $\ket{\varphi_{w_1}}\otimes\cdots\otimes \ket{\varphi_{w_m}}$) for any string $w = w_1 w_2\cdots w_m \in \{0,\ldots,d-1\}^{m}$. $W^\sql_0$ and $W^\sql_1$ are two sets of strings of length $m$ and $\abs{W^\sql_k}$ denotes the order of $W^\sql_k$. To define $W^\sql_{0,1}$, note that for all $m$, it is possible to find two sets of integers $\{m^{[0]}_i\}_{i=0}^{d-1}$ and $\{m^{[1]}_i\}_{i=0}^{d-1}$, such that $\sum_{i=0}^{d-1} m^{[0]}_i = \sum_{i=0}^{d-1} m^{[1]}_i = m$ and 
\begin{equation}
    \abs{\frac{m^{[k]}_i}{m} - \lambda^{[k]}_i} \leq \frac{1}{m},\;\forall i=0,\ldots,d-1,\;k=0,1. 
\end{equation} 
Note that $\lim_{m\rightarrow \infty} m^{[k]}_i/m = \lambda^{[k]}_i$ for $k = 0,1$. 
Then we define $W^\sql_k$ to be the set of strings that contains $m^{[k]}_i$ $i$'s for $0 \leq i \leq d-1$.  
$\fraki^\sql_k(w)$ can be chosen to be any integer functions that satisfy $\fraki^\sql_k(w)\neq \fraki^\sql_k(w')$ when $w$ and $w'$ are different on exactly two sites. Let $\frakI^\sql = \{\fraki^\sql_k(w),\forall k=0,1, w\in W_0\cup W_1\}$. Every integer in $\frakI^\sql$ corresponds to an (orthonormal) basis state of the ancilla $\mH_{A_1}$ and we can assume $\dim(\mH_{A_1}) = \abs{\frakI^\sql}$. We choose $\fraki^\sql_k(w)$ such that $\abs{\frakI^\sql}$ is minimized. Then using graph theory arguments in \appref{app:graph}, $\abs{\frakI^\sql}$ is no larger than $m^2d^2$. 
In the regime of $m \gg 1$, the dimension of the ancilla $\dim(\mH_A) = \dim(\mH_{A_1})\dim(\mH_{A_2}) = 2\abs{\frakI^\sql}$ is exponentially smaller than the dimension of the probes.

Following the discussion in \appref{app:dephasing}, there exists a recovery operation for the small-ancilla code{$^\sql$} such that the master equation for the logical qubit consists of a Pauli-Z Hamiltonian and a dephasing noise, described by 
\begin{equation}
    \frac{d\rho_\tl}{dt} = -i\left[\frac{\omega \trace( \sum_\ell H^{(\ell)} Z_\tl)}{2}Z_\tl + \beta_\tl Z_\tl,\rho_\tl\right] + \frac{\gamma_\tl}{2}(Z_\tl \rho_\tl Z_\tl - \rho_\tl),
\end{equation}
where 
\begin{equation}
\label{eq:gamma-SQL}
\gamma_\tl = \sum_{i=1}^r \sum_{\ell=1}^m - \Re[\bra{0_\tl} L_i^{(\ell)} \ket{0_\tl} \bra{1_\tl} L_i^{(\ell)\dagger} \ket{1_\tl}] +  \frac{1}{2}(\bra{0_\tl} L_i^{(\ell)\dagger} L_i^{(\ell)} \ket{0_\tl} + \bra{1_\tl} L_i^{(\ell)\dagger} L_i^{(\ell)} \ket{1_\tl}) - \norm{B}_1, 
\end{equation}
and 
\begin{equation}
    B = \sum_{i=1}^r  \sum_{\ell=1}^m(\id-\ket{0_\tl}\bra{0_\tl})(L_i^{(\ell)} \ket{0_\tl}\bra{1_\tl} L_i^{(\ell)\dagger} )(\id-\ket{1_\tl}\bra{1_\tl}). 
\end{equation}

Note that for the small-ancilla code{$^\sql$},
\begin{gather}
 \braket{0_\tl|L_i^{(\ell)}|0_\tl} - \braket{0^\sg_\tl|L_i|0^\sg_\tl}
= \trace\left(\sum_{i=0}^{d-1} \left(\frac{m^{[0]}_i}{m}-\lambda^{[0]}_i\right) \ket{i}\bra{i} \cdot L_i\right) = O\left(\frac{1}{m}\right),\\
 \braket{1_\tl|L_i^{(\ell)}|1_\tl} - \braket{1^\sg_\tl|L_i|1^\sg_\tl}
= \trace\left(\sum_{i=0}^{d-1} \left(\frac{m^{[1]}_i}{m}-\lambda^{[1]}_i\right) \ket{i}\bra{i} \cdot L_i\right) = O\left(\frac{1}{m}\right),\\
 \braket{0_\tl|L_i^{(\ell)\dagger}L_j^{(\ell)}|0_\tl} - \braket{0^\sg_\tl|L_i^{\dagger}L_j|0^\sg_\tl} = \trace\left(\sum_{i=0}^{d-1} \left(\frac{m^{[0]}_i}{m}-\lambda^{[0]}_i\right) \ket{i}\bra{i} \cdot L_i^{\dagger}L_j\right)   =  O\left(\frac{1}{m}\right),\\
 \braket{1_\tl|L_i^{(\ell)\dagger}L_j^{(\ell)}|1_\tl} - \braket{1^\sg_\tl|L_i^{\dagger}L_j|1^\sg_\tl} = \trace\left(\sum_{i=0}^{d-1} \left(\frac{m^{[1]}_i}{m}-\lambda^{[1]}_i\right) \ket{i}\bra{i} \cdot L_i^{\dagger}L_j\right) =  O\left(\frac{1}{m}\right).
\end{gather}
Combining \eqref{eq:gamma-sg-SQL} and \eqref{eq:gamma-SQL}, we obtain that 
\begin{equation}
\label{eq:gamma-first-step}
    \gamma_\tl = m \gamma_\tl^\sg - \norm{B}_1 + m \norm{B^\sg}_1 + O(1).
\end{equation}
(As we will see later, \eqref{eq:gamma-first-step} holds for both the small ancilla code$^\sql$ and the ancilla-free random code$^\sql$ in the next section.) 

To show $\gamma_\tl = m \gamma_\tl^\sg = O(1)$, we only need to prove 
\begin{equation}
    - \norm{B}_1 + m \norm{B^\sg}_1 \leq  O(1). 
\end{equation}
To do so, we first define two sets of orthonormal states $\{\ket{D^\sg_{0,i}}\} \subseteq {\rm span}\{(\id-\ket{0^\sg_\tl}\bra{0^\sg_\tl})(L_i \ket{0^\sg_\tl},\forall i\}$ and $\{\ket{D^\sg_{1,i}}\} \subseteq {\rm span}\{(\id-\ket{1^\sg_\tl}\bra{1^\sg_\tl})(L_i \ket{1^\sg_\tl},\forall i\}$ such that 
\begin{equation}
\label{eq:B-single}
    B^\sg = \sum_i \gamma_{i}^\sg \ket{D^\sg_{0,i}}\bra{D^\sg_{1,i}},\quad \text{  for some  }\{\gamma_{i}^\sg\}.
\end{equation}
In particular, let 
\begin{equation}
\ket{D^\sg_{0,i}} = \sum_{j=1}^{r} D_{0,ij}(\id - \ket{0^\sg_\tl}\bra{0^\sg_\tl})L_j\ket{0^\sg_\tl}
,\quad 
\ket{D^\sg_{1,i}} = \sum_{j=1}^{r} D_{1,ij}(\id - \ket{1^\sg_\tl}\bra{1^\sg_\tl})L_j\ket{1^\sg_\tl},
\end{equation}
and define 
\begin{equation}
\ket{D^{(\ell)}_{0,i}} = \sum_{j=1}^{r} D_{0,ij} (\id - \ket{0_\tl}\bra{0_\tl}) L_j^{(\ell)}\ket{0_\tl}
,\quad 
\ket{D^{(\ell)}_{1,i}} = \sum_{j=1}^{r} D_{1,ij} (\id - \ket{1_\tl}\bra{1_\tl}) L_j^{(\ell)}\ket{1_\tl}. 
\end{equation}
They satisfy, for all $\ell \neq \ell'$, 
\begin{gather}
\label{eq:condition-1}
\braket{D^{(\ell)}_{0,i}| D^{(\ell)}_{0,j}}  = 
\sum_{i'j'} D_{0,ii'}^* D_{0,jj'} \bra{0_\tl}L_{i'}^{(\ell)\dagger} (\id - \ket{0_\tl}\bra{0_\tl}) L_{j'}^{(\ell)}\ket{0_\tl} = \delta_{ij} + O\left(\frac{1}{m}\right), 
\\
\label{eq:condition-2}
\braket{D^{(\ell)}_{1,i}| D^{(\ell)}_{1,j}}  = 
\sum_{i'j'} D_{1,ii'}^* D_{1,jj'} \bra{1_\tl}L_{i'}^{(\ell)\dagger} (\id - \ket{1_\tl}\bra{1_\tl}) L_{j'}^{(\ell)}\ket{1_\tl} = \delta_{ij} + O\left(\frac{1}{m}\right), 
\\
\label{eq:condition-3}
\bra{D^{(\ell)}_{0,i}}(\id - \ket{0_\tl}\bra{0_\tl}) L_j^{(\ell)} \ket{0_\tl} 
= \sum_{i'} D^{*}_{0,ii'}  \bra{0^\sg_\tl} L^{\dagger}_{i'} (\id - \ket{0^\sg_\tl}\bra{0^\sg_\tl}) L_j \ket{0^\sg_\tl} + O\left(\frac{1}{m}\right),
\\
\label{eq:condition-4}
\bra{1_\tl} L_j^{(\ell)\dagger} (\id - \ket{1_\tl}\bra{1_\tl}) \ket{D^{(\ell)}_{1,i}}
= \sum_{i''} D_{1,ii''}  \bra{1^\sg_\tl} L^{\dagger}_{j} (\id - \ket{1^\sg_\tl}\bra{1^\sg_\tl}) L_{i''} \ket{1^\sg_\tl} + O\left(\frac{1}{m}\right),
\\
\label{eq:condition-5}
\begin{split}
\braket{D^{(\ell)}_{0,i}| D^{(\ell')}_{0,j}} &= 
\sum_{i'j'} D_{0,ii'}^* D_{0,jj'} \bra{0_\tl}L_{i'}^{(\ell)\dagger} (\id - \ket{0_\tl}\bra{0_\tl}) L_{j'}^{(\ell')}\ket{0_\tl} \\
&= \sum_{i'j'} D_{0,ii'}^* D_{0,jj'} \left(\frac{ \braket{0_\tl|L_{i'}^{(\ell)\dagger}|0_\tl}  \braket{0_\tl|L_{j'}^{(\ell)}|0_\tl} }{m-1} - \sum_{k=0}^{d-1} \frac{m^{[0]}_k}{m(m-1)} \bra{k}L_{i'}^{\dagger}\ket{k}\bra{k}L_{j'}\ket{k}\right) = O\left(\frac{1}{m}\right),
\end{split}\\
\label{eq:condition-6}
\begin{split}
\braket{D^{(\ell)}_{1,i}| D^{(\ell')}_{1,j}} &= 
\sum_{i'j'} D_{1,ii'}^* D_{1,jj'} \bra{1_\tl}L_{i'}^{(\ell)\dagger} (\id - \ket{1_\tl}\bra{1_\tl}) L_{j'}^{(\ell')}\ket{1_\tl} \\
&= \sum_{i'j'} D_{1,ii'}^* D_{1,jj'} \left(\frac{\braket{1_\tl|L_{i'}^{(\ell)\dagger}|1_\tl}  \braket{1_\tl|L_{j'}^{(\ell)}|1_\tl} }{m-1} - \sum_{k=0}^{d-1} \frac{m^{[1]}_k}{m(m-1)} \bra{k}L_{i'}^{\dagger}\ket{k}\bra{k}L_{j'}\ket{k}\right) = O\left(\frac{1}{m}\right),  
\end{split}
\\
\label{eq:condition-7}
\begin{split}
&\bra{D^{(\ell)}_{0,i}}(\id - \ket{0_\tl}\bra{0_\tl}) L_j^{(\ell')} \ket{0_\tl} 
= \sum_{i'} D^{*}_{0,ii'} \bra{0_\tl} L^{(\ell)\dagger}_{i'} (\id - \ket{0_\tl}\bra{0_\tl}) L_j^{(\ell')} \ket{0_\tl} \\
&\qquad\qquad = \sum_{i'} D_{0,ii'}^* \left(\frac{ \braket{0_\tl|L_{i'}^{(\ell)\dagger}|0_\tl}  \braket{0_\tl|L_{j}^{(\ell)}|0_\tl} }{m-1} - \sum_{k=0}^{d-1} \frac{m^{[0]}_k}{m(m-1)} \bra{k}L_{i'}^{\dagger}\ket{k}\bra{k}L_{j'}\ket{k}\right) = O\left(\frac{1}{m}\right), 
\end{split}
\\
\label{eq:condition-8}
\begin{split}
&\bra{1_\tl} L_j^{(\ell)\dagger} (\id - \ket{1_\tl}\bra{1_\tl}) \ket{D^{(\ell')}_{1,i}} = \sum_{i''} D_{1,ii''} \bra{1_\tl} L^{(\ell)\dagger}_{j} (\id - \ket{1_\tl}\bra{1_\tl}) L_{i''}^{(\ell')} \ket{1_\tl} \\
&\qquad\qquad =\sum_{j'}  D_{1,jj'} \left(\frac{ \braket{1_\tl|L_{i}^{(\ell)\dagger}|1_\tl}  \braket{1_\tl|L_{j'}^{(\ell)}|1_\tl} }{m-1} - \sum_{k=0}^{d-1} \frac{m^{[1]}_k}{m(m-1)} \bra{k}L_{i}^{\dagger}\ket{k}\bra{k}L_{j'}\ket{k}\right) =  O\left(\frac{1}{m}\right).
\end{split}
\end{gather}

Using the above relations and adopting a similar approach in \lemmaref{lemma:V-small-ancilla}, we can prove the following lemma. 
\sisinewlong{
\begin{lemma}
\label{lemma:V-small-ancilla-sql}
For the small-ancilla code$^\sql$ and $B = \sum_{i=1}^r  \sum_{\ell=1}^m(\id-\ket{0_\tl}\bra{0_\tl})(L_i^{(\ell)} \ket{0_\tl}\bra{1_\tl} L_i^{(\ell)\dagger} )(\id-\ket{1_\tl}\bra{1_\tl})$, there exists a matrix $V$ such that 
\begin{equation}
\label{eq:lower-2-sql}
\abs{\trace(BV)} = m \norm{B^\sg}_1 + O(1), \quad \text{and}\quad \norm{V}_\infty \leq 1 + O\left(\frac{1}{m}\right),
\end{equation}
with probability $1 - e^{-\Omega(m)}$. Then 
\begin{equation}
    \norm{B}_1 \geq \frac{\abs{\trace(BV)}}{\norm{V}_\infty} \geq m \norm{B^\sg}_1 + O(1). 
\end{equation} 
\end{lemma}

\begin{proof}
Let 
\begin{equation}
V = \sum_{i=1}^r \sum_{\ell=1}^{m-1} 
\ket{f_{i,\ell}} \bra{e_{i,\ell}},
\end{equation}
where for $i = 1,\ldots,r$ and $\ell = 1,\ldots, m-1$, 
\begin{equation}
\ket{e_{i,\ell}} = \frac{1}{\sqrt{m}}\sum_{\ell'=1}^{m} \exp\left(-i \frac{2\pi}{m}\ell \ell'\right) \ket{D_{0,i}^{(\ell')}},\quad 
\ket{f_{i,\ell}} = \frac{1}{\sqrt{m}}\sum_{\ell'=1}^{m} \exp\left(-i \frac{2\pi}{m}\ell \ell'\right) \ket{D_{1,i}^{(\ell')}}. 
\end{equation}
$\{\ket{e_{i,\ell}}\}$ and $\{\ket{f_{i,\ell}}\}$ are Fourier transforms of $\{\ket{D_{0,i}^{(\ell)}}\}$ and $\{\ket{D_{1,i}^{(\ell)}}\}$, except that we do not define $\ket{e_{i,\ell}}$ and $\ket{f_{i,\ell}}$ when $\ell = m$ (while traditionally, the $\ell = m$ terms are also used in Fourier transform). They satisfy, for all $1 \leq \ell \leq m - 1$,  
\begin{equation}
\braket{e_{i,\ell}|e_{j,\ell'}}  = \delta_{\ell\ell'} \left(\delta_{ij} + \epsilon_{0,ij}\right),\quad 
\braket{f_{i,\ell}|f_{j,\ell'}}  = \delta_{\ell\ell'} \left(\delta_{ij} + \epsilon_{1,ij}\right),
\end{equation}
where 
\begin{gather}
    \epsilon_{0,ij} := \braket{D^{(\ell)}_{0,i}| D^{(\ell)}_{0,j}} - \braket{D^{(\ell)}_{0,i}| D^{(\ell'\neq \ell)}_{0,j}} - \delta_{ij}  = O\left(1/m\right) , \\ 
    \epsilon_{1,ij} := \braket{D^{(\ell)}_{1,i}| D^{(\ell)}_{1,j}} - \braket{D^{(\ell)}_{1,i}| D^{(\ell'\neq \ell)}_{1,j}} - \delta_{ij}  = O\left(1/m\right) . 
\end{gather}
Note that $\epsilon_{0,ij}$, $\epsilon_{1,ij}$ are independent from the specific choice of $(\ell,\ell')$ and they are $O(1/m)$ from \eqref{eq:condition-1}, \eqref{eq:condition-2}, \eqref{eq:condition-5} and \eqref{eq:condition-6}. 

First, by direct calculations and using the definitions of $\ket{e_{i,\ell}}$ and $\ket{f_{i,\ell}}$, we have 
\begin{equation}
V = \sum_{i=1}^{r} \sum_{\ell=1}^{m-1} \ket{f_{i,\ell}}\bra{e_{i,\ell}} 
= \sum_{i=1}^{r} \left( \sum_{\ell=1}^m \frac{m-1}{m}\ket{D_{1,i}^{(\ell)}}  \bra{D_{0,i}^{(\ell)}} + \sum_{\ell\neq \ell',\ell,\ell'=1}^m \frac{-1}{m}\ket{D_{1,i}^{(\ell)}}  \bra{D_{0,i}^{(\ell')}} \right).
\end{equation}
Furthermore, 
\begin{align}
    &\quad \;\sum_{i=1}^r\bra{D_{1,i}^{(\ell)}}  B \ket{D_{0,i}^{(\ell)}}\\
    &= \sum_{i,j=1}^r \sum_{\ell'=1}^m \bra{D_{1,j}^{(\ell)}}  (\id-\ket{0_\tl}\bra{0_\tl})(L_i^{(\ell')} \ket{0_\tl}\bra{1_\tl} L_i^{(\ell')\dagger} )(\id-\ket{1_\tl}\bra{1_\tl})  \ket{D_{0,j}^{(\ell)}} \\
    &= \sum_{i,j,i',i''=1}^r  D^{*}_{0,ji'} D_{1,ji''}  \bra{0^\sg_\tl} L^{\dagger}_{i'} (\id - \ket{0^\sg_\tl}\bra{0^\sg_\tl}) L_i \ket{0^\sg_\tl}  \bra{1^\sg_\tl} L^{\dagger}_{i} (\id - \ket{1^\sg_\tl}\bra{1^\sg_\tl}) L_{i''} \ket{1^\sg_\tl} + O\left(\frac{1}{m}\right) 
    \\
    &= \sum_{j=1}^r  \bra{D^\sg_{0,j}} B^\sg \ket{D^\sg_{1,j}} + O\left(\frac{1}{m}\right)  =  \norm{B^\sg}_1 + O\left(\frac{1}{m}\right) ,
\end{align}
where we use \eqref{eq:condition-3}, \eqref{eq:condition-4}, \eqref{eq:condition-7} and \eqref{eq:condition-8} in the second step and \eqref{eq:B-single} in the last step, 
and 
\begin{equation}
    \sum_{i=1}^r\bra{D_{1,i}^{(\ell)}}  B \ket{D_{0,i}^{(\ell')}} = O\left(\frac{1}{m}\right),\quad \text{when }\ell\neq\ell'.
\end{equation}
where we use \eqref{eq:condition-3}, \eqref{eq:condition-4}, \eqref{eq:condition-7} and \eqref{eq:condition-8}. 

Then we have 
\begin{align}
\trace(BV) = \sum_{i,\ell} \bra{e_{i,\ell}} B \ket{f_{i,\ell}}= m \norm{B^\sg}_{1} + O(1),
\end{align}
proving the first equation in \eqref{eq:lower-2-sql}. 

Meanwhile, 
\begin{align}
\norm{V}_\infty^2 &= \big\|V^\dagger V\big\|_\infty = \bigg\|\sum_{i,\ell,j,\ell'} \ket{e_{i,\ell}}\braket{f_{i,\ell}|f_{j,\ell'}}\bra{e_{j,\ell'}}\bigg\|_\infty\\
&= \bigg\|\sum_{i,j,\ell} \left(\delta_{ij} + \epsilon_{1,ij}\right) \ket{e_{i,\ell}}\bra{e_{j,\ell}}\bigg\|_\infty\\
&\leq \max_{\ell \in [1,m]} \left( \Big\|\sum_{i} \ket{e_{i,\ell}}  \bra{e_{i,\ell}}\Big\|_\infty + \Big\|\sum_{i,j} \epsilon_{1,ij} \ket{e_{i,\ell}}  \bra{e_{j,\ell}}\Big\|_\infty \right), 
\end{align}
where in the last step we use the fact that $\ket{e_{i,\ell}}$ and $\ket{e_{j,\ell'}}$ are strictly orthogonal to each other for any $\ell \neq \ell'$ and $i,j$. These two terms (for any $\ell$) can be bounded through the following. Choosing an arbitrary orthonormal set of vectors $\{\ket{\widehat{\v{e}}_i}\}_{i=1}^m$, the first term satisfies
\begin{align}
\Big\|\sum_{i} \ket{e_{i,\ell}}  \bra{e_{i,\ell}}\Big\|_\infty 
&= \Big\|\Big(\sum_{i} \ket{e_{i,\ell}}\bra{\widehat{\v{e}}_i}\Big)\Big(\sum_{j} \ket{\widehat{\v{e}}_j}\bra{e_{j,\ell}}\Big)\Big\|_\infty \\
&= \Big\| \Big(\sum_{i} \ket{\widehat{\v{e}}_i}\bra{e_{i,\ell}}\Big)\Big(\sum_{j} \ket{e_{j,\ell}}\bra{\widehat{\v{e}}_j}\Big) \Big\|_\infty \\
&= \Big\| \sum_{ij} \left( \delta_{ij} + \epsilon_{0,ij} \right) \ket{\widehat{\v{e}}_i}\bra{\widehat{\v{e}}_j}\Big\|_\infty \leq 1 + O\left(\frac{1}{m}\right),
\end{align}
where we use $\norm{A^\dagger A}_\infty = \norm{AA^\dagger}_\infty$ in the second step. The second term satisfies 
\begin{align}
\Big\|\sum_{i,j} \epsilon_{1,ij} \ket{e_{i,\ell}}  \bra{e_{i',\ell}}\Big\|_\infty 
&= \max_{\text{unit vector }\ket{v}}  \abs{ \sum_{i,j} \epsilon_{1,ij}   \braket{v|e_{i,\ell}}  \braket{e_{j,\ell}|v} } \leq \sum_{i,j} \abs{\epsilon_{1,ij} }\abs{\delta_{ij} + \epsilon_{0,ij}} = O\left(\frac{1}{m}\right). 
\end{align}
Therefore, 
\begin{equation}
\norm{V}_\infty^2 \leq 1 + O\left(\frac{1}{m}\right),
\end{equation}
proving the second inequality in \eqref{eq:lower-2-sql}. 

\end{proof}

Using \lemmaref{lemma:V-small-ancilla-sql} and \eqref{eq:gamma-first-step}, we have 
\begin{equation}
\label{eq:gamma-scaling-SQL}
    \gamma_\tl = m \gamma_\tl^\sg - \norm{B}_1 + m \norm{B^\sg}_1 + O(1) = m \gamma_\tl^\sg + O(1).
\end{equation}
Moreover, by direct calculations, we have 
\begin{equation}
\label{eq:hamt-scaling-SQL}
\begin{split}
    \trace\Big(\sum_\ell H^{(\ell)} Z_\tl\Big) 
    &= m \sum_{i=0}^{d-1} \left(\frac{m^{[0]}_i}{m} - \frac{m^{[1]}_i}{m} \right)\bra{i}H\ket{i}\\
    &= m \sum_{i=0}^{d-1} \left(\lambda^{[0]}_i - \lambda^{[1]}_i \right)\bra{i}H\ket{i} + O(1)
    = m \trace(H Z^\sg_\tl) + O(1). 
\end{split}
\end{equation}
\thmref{thm:near-optimal} is then proven for the small-ancilla code$^\sql$. 
Finally, we remark that $\ket{e_{i,\ell}}$ and $\ket{f_{i,\ell}}$ in the proof above generate an asymptotically optimal recovery channel $\mR_{\rm opt}$ of the form in \eqref{eq:recovery} and \eqref{eq:optimal-recovery} for achieving the optimal SQL using the small-ancilla code$^\sql$, similar to the case of achieving the optimal HL using the small-ancilla code. 
}

\subsubsection{The qubit-ancilla random code\texorpdfstring{$^\sql$}{}}

Similar to the ancilla-free random code defined in the main text for achieving the optimal HL when HNLS holds, here we define a qubit-ancilla random code{$^\sql$} that achieves the optimal SQL when HNLS fails. The codewords, acting on $m\geq 3$ probes, are defined as 
\begin{equation}
\label{eq:code-qubit-ancilla-SQL}
\begin{split}
    \ket{0_\tl} &= \frac{1}{\sqrt{\abs{W^\sql_0}}} \sum_{w \in W^\sql_0} e^{i\theta^\sql_w} \ket{\phi_w}_{P^{\otimes m}}  \ket{0}_{A},\\
    \ket{1_\tl} &= \frac{1}{\sqrt{\abs{W^\sql_1}}} \sum_{w \in W^\sql_1} e^{i\theta^\sql_w} \ket{\varphi_w}_{P^{\otimes m}}  \ket{1}_{A},
\end{split}
\end{equation}
where $\{\theta^\sql_w\}_{w\in W_0^\sql\cup W_1^\sql}$ are sampled from a set of independent and identically distributed random variables following the uniform distribution in $[0,2\pi)$. Note that here $\dim(\mH_A) = 2$, i.e., one qubit ancilla is needed. 

In order to prove \thmref{thm:near-optimal} for the ancilla-free random code$^\sql$, we first note that \eqref{eq:hamt-scaling-SQL} holds true for the ancilla-free random code$^\sql$, because
\begin{equation}
    \trace_{\backslash\{\ell\}}(\ket{k_\tl}\bra{k_\tl}) = \sum_{i=0}^{d-1} \frac{m^{[k]}_i}{m}\ket{i}\bra{i}, \quad \text{for } k = 0,1.
\end{equation} 
Thus, it only remains to show 
\begin{equation}
    \gamma_\tl = m \gamma_\tl^\sg - \norm{B}_1 + m \norm{B^\sg}_1 + O(1) = m \gamma_\tl^\sg + O(1). 
\end{equation}
For $\ell \neq \ell'$ and $k = 0,1$, let  
\begin{gather}
    \trace_{\backslash\{\ell,\ell'\}}(\ket{k_\tl}\bra{k_\tl}) = \sum_{i,j=0}^{d-1} m^{[k]}_{ij} \ket{ij}\bra{ij} + \sum_{\substack{i\neq j,\\i,j=0}}^{d-1}\Delta^{(\ell,\ell')}_{k,ij} \ket{ij}\bra{ji}, 
\end{gather}
where 
\begin{equation}
m^{[k]}_{ij} = 
\frac{1}{m(m-1)} \times
\begin{cases}
m^{[k]}_i m^{[k]}_j, & i\neq j,\\
m^{[k]}_i (m^{[k]}_i - 1),& i = j. 
\end{cases}
\end{equation}
and $\Delta^{(\ell,\ell')}_{k,ij}$ are defined for all $k$, $\ell\neq \ell'$ and $i\neq j$. They will contribute to some additional terms when computing $\norm{B}_1$, compared to the small-ancilla case. We will show that these contributions are negligibly small with high probability. 

\sisinewlong{
First, we prove a lemma (similar to \lemmaref{lemma:Delta} in the HL case) that guarantees the smallness of $\Delta^{(\ell,\ell')}_{k,ij}$ with high probability. 

\begin{lemma}
\label{lemma:Delta-sql}
$\Delta^{(\ell,\ell')}_{0,ij}$ (or $\Delta^{(\ell,\ell')}_{1,ij}$) defined above for all $\ell \neq \ell'$ and $i\neq j \in [0,d_0-1]$ (or $i\neq j \in [d_0,d-1]$) satisfies
\begin{equation}
    \prob\left(\abs{\Delta_{k,ij}^{(\ell,\ell')}} > \frac{1}{m^3},\exists k,i,j,\ell,\ell'\right) 
    = e^{-\Omega(m)}.
\end{equation}
\end{lemma}

\begin{proof}
For $\ell \neq \ell'$, 
\begin{gather}
    \trace_{\backslash\{\ell,\ell'\}}(\ket{k_\tl}\bra{k_\tl}) = \sum_{i,j=0}^{d-1} m^{[k]}_{ij} \ket{ij}\bra{ij} + \sum_{\substack{i\neq j,\\i,j=0}}^{d-1}\Delta^{(\ell,\ell')}_{k,ij} \ket{ij}\bra{ji},
\end{gather}
where for $k=0,1$ and $i\neq j$, 
\begin{equation}
    \Delta^{(\ell,\ell')}_{k,ij} = \frac{1}{\abs{W_k^\sql}}\sum_{\widetilde{w} \in \widetilde{W}_{k,ij}^\sql} e^{-i \theta^\sql_{ij\widetilde{w}}} e^{i \theta^\sql_{ji\widetilde{w}}},
\end{equation}
and $\widetilde{W}_{0,ij}^\sql$ (or $\widetilde{W}_{1,ij}^\sql$) for $i\neq j$ and $i,j \in \{0,\cdots,d-1\}$ is the set of strings of length $m-2$ which contains $m^{[k]}_i-1$ $i$'s, $m^{[k]}_j-1$ $j$'s and $m^{[k]}_{l}$ $l$'s for all $l \neq i,j$ and $l \in \{0,\cdots,d-1\}$. 
We have
\begin{gather}
    \bE[\Delta^{(\ell,\ell')}_{k,ij}] = \frac{1}{\abs{W_k^\sql}}\sum_{\widetilde{w} \in \widetilde{W}_{k,ij}^\sql} \bE[e^{-i \theta^\sql_{ij\widetilde{w}}}] \bE[e^{i \theta^\sql_{ji\widetilde{w}}}] = 0, 
    \\
    \bE[\abs{\Delta^{(\ell,\ell')}_{k,ij}}^2] = \frac{1}{\abs{W_k^\sql}^2}\sum_{\widetilde{w},\widetilde{w}' \in \widetilde{W}_{k,ij}^\sql} \bE[e^{-i \theta^\sql_{ij\widetilde{w}}}e^{i \theta^\sql_{ji\widetilde{w}}}e^{i \theta^\sql_{ij\widetilde{w}'}}e^{-i \theta^\sql_{ji\widetilde{w}'}}]  = \frac{\abs{\widetilde{W}_{k,ij}^\sql}}{\abs{W_k^\sql}^2} \leq \frac{1}{\abs{W_k^\sql}}. 
\end{gather}
Using the Chebyshev's inequality, we have for all $i\neq j$, 
$    \prob\left(\Re[\Delta^{(\ell,\ell')}_{k,ij}] > \frac{1}{\sqrt{2} m^3}\right)  \leq\frac{2m^6}{\abs{W_k}}$, 
$    \prob\left(\Im[\Delta^{(\ell,\ell')}_{k,ij}] > \frac{1}{\sqrt{2}m^3}\right)\leq \frac{2m^6}{\abs{W_k^\sql}}$.
Then the union bound implies, 
\begin{equation}
    \prob\left(\abs{\Delta^{(\ell,\ell')}_{k,ij}} > \frac{1}{m^3}\right) \leq \frac{4m^6}{\abs{W_k^\sql}}.
\end{equation}
The total number of $(k,i,j,\ell,\ell')$ is 
\begin{equation}
    \frac{m(m-1)}{2}{d(d-1)}{} + \frac{m(m-1)}{2}{d(d-1)}{} \leq m^2 d^2. 
\end{equation}
Using the union bound again, we have 
\begin{equation}
    \prob\left(\abs{\Delta_{k,ij}^{(\ell,\ell')}} > \frac{1}{m^3},\exists k,i,j,\ell,\ell'\right) \leq \frac{4m^6}{\min\{\abs{W_0^\sql},\abs{W_1^\sql}\}}m^2d^2 = e^{-\Omega(m)},
\end{equation}
(Note that above we use $\abs{W_0^\sql} = \frac{m!}{m^{[0]}_0! m^{[0]}_1! \cdots m^{[0]}_{d-1}!} = e^{\Omega(m)}$ and $\abs{W_1^\sql} = \frac{m!}{m^{[1]}_0! m^{[1]}_1! \cdots m^{[1]}_{d-1}!} = e^{\Omega(m)}$.)

\end{proof}

Now we prove \lemmaref{lemma:V-ancilla-free-sql} which directly leads to \thmref{thm:near-optimal}.

\begin{lemma}
\label{lemma:V-ancilla-free-sql}
For the ancilla-free random code$^\sql$ and $B = \sum_{i=1}^r  \sum_{\ell=1}^m(\id-\ket{0_\tl}\bra{0_\tl})(L_i^{(\ell)} \ket{0_\tl}\bra{1_\tl} L_i^{(\ell)\dagger} )(\id-\ket{1_\tl}\bra{1_\tl})$,, there exists a matrix $V$ such that 
\begin{equation}
\label{eq:lower-4-sql}
\abs{\trace(BV)} = m \norm{B^\sg}_1 + O(1), \quad \text{and}\quad \norm{V}_\infty \leq 1 + O\left(\frac{1}{m}\right),
\end{equation}
with probability $1 - e^{-\Omega(m)}$. Then 
\begin{equation}
    \norm{B}_1 \geq \frac{\abs{\trace(BV)}}{\norm{V}_\infty} \geq m \norm{B^\sg}_1 + O(1). 
\end{equation}
\end{lemma}

\begin{proof}
Here we assume for all $k,i\neq j,\ell\neq\ell'$, 
\begin{equation}
    \abs{\Delta_{k,ij}^{(\ell,\ell')}} < \frac{1}{m^3},
\end{equation}
which is true with probability $1 - e^{-\Omega(m)}$ from \lemmaref{lemma:Delta-sql}.

Let 
\begin{equation}
V = \sum_{i=1}^r \sum_{\ell=1}^{m-1} 
\ket{f_{i,\ell}} \bra{e_{i,\ell}},
\end{equation}
where for $i = 1,\ldots,r$ and $\ell = 1,\ldots, m-1$, 
\begin{equation}
\ket{e_{i,\ell}} = \frac{1}{\sqrt{m}}\sum_{\ell'=1}^{m} \exp\left(-i \frac{2\pi}{m}\ell \ell'\right) \ket{D_{0,i}^{(\ell')}},\quad 
\ket{f_{i,\ell}} = \frac{1}{\sqrt{m}}\sum_{\ell'=1}^{m} \exp\left(-i \frac{2\pi}{m}\ell \ell'\right) \ket{D_{1,i}^{(\ell')}},
\end{equation}
and
\begin{equation}
\ket{D^{(\ell)}_{0,i}} = \sum_{j=1}^{r} D_{0,ij} (\id - \ket{0_\tl}\bra{0_\tl}) L_j^{(\ell)}\ket{0_\tl}
,\quad 
\ket{D^{(\ell)}_{1,i}} = \sum_{j=1}^{r} D_{1,ij} (\id - \ket{1_\tl}\bra{1_\tl}) L_j^{(\ell)}\ket{1_\tl}. 
\end{equation}
Note that for the ancilla-free random code$^\sql$, 
\begin{gather}
 \braket{0_\tl|L_i^{(\ell)}|0_\tl} - \braket{0^\sg_\tl|L_i|0^\sg_\tl}
= \trace\left(\sum_{i=0}^{d-1} \left(\frac{m^{[0]}_i}{m}-\lambda^{[0]}_i\right) \ket{i}\bra{i} \cdot L_i\right) = O\left(\frac{1}{m}\right),\\
 \braket{1_\tl|L_i^{(\ell)}|1_\tl} - \braket{1^\sg_\tl|L_i|1^\sg_\tl}
= \trace\left(\sum_{i=0}^{d-1} \left(\frac{m^{[1]}_i}{m}-\lambda^{[1]}_i\right) \ket{i}\bra{i} \cdot L_i\right) = O\left(\frac{1}{m}\right),\\
 \braket{0_\tl|L_i^{(\ell)\dagger}L_j^{(\ell)}|0_\tl} - \braket{0^\sg_\tl|L_i^{\dagger}L_j|0^\sg_\tl} = \trace\left(\sum_{i=0}^{d-1} \left(\frac{m^{[0]}_i}{m}-\lambda^{[0]}_i\right) \ket{i}\bra{i} \cdot L_i^{\dagger}L_j\right)   =  O\left(\frac{1}{m}\right),\\
 \braket{1_\tl|L_i^{(\ell)\dagger}L_j^{(\ell)}|1_\tl} - \braket{1^\sg_\tl|L_i^{\dagger}L_j|1^\sg_\tl} = \trace\left(\sum_{i=0}^{d-1} \left(\frac{m^{[1]}_i}{m}-\lambda^{[1]}_i\right) \ket{i}\bra{i} \cdot L_i^{\dagger}L_j\right) =  O\left(\frac{1}{m}\right).
\end{gather}
For all $\ell \neq \ell'$, 
\begin{gather}
\label{eq:condition-1-r}
\braket{D^{(\ell)}_{0,i}| D^{(\ell)}_{0,j}}  = 
\sum_{i'j'} D_{0,ii'}^* D_{0,jj'} \bra{0_\tl}L_{i'}^{(\ell)\dagger} (\id - \ket{0_\tl}\bra{0_\tl}) L_{j'}^{(\ell)}\ket{0_\tl} = \delta_{ij} + O\left(\frac{1}{m}\right), 
\\
\label{eq:condition-2-r}
\braket{D^{(\ell)}_{1,i}| D^{(\ell)}_{1,j}}  = 
\sum_{i'j'} D_{1,ii'}^* D_{1,jj'} \bra{1_\tl}L_{i'}^{(\ell)\dagger} (\id - \ket{1_\tl}\bra{1_\tl}) L_{j'}^{(\ell)}\ket{1_\tl} = \delta_{ij} + O\left(\frac{1}{m}\right), 
\\
\label{eq:condition-3-r}
\bra{D^{(\ell)}_{0,i}}(\id - \ket{0_\tl}\bra{0_\tl}) L_j^{(\ell)} \ket{0_\tl} 
= \sum_{i'} D^{*}_{0,ii'}  \bra{0^\sg_\tl} L^{\dagger}_{i'} (\id - \ket{0^\sg_\tl}\bra{0^\sg_\tl}) L_j \ket{0^\sg_\tl} + O\left(\frac{1}{m}\right),
\\
\label{eq:condition-4-r}
\bra{1_\tl} L_j^{(\ell)\dagger} (\id - \ket{1_\tl}\bra{1_\tl}) \ket{D^{(\ell)}_{1,i}}
= \sum_{i''} D_{1,ii''}  \bra{1^\sg_\tl} L^{\dagger}_{j} (\id - \ket{1^\sg_\tl}\bra{1^\sg_\tl}) L_{i''} \ket{1^\sg_\tl} + O\left(\frac{1}{m}\right),
\end{gather}
{\small
\begin{gather}
\label{eq:condition-5-r}
\begin{split}
&\quad\;\braket{D^{(\ell)}_{0,i}| D^{(\ell')}_{0,j}} \\ &= 
\sum_{i'j'} D_{0,ii'}^* D_{0,jj'} \bra{0_\tl}L_{i'}^{(\ell)\dagger} (\id - \ket{0_\tl}\bra{0_\tl}) L_{j'}^{(\ell')}\ket{0_\tl} \\
&= \sum_{i'j'} D_{0,ii'}^* D_{0,jj'} \left(\frac{ \braket{0_\tl|L_{i'}^{(\ell)\dagger}|0_\tl}  \braket{0_\tl|L_{j'}^{(\ell)}|0_\tl} }{m-1} - \sum_{k=0}^{d-1} \frac{m^{[0]}_k}{m(m-1)} \bra{k}L_{i'}^{\dagger}\ket{k}\bra{k}L_{j'}\ket{k} + \sum_{\substack{k\neq k',\\ k,k'=0}}^{d-1} \Delta^{(\ell,\ell')}_{0,kk'} \bra{k'}L_{i}^{\dagger}\ket{k}\bra{k}L_{j}\ket{k'}\right)\\
&=: \eta_{0,ij} + \teta^{(\ell,\ell')}_{0,ij} = O\left(\frac{1}{m}\right) + \teta^{(\ell,\ell')}_{0,ij},
\end{split}\\
\label{eq:condition-6-r}
\begin{split}
&\quad\;\braket{D^{(\ell)}_{1,i}| D^{(\ell')}_{1,j}} \\&= 
\sum_{i'j'} D_{1,ii'}^* D_{1,jj'} \bra{1_\tl}L_{i'}^{(\ell)\dagger} (\id - \ket{1_\tl}\bra{1_\tl}) L_{j'}^{(\ell')}\ket{1_\tl} \\
&= \sum_{i'j'} D_{1,ii'}^* D_{1,jj'} \left(\frac{\braket{1_\tl|L_{i'}^{(\ell)\dagger}|1_\tl}  \braket{1_\tl|L_{j'}^{(\ell)}|1_\tl} }{m-1} - \sum_{k=0}^{d-1} \frac{m^{[1]}_k}{m(m-1)} \bra{k}L_{i'}^{\dagger}\ket{k}\bra{k}L_{j'}\ket{k} + \sum_{\substack{k\neq k',\\ k,k'=0}}^{d-1} \Delta^{(\ell,\ell')}_{1,kk'} \bra{k'}L_{i}^{\dagger}\ket{k}\bra{k}L_{j}\ket{k'}
\right)\\
&=: \eta_{1,ij} + \teta^{(\ell,\ell')}_{1,ij} = O\left(\frac{1}{m}\right) + \teta^{(\ell,\ell')}_{1,ij}, \end{split}
\\
\label{eq:condition-7-r}
\begin{split}
&\quad\;\bra{D^{(\ell)}_{0,i}}(\id - \ket{0_\tl}\bra{0_\tl}) L_j^{(\ell')} \ket{0_\tl} 
= \sum_{i'} D^{*}_{0,ii'} \bra{0_\tl} L^{(\ell)\dagger}_{i'} (\id - \ket{0_\tl}\bra{0_\tl}) L_j^{(\ell')} \ket{0_\tl} \\
&= \sum_{i'} D_{0,ii'}^* \left(\frac{ \braket{0_\tl|L_{i'}^{(\ell)\dagger}|0_\tl}  \braket{0_\tl|L_{j}^{(\ell)}|0_\tl} }{m-1} \!-\! \sum_{k=0}^{d-1} \frac{m^{[0]}_k}{m(m-1)} \bra{k}L_{i'}^{\dagger}\ket{k}\bra{k}L_{j'}\ket{k} \!+\! \sum_{\substack{k\neq k',\\ k,k'=0}}^{d-1} \Delta^{(\ell,\ell')}_{1,kk'} \bra{k'}L_{i}^{\dagger}\ket{k}\bra{k}L_{j}\ket{k'}\right) = O\left(\frac{1}{m}\right), 
\end{split}
\\
\label{eq:condition-8-r}
\begin{split}
&\quad\;\bra{1_\tl} L_j^{(\ell)\dagger} (\id - \ket{1_\tl}\bra{1_\tl}) \ket{D^{(\ell')}_{1,i}} = \sum_{i''} D_{1,ii''} \bra{1_\tl} L^{(\ell)\dagger}_{j} (\id - \ket{1_\tl}\bra{1_\tl}) L_{i''}^{(\ell')} \ket{1_\tl} \\
&=\sum_{j'}  D_{1,jj'} \left(\frac{ \braket{1_\tl|L_{i}^{(\ell)\dagger}|1_\tl}  \braket{1_\tl|L_{j'}^{(\ell)}|1_\tl} }{m-1} \!-\! \sum_{k=0}^{d-1} \frac{m^{[1]}_k}{m(m-1)} \bra{k}L_{i}^{\dagger}\ket{k}\bra{k}L_{j'}\ket{k} \!+\! \sum_{\substack{k\neq k',\\ k,k'=0}}^{d-1} \Delta^{(\ell,\ell')}_{1,kk'} \bra{k'}L_{i}^{\dagger}\ket{k}\bra{k}L_{j}\ket{k'}\right) \!=\!  O\left(\frac{1}{m}\right),
\end{split}
\end{gather}}where $\eta_{0,ij}$ (or $\eta_{1,ij}$) corresponds to the first two terms in the third line of \eqref{eq:condition-5-r} (or \eqref{eq:condition-6-r}) that is independent of the specific choice of $\ell$ and $\ell'$; and $\teta^{(\ell,\ell')}_{0,ij}$ (or $\teta^{(\ell,\ell')}_{1,ij}$) corresponds to the last term in the third line of \eqref{eq:condition-5-r} (or \eqref{eq:condition-6-r}) . 

Using the mathematical relations above, we know, for all $1 \leq \ell \leq m - 1$,  
\begin{equation}
\braket{e_{i,\ell}|e_{j,\ell'}}  = \delta_{\ell\ell'} \left(\delta_{ij}+\tepsilon_{0,ij}\right) + \tepsilon^{(\ell,\ell')}_{0,ij},\quad 
\braket{f_{i,\ell}|f_{j,\ell'}}  = \delta_{\ell\ell'} \left(\delta_{ij}+\tepsilon_{1,ij}\right) + \tepsilon^{(\ell,\ell')}_{1,ij},
\end{equation}
where for $k = 0,1$ and all $i,j$, 
\begin{gather}
    \epsilon_{k,ij} := \braket{D^{(\ell)}_{k,i}| D^{(\ell)}_{k,j}} - \eta_{k,ij} - \delta_{ij}  = O\left(\frac{1}{m}\right) , \\ 
\tepsilon^{(\ell,\ell')}_{k,ij} := \frac{1}{m} \sum_{l,l'=1}^m \exp\left({-i\frac{2\pi}{m} (\ell l-\ell' l') } \right) \teta^{(l,l')}_{k,i,j} = O\left(\frac{1}{m^2}\right). 
\end{gather}
Note that the new terms $\tepsilon^{(\ell,\ell')}_{k,ij}$ appear in the case of the ancilla-free random code, compared to the previous small-ancilla case. However, as we will show later, the contribution of $\tepsilon^{(\ell,\ell')}_{k,ij}$ is vanishingly small. 

First, by direct calculations and using the definitions of $\ket{e_{i,\ell}}$ and $\ket{f_{i,\ell}}$, we have 
\begin{equation}
V = \sum_{i=1}^{r} \sum_{\ell=1}^{m-1} \ket{f_{i,\ell}}\bra{e_{i,\ell}} 
= \sum_{i=1}^{r} \left( \sum_{\ell=1}^m \frac{m-1}{m}\ket{D_{1,i}^{(\ell)}}  \bra{D_{0,i}^{(\ell)}} + \sum_{\ell\neq \ell',\ell,\ell'=1}^m \frac{-1}{m}\ket{D_{1,i}^{(\ell)}}  \bra{D_{0,i}^{(\ell')}} \right). 
\end{equation}
Furthermore, from \eqref{eq:condition-3-r}, \eqref{eq:condition-4-r}, \eqref{eq:condition-7-r} and \eqref{eq:condition-8-r}, we have 
\begin{equation}
    \sum_{i=1}^r\bra{D_{1,i}^{(\ell)}}  B \ket{D_{0,i}^{(\ell)}}
    = \norm{B^\sg}_1 + O\left(\frac{1}{m}\right), \quad
    \text{and}\quad
    \sum_{i=1}^r\bra{D_{1,i}^{(\ell)}}  B \ket{D_{0,i}^{(\ell')}}
    = O\left(\frac{1}{m}\right),\quad \text{when }\ell\neq\ell'.
\end{equation}
Then we have 
\begin{align}
\trace(BV) = \sum_{i,\ell} \bra{e_{i,\ell}} B \ket{f_{i,\ell}}= m \norm{B^\sg}_1 + O(1),
\end{align}
proving the first equation in \eqref{eq:lower-4-sql}.

Meanwhile, we want to prove an upper bound of $1 + O(1/m)$ on 
\begin{align}
\norm{V}_\infty^2 &= \big\|V^\dagger V\big\|_\infty = \bigg\|\sum_{i,\ell,j,\ell'} \ket{e_{i,\ell}}\braket{f_{i,\ell}|f_{j,\ell'}}\bra{e_{j,\ell'}}\bigg\|_\infty\\ 
&= \bigg\|\sum_{i,j,\ell,\ell'} \left(\delta_{\ell\ell'} \left(\delta_{ij} + \epsilon_{0,ij} \right) + \tepsilon^{(\ell,\ell')}_{1,ij} \right) \ket{e_{i,\ell}}\bra{e_{j,\ell'}}\bigg\|_\infty\\ 
& = \max_{\text{unit vector }\ket{v}} \sum_{i,j,\ell,\ell'} \left(\delta_{\ell\ell'} \left( \delta_{ij} + \epsilon_{1,ij} \right) + \tepsilon^{(\ell,\ell')}_{1,ij} \right) \braket{v|e_{i,\ell}}\braket{e_{j,\ell'}|v}. 
\end{align}
Let $\ket{v} = \sum_{i,\ell} v_{i,\ell} \ket{e_{i,\ell}}$ be an arbitrary unitary vector. Then using 
\begin{align}
    1 &= \braket{v|v} = \sum_{i,j,\ell,\ell'}  v^*_{i,\ell}v_{j,\ell'} \braket{e_{i,\ell}|e_{j,\ell'}} = 
    \sum_{i,\ell}  |v_{i,\ell}|^2  
    + \sum_{i,j,\ell}  v^*_{i,\ell}v_{j,\ell} \epsilon_{0,ij} + 
    \sum_{i,j,\ell,\ell'}  v^*_{i,\ell}v_{j,\ell'}  \tepsilon^{(\ell,\ell')}_{0,ij}\\
    & \geq \sum_{i,\ell}  |v_{i,\ell}|^2  
    - \sum_{i,j,\ell}  \frac{|v_{i,\ell}|^2 + |v_{j,\ell}|^2}{2}  \abs{\epsilon_{0,ij}} - 
    \sum_{i,j,\ell,\ell'}  \frac{|v_{i,\ell}|^2 + |v_{j,\ell'}|^2}{2}   |\tepsilon^{(\ell,\ell')}_{0,ij}|\\
    & =: \sum_{i,\ell}  |v_{i,\ell}|^2  \left(1 + \nu_{i,\ell}\right), \quad \text{where }|\nu_{i,\ell}| = O\left(\frac{1}{m}\right), 
\end{align}
and 
\begin{align}
&\quad \sum_{i,j,\ell,\ell'} \left(\delta_{\ell\ell'} \left( \delta_{ij} + \epsilon_{1,ij}\right)  + \tepsilon^{(\ell,\ell')}_{1,ij}\right) \braket{v|e_{i,\ell}}\braket{e_{j,\ell'}|v} \\
&= 
\sum_{i,j,\ell,\ell',k,k',q,q'} v_{k,q}^* v_{k',q'}  \left(\delta_{\ell\ell'} \left( \delta_{ij} + \epsilon_{1,ij} \right)  + \tepsilon^{(\ell,\ell')}_{1,ij} \right) 
\left(\delta_{q\ell} \left( \delta_{ki} + \epsilon_{0,ki} \right)  + \tepsilon^{(q,\ell)}_{0,ki}\right)\left(\delta_{\ell'q'} \left( \delta_{jk'} + \epsilon_{0,jk'} \right)  + \tepsilon^{(\ell',q')}_{0,jk'}\right)\\
&\leq \sum_{k,k',q,q'}\frac{|v_{k,q}|^2 + |v_{k',q'}|^2}{2} \Bigg| \sum_{i,j,\ell,\ell'} \left(\delta_{\ell\ell'}  \left( \delta_{ij} + \epsilon_{1,ij} \right)  + \tepsilon^{(\ell,\ell')}_{1,ij} \right) 
\left(\delta_{q\ell} \left( \delta_{ki} + \epsilon_{0,ki} \right)  + \tepsilon^{(q,\ell)}_{0,ki}\right) 
\nonumber \\ & \qquad \qquad \qquad\qquad \qquad \qquad  \qquad \qquad \qquad  \qquad \qquad \qquad  \qquad \qquad \qquad  \times \left(\delta_{\ell'q'} \left( \delta_{jk'} + \epsilon_{0,jk'} \right)  + \tepsilon^{(\ell',q')}_{0,jk'}\right) \Bigg|\\
&=: \sum_{i,\ell} |v_{i,\ell}|^2 \left(1 + \tilde\nu_{i,\ell}\right), \quad \text{where }|\tilde\nu_{i,\ell}| = O\left(\frac{1}{m}\right), 
\end{align}
we have 
\begin{align}
    \norm{V}_\infty^2 
    &\leq \max_{\text{unit vector }\ket{v}} \sum_{i,\ell} |v_{i,\ell}|^2 \left(1 + \tilde\nu_{i,\ell}\right) \leq  \max_{\text{unit vector }\ket{v}} \sum_{i,\ell} |v_{i,\ell}|^2 \left(1 + \tilde\nu_{i,\ell} + \nu_{i,\ell} - \nu_{i,\ell}\right)\\
    &\leq 1 + \max_{\text{unit vector }\ket{v}} \sum_{i,\ell} |v_{i,\ell}|^2 \abs{ 1 + \nu_{i,\ell}}  \abs{ \frac{\tilde\nu_{i,\ell}- \nu_{i,\ell}}{1 + \nu_{i,\ell}} } = 1 + O\left(\frac{1}{m}\right),
\end{align}
where the last step follows from $\tilde\nu_{i,\ell} = O(1/m)$, 
$\nu_{i,\ell} = O(1/m)$ and $\sum_{i,\ell} |v_{i,\ell}|^2 (1 + \nu_{i,\ell} ) = 1$. 
\end{proof}

Using \lemmaref{lemma:V-ancilla-free-sql} and \eqref{eq:gamma-first-step}, we have 
\begin{equation}
    \gamma_\tl = m \gamma_\tl^\sg - \norm{B}_1 + m \norm{B^\sg}_1 + O(1) = m \gamma_\tl^\sg + O(1),
\end{equation}
proving \thmref{thm:near-optimal} for the ancilla-free random code$^\sql$. 

Finally, we remark that $\ket{e_{i,\ell}}$ and $\ket{f_{i,\ell}}$ in the proof above generate an asymptotically optimal recovery channel $\mR_{\rm opt}$ of the form in \eqref{eq:recovery} and \eqref{eq:optimal-recovery-random} for achieving the optimal SQL using the ancilla-free random code$^\sql$, similar to the case of achieving the optimal HL using the ancilla-free random code. 
}

\sisilong{
\section{Achieving the optimal HL: Beyond Hamiltonian estimation under Markovian noise}
\label{app:beyond}

In the main text and the appendixes above, we assumed that the unknown parameter $\omega$ to be estimated is the linear factor of the Hamiltonian and that the Lindblad operators are exactly known. In fact, when $H \in \mS$ our discussion can be generalized naturally to cases where both the Hamiltonian and the Lindblad operators are functions of the unknown parameter. The code construction and the proof of optimality directly follow from our previous discussion. We also remark here that when $H \in \mS$, although we expect that a similar generalization holds, different proof techniques are needed and we will leave it for future work. 

In this appendix, we will first review previous results on the QFI upper bounds in noisy quantum metrology in general cases. Then we will show that the small-ancilla code and the ancilla-free random code from our previous constructions attain the optimal HL even in general cases, beyond Hamiltonian estimation under Markovian noise. 

\subsection{Upper bound on the QFI}
\label{app:notation-beyond}

We consider the following evolution of an open quantum system
\begin{equation}
\label{eq:generalized-setting}
\frac{d\rho}{dt} = -i[E(\omega),\rho] + \sum_{i=1}^r L_i(\omega) \rho L_i(\omega)^\dagger - \frac{1}{2}\{L_i(\omega)^\dagger L_i(\omega),\rho\}, 
\end{equation}
where $E(\omega)$ is the Hamiltonian and $L_i(\omega)$ are Lindblad operators. Both of them are functions of $\omega$. (Note that we use letter ``$E$'' to represent the Hamiltonian here to distinguish it from ``$H$'' that we have used in the main text to denote the derivative of the Hamiltonian with respect to $\omega$.) Previous works have been done to derive the upper bounds of the QFI under the evolution above. First, from Eq.~(F11) in~\cite{kurdzialek2022using}, we know that the QFI of a quantum state at time $T$ has the following upper bound in the scenario of adaptive quantum strategies (see \figaref{fig:strategy}, where $\mE_\omega(\rho) = \rho + \left( -i[E(\omega),\rho] + \sum_{i=1}^r L_i(\omega) \rho L_i(\omega)^\dagger - \frac{1}{2}\{L_i(\omega)^\dagger L_i(\omega),\rho\}\right) dt + O(dt^2)$) 
\begin{equation}
\frac{dF(\rho_\omega(T))}{dT} \leq 4 \min_{\substack{h_{00} \in \bR,\,\vh \in \bC^r,\\ \text{~Hermitian~}\frakh \in \bC^{r \times r}}}\left( \norm{O_1} + \norm{O_2} \sqrt{F(\rho_\omega(T))} \right), 
\end{equation}
where
\begin{equation}
O_1 = \sum_i \bigg(\vh_i \id + \sum_j \frakh_{ij} L_j(\omega) - i \partial_\omega  L_i(\omega)\bigg)^\dagger \bigg(\vh_i \id +  \sum_j \frakh_{ij} L_j(\omega) - i \partial_\omega L_i(\omega)\bigg), 
\end{equation}
\begin{equation}
O_2 = \partial_\omega E(\omega) - \frac{i}{2} \sum_{i} \left( \partial_\omega L_i(\omega)^\dagger L_i(\omega) - L_i(\omega) \partial_\omega L_i(\omega)\right) - h_{00} \id - \sum_{j}(\vh_j^* L_j(\omega) + \vh_j L_j(\omega)^\dagger) - \sum_{ij} \frakh_{ij}L_i(\omega)^\dagger L_j(\omega). 
\end{equation}
and calculations from Eq.~(C17) in \cite{wan2022bounds} indicate that 
\begin{equation}
F(\rho_\omega(T)) \leq 4 \min_{\substack{h_{00} \in \bR,\,\vh \in \bC^r,\\ \text{~Hermitian~}\frakh \in \bC^{r \times r}}} \norm{O_1} T + \norm{O_2}^2 T^2 + 2\norm{O_2}\sqrt{\norm{O_1}} T^{3/2}, 
\end{equation}
In the multi-probe scenario where $N$ is the number of probes and $t$ is the probing time (see \figbref{fig:strategy}), we can replace $T$ with $N \cdot t$ in the inequality above, and we have 
\begin{equation}
\label{eq:upper-beyond}
F(\rho_\omega(t)) \leq 4 \min_{\substack{h_{00} \in \bR,\,\vh \in \bC^r,\\ \text{~Hermitian~}\frakh \in \bC^{r \times r}}} \norm{O_1} Nt + \norm{O_2}^2 (Nt)^2 + 2\norm{O_2}\sqrt{\norm{O_1}} (Nt)^{3/2}, 
\end{equation}
where $N$ is the number of probes. 

Before we proceed to prove the optimality of our ancilla-free QEC construction, we first explain how the definitions of operators, codes, and subspaces in the main text are extended to the general case here. For the simplicity of notations, below we will drop the explicit $\omega$-dependence in operators, e.g., in $E(\omega)$ and $L_i(\omega)$. We define
\begin{equation}
\label{eq:generalized-H}
H := \partial_\omega E - \frac{i}{2} \sum_{i} \left( (\partial_\omega L_i)^\dagger L_i - L_i^\dagger (\partial_\omega L_i)\right),
\end{equation}
and $\mS$ to be the linear subspace of Hermitian operators spanned by $\id,L_i,L_i^\dagger$ and $L_i^\dagger L_j$ for all $i,j$. 
Note that the above definitions of $H$ and $\mS$ coincide with the ones in the main text under the scenario of Hamiltonian estimation under Markovian noise. We can also extend the definitions of $\rho_0$, $\rho_1$, $\lambda_i$, $\{\ket{i}_P\}_{i=0}^{d-1}$, $\{\ket{i}_A\}_{i=0}^{d-1}$ through \eqref{eq:rho-condition-1}, \eqref{eq:rho-condition-2} and \eqref{eq:code-ancilla}, using our extended definitions of $H$ and $\mS$. The extended definitions of the small-ancilla code and the ancilla-free random code also follow respectively from \eqref{eq:code-few-ancilla} and \eqref{eq:code-ancilla-free}. We can then define $\frakL_k = {\rm span}\{\ket{k_\tl},L_i^{(\ell)}\ket{k_\tl},\forall i,\ell\}$ for $k=0,1$ as before, and the gauge constraints \eqref{eq:gauge-1} and \eqref{eq:gauge-2} are assumed to be satisfied here.

From the upper bound of QFI in \eqref{eq:upper-beyond}, we have the following upper bound on the asymptotic HL coefficient 
\begin{equation}
\label{eq:upper-beyond-HL}
\sup_{t > 0} \lim_{N\rightarrow \infty}  \frac{F(\rho_\omega(t))}{N^2t^2} \leq 4 \norm{H-\mS}^2,
\end{equation}
which is positive only when $H \notin \mS$. 
When $H \in \mS$, the HL is not attainable and the SQL coefficients has the following upper bound
\begin{equation}
\sup_{t > 0} \lim_{N\rightarrow \infty}  \frac{F(\rho_\omega(t))}{N t}  \leq 4 o_1,
\end{equation}
where 
\begin{equation}
\label{eq:upper-beyond-SQL}
    o_1 := \min_{\substack{\vh \in \bC^r,\text{~Hermitian~}\frakh \in \bC^{r \times r}, \text{ s.t.} \\ H - \sum_{j}(\vh_j^* L_j + \vh_j L_j^\dagger) - \sum_{ij} \frakh_{ij}L_i^\dagger L_j \propto \id}} \norm{O_1}.
\end{equation}

Below, we will use results from \appref{app:dephasing} that in the limit of infinitely fast QEC, when the recovery channel can be written as 
\begin{multline}
    \mR(\cdot) = \sum_p (\ket{0_\tl}\bra{R_p}+\ket{1_\tl}\bra{S_p})(\cdot) (\ket{R_p}\bra{0_\tl}+\ket{S_p}\bra{1_\tl}) \\
    +  \trace((\cdot)(\Pi_0-\tPi_R)) \ket{0_\tl}\bra{0_\tl}
    +  \trace((\cdot)(\id - \Pi_0 - \tPi_S)) \ket{1_\tl}\bra{1_\tl} ,
\end{multline}
where $\{\ket{R_p},\forall p\}$, $\{\ket{S_p},\forall p\}$ are two sets of vectors such that $\Pi_0 - \tPi_{R}$ and $\Pi_1 - \tPi_{S}$ are positive semidefinite (We use $\Pi_{0,1}$ to denote projections onto $\frakL_{0,1}$ and we define $\tPi_{R} := \sum_p \ket{R_p}\bra{R_p}$ and $\tPi_{S} := \sum_p \ket{S_p}\bra{S_p}$.), the logical master equation is 
\begin{equation}
\label{eq:logical-master-beyond}
\frac{d\rho_\tl}{dt} = -i\left[\frac{\trace(\sum_\ell E^{(\ell)} Z_\tl)}{2} Z_\tl + \frac{\beta_\tl(\mR)}{2} Z_\tl,\rho_\tl\right] + \frac{\gamma_\tl(\mR)}{2}(Z_\tl \rho_\tl Z_\tl - \rho_\tl),
\end{equation}
where $Z_\tl = \ket{0_\tl}\bra{0_\tl} - \ket{1_\tl}\bra{1_\tl}$, 
\begin{equation}
\gamma_\tl(\mR) = -\frac{1}{2}(x+x^*),\quad 
\beta_\tl(\mR) = -\frac{1}{2i}(x-x^*),
\end{equation}
and 
\begin{multline}
x= \sum_{i=1}^r \sum_{\ell=1}^m \bra{0_\tl} L_i^{(\ell)} \ket{0_\tl} \bra{1_\tl} L_i^{(\ell)\dagger} \ket{1_\tl}- \frac{1}{2}(\bra{0_\tl} L_i^{(\ell)\dagger} L_i^{(\ell)} \ket{0_\tl} + \bra{1_\tl} L_i^{(\ell)\dagger} L_i^{(\ell)} \ket{1_\tl}) \\ + \sum_{p}\bra{R_p}(\id - \ket{0_\tl}\bra{0_\tl})(L_i^{(\ell)} \ket{0_\tl}\bra{1_\tl} L_i^{(\ell)\dagger} )(\id - \ket{1_\tl}\bra{1_\tl})\ket{S_p}. 
\end{multline}
Unlike the case of Hamiltonian estimation where $\beta_\tl(\mR)$ is independent of $\omega$ and can be ignored, here $\beta_\tl(\mR)$ is a function of $\omega$ and should be taken into consideration. 

We will show that when $H\notin \mS$, for both the small-ancilla code and the ancilla-free random code, there exists an input state and a recovery channel that together achieve the optimal HL. 

The obstacle in proving a similar result when $H \in \mS$ is that the optimal recovery channel is not necessarily the one that minimizes the noise rate. Therefore, some techniques used before in calculating the trace norm of operators are no longer useful. Therefore, we only restrict our discussion to the $H \notin \mS$ case in this appendix.  

\subsection{Code optimality when HNLS is satisfied}

Assume $H \notin \mS$. We will show that the small-ancilla code achieves the upper bound in \eqref{eq:upper-beyond-HL} and the ancilla-free random code achieves it with high probability. Consider an input state 
\begin{equation}
\ket{\psi_{\rm in}} = \frac{\ket{0_\tl} + \ket{1_\tl}}{\sqrt{2}},  
\end{equation}
where $\{\ket{0_\tl},\ket{1_\tl}\}$ is a QEC code of $m = N$ probes. 
Then the output state is 
\begin{equation}
\rho_\omega(t) = 
\begin{pmatrix}
\frac{1}{2} & \frac{1}{2}e^{-i \left(\trace(\sum_\ell E^{(\ell)} Z_\tl) + \beta_\tl(\mR)\right) t - \gamma_\tl(\mR) t} \\ 
\frac{1}{2}e^{i \left(\trace(\sum_\ell E^{(\ell)} Z_\tl) + \beta_\tl(\mR)\right) t - \gamma_\tl(\mR) t} & \frac{1}{2}
\end{pmatrix}, 
\end{equation}
and its QFI can be calculated as 
\begin{equation}
F(\rho_\omega(t)) = t^2 e^{- 2\gamma_\tl(\mR) t} \left(\trace\left(\sum_\ell \partial_\omega E^{(\ell)} Z_\tl\right) + \partial_\omega\beta_\tl(\mR)\right)^2 + \frac{(- t \partial_\omega \gamma_\tl (\mR)e^{- \gamma_\tl(\mR) t})^2}{1 - e^{- 2\gamma_\tl(\mR) t}}. 
\end{equation}
Our goal is to prove that there exists a recovery channel $\mR$ such that 
\begin{equation}
\label{eq:beyond-goal}
    \gamma_\tl(\mR) = O(1),\quad \text{and} \quad 
    \trace\left(\sum_\ell \partial_\omega E^{(\ell)} Z_\tl\right) + \partial_\omega\beta_\tl(\mR) = 2 m \norm{H - \mS} + o(m). 
\end{equation}
which then indicates that 
\begin{equation}
\sup_{t > 0} \lim_{N\rightarrow \infty}  \frac{F(\rho_\omega(t))}{N^2t^2} = \sup_{t > 0} \lim_{m\rightarrow \infty}  \frac{F(\rho_\omega(t))}{m^2t^2} \geq 4 \norm{H - \mS}^2,
\end{equation}
i.e., the QFI upper bound is attained. 
\begin{theorem}[Achieving the optimal HL beyond Hamiltonian estimation]
    Consider a quantum system evolution described by \eqref{eq:beyond-lindblad}. The optimal HL, i.e., 
    \begin{equation}
\sup_{t > 0} \lim_{N\rightarrow \infty}  \frac{F(\rho_\omega(t))}{N^2t^2} = 4 \norm{H - \mS}^2,
\end{equation}
is achievable using the small-ancilla code, or the ancilla-free random code with probability $1 - e^{-\Omega(m)}$. (Notations and codes were defined in \appref{app:notation-beyond}.)
\end{theorem}

\begin{proof}
We first show \eqref{eq:beyond-goal} for the small-ancilla code. We follow the same definitions in \appref{app:limit-1} where we defined for $i = 1,\ldots,r$ and $\ell = 1,\ldots, m-1$, 
\begin{equation}
\ket{e_{i,\ell}} = \frac{1}{\sqrt{m}}\sum_{\ell'=1}^{m} \exp\left(-i \frac{2\pi}{m}\ell \ell'\right) \ket{\widehat{J}_{0,i}^{(\ell')}},\quad 
\ket{f_{i,\ell}} = \frac{1}{\sqrt{m}}\sum_{\ell'=1}^{m} \exp\left(-i \frac{2\pi}{m}\ell \ell'\right) \ket{\widehat{J}_{1,i}^{(\ell')}},
\end{equation}
where 
\begin{equation}
\ket{\widehat{J}_{k,i}^{(\ell)}} = \frac{1}{\sqrt{\mu_i}}(L_i^{(\ell)} - b_{k,i})\ket{k_\tl}\text{~~and~~}b_{k,i} = \braket{k_\tl | L_i^{(\ell)}| k_\tl}, 
\end{equation}
for $i = 1,\ldots,r$, $\ell = 1,\ldots, m$ and $k = 0,1$. Then we define the following recovery channel $\mR_{\rm opt}$ (as in \eqref{eq:optimal-recovery})
\begin{multline}
    \mR_{\rm opt}(\cdot) = \sum_{i,\ell} (\ket{0_\tl}\bra{R_{(i,\ell)}}+\ket{1_\tl}\bra{S_{(i,\ell)}})(\cdot) (\ket{R_{(i,\ell)}}\bra{0_\tl}+\ket{S_{(i,\ell)}}\bra{1_\tl}) 
    \\ +  \trace((\cdot)(\Pi_0-\tPi_R)) \ket{0_\tl}\bra{0_\tl}
    +  \trace((\cdot)(\id - \Pi_0 - \tPi_S)) \ket{1_\tl}\bra{1_\tl} ,
\end{multline}
where 
\begin{equation}
\ket{R_{(i,\ell)}} = \frac{\ket{e_{i,\ell}}}{\sqrt{\norm{\sum_{j,k}\ket{e_{j,k}}\bra{e_{j,k}}}_\infty}},\quad 
\ket{S_{(i,\ell)}} = \frac{\ket{f_{i,\ell}}}{\sqrt{\norm{\sum_{j,k}\ket{f_{j,k}}\bra{f_{j,k}}}_\infty}},
\end{equation}
and we use $\Pi_{0,1}$ to denote projections onto $\frakL_{0,1}$ and $\tPi_{R} = \sum_{i,\ell} \ket{R_{(i,\ell)}}\bra{R_{(i,\ell)}}$ and $\tPi_{S} = \sum_{i,\ell} \ket{S_{(i,\ell)}}\bra{S_{(i,\ell)}}$. As we have shown in \appref{app:limit-1}, $\gamma_\tl(\mR_{\rm opt}) = O(1)$, and we only need to show 
\begin{equation}
\label{eq:beyond-signal}
    \trace\left(\sum_\ell \partial_\omega E^{(\ell)} Z_\tl\right) + \partial_\omega\beta_\tl(\mR_{\rm opt}) = 2 m \norm{H - \mS} + o(m).
\end{equation}
First, we have 
\begin{equation}
    \trace\left(\sum_\ell \partial_\omega E^{(\ell)} Z_\tl\right) = m \trace\left(\partial_\omega E \sum_{i=0}^{d_0-1} \frac{m_{i}}{m} \ket{i}\bra{i} - \sum_{i=d_0}^{d-1} \frac{m_{i}}{m} \ket{i}\bra{i}\right) = m \trace\left(\partial_\omega E (\rho_0 - \rho_1)\right) + O(1). \label{eq:beyond-hamt}
\end{equation}
Second, we have 
\begin{align}
    &\quad \; \partial_\omega\beta_\tl(\mR_{\rm opt}) \\
    & =  \sum_{i=1}^r \sum_{\ell=1}^m - \Im[\bra{0_\tl} \partial_\omega L_i^{(\ell)} \ket{0_\tl} \bra{1_\tl} L_i^{(\ell)\dagger} \ket{1_\tl}] - \Im[\bra{0_\tl} L_i^{(\ell)} \ket{0_\tl} \bra{1_\tl} \partial_\omega  L_i^{(\ell)\dagger} \ket{1_\tl}] \\& \quad -  \sum_{i,i'=1}^{r}\sum_{\ell=1}^{m-1} \sum_{\ell'=1}^m \Im[\bra{R_{(i',\ell')}}\partial_\omega  L_i^{(\ell)} \ket{0_\tl}\bra{1_\tl} L_i^{(\ell)\dagger} \ket{S_{(i',\ell')}}] + \Im[\bra{R_{(i',\ell')}} L_i^{(\ell)} \ket{0_\tl}\bra{1_\tl} \partial_\omega  L_i^{(\ell)\dagger} \ket{S_{(i',\ell')}}]\\
     & = O(1) -  \sum_{i,i'=1}^{r}\sum_{\ell=1}^{m-1} \sum_{\ell'=1}^m  \Im[\bra{R_{(i',\ell')}}\partial_\omega  L_i^{(\ell)} \ket{0_\tl}\bra{1_\tl} L_i^{(\ell)\dagger} \ket{S_{(i
    ',\ell')}}] +  \Im[\bra{R_{(i',\ell')}} L_i^{(\ell)} \ket{0_\tl}\bra{1_\tl} \partial_\omega  L_i^{(\ell)\dagger} \ket{S_{(i',\ell')}}]\label{eq:beta-temp},
\end{align}
where we use $\braket{R_{(i,\ell)}|0_\tl} = \braket{S_{(i,\ell)}|1_\tl} = 0$ and $\braket{k_\tl|L_i^{(\ell)}|k_\tl} = O(1/m)$ for $k = 0,1$. 
Furthermore, a few lines of calculation show that 
\begin{gather}
\braket{0_\tl|L_{i}^{(\ell)\dagger}|e_{i',\ell'}} 
= \frac{1}{\sqrt{m}}\exp\left(-i \frac{2\pi}{m}{\ell\ell'}\right)  \left(\sqrt{\mu_{i}}\delta_{ii'} + O\left(\frac{1}{m}\right)\right),\label{eq:beyond-condition-1}
\\
\braket{1_\tl|L_{i}^{(\ell)\dagger}|f_{i',\ell'}} 
= \frac{1}{\sqrt{m}}\exp\left(-i \frac{2\pi}{m}{\ell\ell'}\right)  \left(\sqrt{\mu_{i}}\delta_{ii'} + O\left(\frac{1}{m}\right)\right),\label{eq:beyond-condition-2}\\
\braket{0_\tl|\partial_\omega L_{i}^{(\ell)\dagger}|e_{i',\ell'}} 
= \frac{1}{\sqrt{m}} \exp\left(-i \frac{2\pi}{m}{\ell \ell'}\right)  \left( \frac{1}{\sqrt{\mu_{i'}}} 
\trace(\rho_0 \partial_\omega L_{i}^{\dagger }L_{i'}) + O\left(\frac{1}{m}\right)\right) ,\label{eq:beyond-condition-3}\\ 
\braket{1_\tl|\partial_\omega  L_{i}^{(\ell)\dagger}|f_{i',\ell'}} 
= \frac{1}{\sqrt{m}} \exp\left(-i \frac{2\pi}{m}{\ell \ell'}\right) \left( \frac{1}{\sqrt{\mu_{i'}}} \trace(\rho_1 \partial_\omega L_{i}^{\dagger }L_{i'}) + O\left(\frac{1}{m}\right) \right).\label{eq:beyond-condition-4} 
\end{gather}
The $O(1/m)$ terms above are independent of $\ell$ and $\ell'$. Additionally, from \eqref{eq:ee-matrix-upper} and \eqref{eq:ff-matrix-upper} 
\begin{gather}
\bigg\|\sum_{j,k}\ket{e_{j,k}}\bra{e_{j,k}}\bigg\|_\infty = 1 + O\left(\frac{1}{m}\right),\label{eq:ee-matrix}\\
\bigg\|\sum_{j,k}\ket{f_{j,k}}\bra{f_{j,k}}\bigg\|_\infty = 1 + O\left(\frac{1}{m}\right),\label{eq:ff-matrix}
\end{gather}
where the norms of the two operators are lower bounded by $1+O(1/m)$ can be shown by noting that the expectation values of them on $\ket{e_{j,k}}$ and $\ket{f_{j,k}}$ respectively are lower bounded by $1+O(1/m)$. 
Then we have, from \eqref{eq:beta-temp}, \eqref{eq:ee-matrix} and \eqref{eq:ff-matrix}, that 
\begin{align}
\partial_\omega \beta_\tl(\mR_{\rm opt}) &= O(1) + m \sum_{i=1}^r \Im[\trace(\rho_0 L_i^\dagger \partial_\omega L_i)] + \Im[\trace(\rho_1 \partial_\omega L_i^\dagger  L_i)] \\&= m \frac{i}{2}\trace((\rho_0-\rho_1) L_i^\dagger \partial_\omega L_i - \partial_\omega L_i^\dagger L_i) + O(1). \label{eq:beyond-lindblad}
\end{align}
Combining \eqref{eq:beyond-hamt} and \eqref{eq:beyond-lindblad}, \eqref{eq:beyond-signal} is then proven.

We have shown above the optimality of the small-ancilla code. For the ancilla-free random code, we adopt the same definitions of $\ket{e_{i,\ell}}$, $\ket{f_{i,\ell}}$, $\ket{\widehat{J}_{k,i}^{(\ell)}}$, $\mR_{\rm opt}$, etc., as in \appref{app:limit-2}. Then the proof of optimality follows identically from the proof above, except that \eqref{eq:beyond-condition-1}--\eqref{eq:beyond-condition-4} should be replaced with 
\begin{gather}
\braket{0_\tl|L_{i}^{(\ell)\dagger}|e_{i',\ell'}} 
= \frac{1}{\sqrt{m}}\exp\left(-i \frac{2\pi}{m}{\ell\ell'}\right)  \left(\sqrt{\mu_{i}}\delta_{ii'} + O\left(\frac{1}{m}\right) + O\left(\frac{1}{m^2}\right)\right),
\\
\braket{1_\tl|L_{i}^{(\ell)\dagger}|f_{i',\ell'}} 
= \frac{1}{\sqrt{m}}\exp\left(-i \frac{2\pi}{m}{\ell\ell'}\right)  \left(\sqrt{\mu_{i}}\delta_{ii'} + O\left(\frac{1}{m}\right)+ O\left(\frac{1}{m^2}\right)\right),\\
\braket{0_\tl|\partial_\omega L_{i}^{(\ell)\dagger}|e_{i',\ell'}} 
= \frac{1}{\sqrt{m}} \exp\left(-i \frac{2\pi}{m}{\ell \ell'}\right)  \left( \frac{1}{\sqrt{\mu_{i'}}} 
\trace(\rho_0 \partial_\omega L_{i}^{\dagger }L_{i'}) + O\left(\frac{1}{m}\right)+ O\left(\frac{1}{m^2}\right)\right) ,\\ 
\braket{1_\tl|\partial_\omega  L_{i}^{(\ell)\dagger}|f_{i',\ell'}} 
= \frac{1}{\sqrt{m}} \exp\left(-i \frac{2\pi}{m}{\ell \ell'}\right) \left( \frac{1}{\sqrt{\mu_{i'}}} \trace(\rho_1 \partial_\omega L_{i}^{\dagger }L_{i'}) + O\left(\frac{1}{m}\right) + O\left(\frac{1}{m^2}\right)\right).
\end{gather}
where the $O(1/m)$ terms are independent of $\ell,\ell'$ and the $O(1/m^2)$ terms may be functions of $\ell,\ell'$. The above relations are true with probability $1 - e^{-\Omega(m)}$ and can be derived using \lemmaref{lemma:Delta}.  
\end{proof}

}

\bibliography{refs-multi-probe-0}

\begin{thebibliography}{66}
\bibitem [{1fo()}]{1footnote-heisenberg}%
  \BibitemOpen
  \href@noop {} {}\bibinfo {note} {The case $F(\rho_\omega) = \Theta(t^2)$
  where $t$ is the probing time is sometimes also referred to as the HL (see
  \appref{app:heisenberg}), but we will focus on only the scaling with respect
  to $N$ (the number of probes) and fix the value of $t$ in this
  work.}\BibitemShut {Stop}%
\bibitem [{2fo()}]{2footnote-GHZ}%
  \BibitemOpen
  \href@noop {} {}\bibinfo {note} {We abuse the name ``GHZ state'' a bit, by
  calling any state of the form $(\ket{0}^{\otimes N} + \ket{1}^{\otimes
  N})/\sqrt{2}$ a GHZ state, for all $N \geq 1$.}\BibitemShut {Stop}%
\end{thebibliography}

 
\end{document}